\documentclass[a4paper,11pt,onecolumn]{article}
\usepackage[utf8]{inputenc}

\usepackage{amsthm,amsmath,amssymb,bbm,mathtools}

\usepackage{microtype}
\usepackage{todonotes}

\usepackage{enumerate}

\usepackage{authblk}

\usepackage[backend=bibtex,style=alphabetic,doi=false,isbn=false,url=false,maxbibnames=5,sorting=nty,sortcites=false]{biblatex}

\usepackage[colorlinks=true]{hyperref}

\input{braket}
\newcommand{\dx}{0.02}
\newcommand{\ts}{.75cm}
\newcounter{tx}
\newcounter{ty}

\usepackage{xparse}
\usepackage{tikz}
\usetikzlibrary{patterns, decorations.pathreplacing,arrows.meta}

\makeatletter
\tikzset{
    hatch distance/.store in=\hatchdistance,
    hatch distance=4pt,
    hatch thickness/.store in=\hatchthickness,
    hatch thickness=2pt
}
\pgfdeclarepatternformonly[\hatchdistance,\hatchthickness]{stripey}
    {\pgfqpoint{-1pt}{-1pt}}
    {\pgfqpoint{\hatchthickness}{10pt}}
    {\pgfqpoint{\hatchdistance}{10pt}}
    {
        \pgfsetcolor{\tikz@pattern@color}
        \pgfsetlinewidth{\hatchthickness}
        \pgfpathmoveto{\pgfqpoint{0pt}{-1pt}}
        \pgfpathlineto{\pgfqpoint{0pt}{10pt}}
        \pgfusepath{stroke}
    }
\makeatother

\DeclareDocumentEnvironment{tiles}{ s O{1} }{%
  \setcounter{tx}{0}%
  \setcounter{ty}{0}%
  \IfBooleanTF{#1}{%
    \begin{tikzpicture}[
      draw=black,
      font=\scriptsize,
      line width=\dx*\ts,
      x=#2*\ts,y=#2*\ts,
      baseline={([yshift=-.6ex]current bounding box.center)}
    ]
  }{%
    \begin{tikzpicture}[
    draw=black,
    font=\scriptsize,
    line width=\dx*\ts,
    x=#2*\ts,y=#2*\ts
    ]
  }%
}{
  \end{tikzpicture}%
}

\renewcommand\${
\setcounter{tx}{0}
\addtocounter{ty}{-1}
}
\renewcommand\~{
\addtocounter{tx}{1}
}
\DeclareDocumentCommand{\tile}{m o o}{
  \addtocounter{tx}{1}
  \addtocounter{ty}{0}
  \begin{scope}[
    shift={(\thetx,\thety)}
  ]
  #1
  \draw[fill=\IfValueTF{#2}{colour#2}{white},opacity=1] (.1,.1) rectangle (.9,.9);
  \IfValueT{#3}{
  		\node[] at (.5,.472) {#3};
  }
  \end{scope}
}

\DeclareDocumentCommand{\defineEdge}{ m m m m m }{
	\DeclareDocumentCommand{#1}{ s m } {
		\draw[fill=colour##2,colour##2,rounded corners=1] (#2) rectangle (#3);
		\IfBooleanT{##1}{\draw[fill=white,white,line width=2*\dx*\ts] (#4) rectangle (#5);}
	}
}
\defineEdge\edgeN{.3,.8}{.7,1}{.45,.8}{.55,1}
\defineEdge\edgeE{.8,.3}{1,.7}{.8,.45}{1,.55}
\defineEdge\edgeS{.3,0}{.7,.2}{.45,0}{.55,.2}
\defineEdge\edgeW{0,.3}{.2,.7}{0,.45}{.2,.55}

\DeclareDocumentCommand{\Tomg}{m m m m o o}{%
	\tile{\edgeN#1\edgeE#2\edgeS#3\edgeW#4}[#5][#6]
}
\def\T(#1,#2,#3,#4) {\Tomg{#1}{#2}{#3}{#4}}
\def\TT(#1,#2,#3,#4,#5) {\Tomg{#1}{#2}{#3}{#4}[#5]}

\def\Ta(#1,#2,#3,#4,#5) {\Tomg{#1}{#2}{#3}{#4}[][#5]}
\def\TTa(#1,#2,#3,#4,#5,#6) {\Tomg{#1}{#2}{#3}{#4}[#5][#6]}

\def\Ts(#1,#2,#3,#4){\begin{tiles}*%
\Tomg{#1}{#2}{#3}{#4}%
\end{tiles}}
\def\TTs(#1,#2,#3,#4,#5){\begin{tiles}*%
\Tomg{#1}{#2}{#3}{#4}[#5]%
\end{tiles}}

\definecolor{colour}{RGB}{255,255,255}
\definecolor{colour0}{RGB}{0,0,0}
\definecolor{colour1}{RGB}{155,155,155}
\definecolor{colourR}{RGB}{255,34,44}
\definecolor{colourB}{RGB}{11,92,255}
\definecolor{colourD}{RGB}{200,0,100}

\definecolor{colourX}{RGB}{32,32,16}

\definecolor{colourT}{RGB}{192,192,192}
\definecolor{colourS}{RGB}{235,235,235}

\newcommand{\mysymbol}[1]{\mbox{\raisebox{-0.3em}{\includegraphics[height=14pt]{images/#1}}}}

\newcommand{\leftend}{\mysymbol{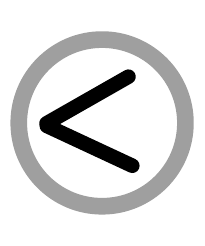}}
\newcommand{\rightend}{\mysymbol{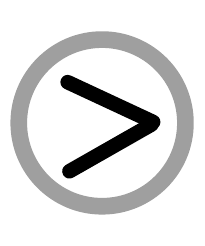}}
\newcommand{\midend}{\mysymbol{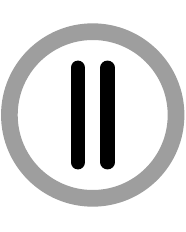}}

\newcommand{\arrLone}{\mysymbol{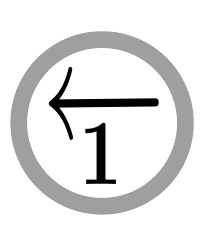}}

\RequirePackage{xparse}

\DeclareDocumentCommand\bra{ s m t\ket s g }
{ 
	\IfBooleanTF{#3}
	{ 
		\IfBooleanTF{#1}
		{ 
			\IfNoValueTF{#5}
			{\braket*{#2}{} \IfBooleanTF{#4}{*}{}}
			{\braket*{#2}{#5}}
		}
		{
			\IfBooleanTF{#4}
			{ 
				\IfNoValueTF{#5}
				{\braket{#2}{} *}
				{\braket*{#2}{#5}}
			}
			{\braket{#2}{\IfNoValueTF{#5}{}{#5}}} 
		}
	}
	{ 
		\IfBooleanTF{#1}
		{\vphantom{#2}\left\langle\smash{#2}\right\rvert}
		{\left\langle{#2}\right\rvert}
		\IfBooleanTF{#4}{*}{}
		\IfNoValueTF{#5}{}{#5}
	}
}

\DeclareDocumentCommand\ket{ s m }
{ 
	\IfBooleanTF{#1}
	{\vphantom{#2}\left\lvert\smash{#2}\right\rangle} 
	{\left\lvert{#2}\right\rangle} 
}

\DeclareDocumentCommand\innerproduct{ s m g }
{ 
	\IfBooleanTF{#1}
	{ 
		\IfNoValueTF{#3}
		{\vphantom{#2}\left\langle\smash{#2}\middle\vert\smash{#2}\right\rangle}
		{\vphantom{#2#3}\left\langle\smash{#2}\middle\vert\smash{#3}\right\rangle}
	}
	{ 
		\IfNoValueTF{#3}
		{\left\langle{#2}\middle\vert{#2}\right\rangle}
		{\left\langle{#2}\middle\vert{#3}\right\rangle}
	}
}
\DeclareDocumentCommand\braket{}{\innerproduct} 

\DeclareDocumentCommand\outerproduct{ s m g }
{ 
	\IfBooleanTF{#1}
	{ 
		\IfNoValueTF{#3}
		{\vphantom{#2}\left\lvert\smash{#2}\middle\rangle\!\middle\langle\smash{#2}\right\rvert}
		{\vphantom{#2#3}\left\lvert\smash{#2}\middle\rangle\!\middle\langle\smash{#3}\right\rvert}
	}
	{ 
		\IfNoValueTF{#3}
		{\left\lvert{#2}\middle\rangle\!\middle\langle{#2}\right\rvert}
		{\left\lvert{#2}\middle\rangle\!\middle\langle{#3}\right\rvert}
	}
}
\DeclareDocumentCommand\ketbra{}{\outerproduct} 

\usepackage{cleveref}
\Crefname{lemma}{Lemma}{Lemmas}
\Crefname{proposition}{Proposition}{Propositions}
\Crefname{definition}{Definition}{Definitions}
\Crefname{theorem}{Theorem}{Theorems}
\Crefname{conjecture}{Conjecture}{Conjectures}
\Crefname{corollary}{Corollary}{Corollaries}
\Crefname{example}{Example}{Examples}
\Crefname{section}{Section}{Sections}
\Crefname{appendix}{Appendix}{Appendices}
\Crefname{figure}{Fig.}{Figs.}
\Crefname{equation}{Eq.}{Eqs.}
\Crefname{table}{Table}{Tables}
\Crefname{item}{Property}{Properties}
\Crefname{remark}{Remark}{Remarks}

\newtheorem{theorem}{Theorem}[section]
\newtheorem{lemma}[theorem]{Lemma}
\newtheorem{proposition}[theorem]{Proposition}

\newtheorem{corollary}[theorem]{Corollary}
\newtheorem{definition}[theorem]{Definition}

\newtheorem{remark}[theorem]{Remark}

\newcommand\field\mathbb
\newcommand\1{\mathbbm{1}}
\newcommand\op\relax
\newcommand\mathds\mathbb

\DeclareMathOperator{\C}{\mathbb{C}}
\DeclareMathOperator{\ox}{\otimes}
\DeclareMathOperator{\HS}{\mathcal{H}}
\DeclareMathOperator{\supp}{supp}

\DeclareMathOperator{\poly}{poly}

\DeclareMathOperator{\BigO}{O}
\DeclareMathOperator{\clz}{clz}
\DeclareMathOperator{\pfx}{pfx}
\DeclareMathOperator{\sgn}{sgn}

\newcommand\ee{\mathrm e}
\newcommand\ii{ i}
\newcommand\abs[1]{|#1|}
\newcommand\bl{\raisebox{.1ex}{$\vartriangleright$}}
\newcommand\bd{\raisebox{.0ex}{$\blacksquare$}}
\newcommand\hd{\raisebox{.1ex}{$\blacktriangleright$}}
\newcommand\bdbullet{\ensuremath{\bigstar}}
\newcommand\lmin{\lambda_\mathrm{min}}

\newcommand{\ivar}{\eta}

\newcommand{\Hchecker}{\op H_\mathrm{cb}}
\newcommand{\Hmarker}{\op H^{(f)}}
\newcommand{\HmarkerN}{\op H^{(\boxplus,f)}}
\newcommand{\HmarkerM}{\op H^{(\boxplus)}}
\newcommand{\Htriv}{\op H_\mathrm{trivial}}
\newcommand{\Hdens}{\op H_\mathrm{dense}}
\newcommand{\Hguard}{\op H_\mathrm{guard}}
\newcommand{\Hundec}{\op H^{\Lambda}}
\newcommand{\HundecL}{\op H^{\Lambda(L)}}
\newcommand{\HUTM}{\op H_\mathrm{comp}}
\newcommand{\UTM}{\mathcal M}

\DeclareMathOperator{\bigO}{O}
\DeclareMathOperator{\smallO}{o}
\DeclareMathOperator{\argmin}{argmin}
\DeclareMathOperator{\spec}{spec}

\renewbibmacro{in:}{}
\bibliography{bibliography}
\newbibmacro{string+doi}[1]{\iffieldundef{doi}{#1}{\href{https://dx.doi.org/\thefield{doi}}{#1}}}
\DeclareFieldFormat{title}{\usebibmacro{string+doi}{\mkbibemph{#1}}}
\DeclareFieldFormat[article]{title}{\usebibmacro{string+doi}{\mkbibquote{#1}}}
\DeclareFieldFormat[incollection]{title}{\usebibmacro{string+doi}{\mkbibquote{#1}}}
\DeclareFieldFormat[inproceedings]{title}{\usebibmacro{string+doi}{\mkbibquote{#1}}}
\renewbibmacro{in:}{}

\begin{document}
	\title{Uncomputability of Phase Diagrams}
	\author[1]{Johannes Bausch\thanks{\texttt{jkrb2@cam.ac.uk}}}
	\affil[1]{CQIF, DAMTP, University of Cambridge, UK}
	\author[2]{Toby S. Cubitt\thanks{\texttt{t.cubitt@ucl.ac.uk}}}
	\affil[2]{Department of Computer Science, University College London, UK}
	\author[2]{James D. Watson\thanks{\texttt{ucapjdj@ucl.ac.uk}}}
	\maketitle

	\begin{abstract}
		The phase diagram of a material is of central importance in describing the properties and behaviour of a condensed matter system.
		In this work, we prove that the task of determining the phase diagram of a many-body Hamiltonian is in general uncomputable, by explicitly constructing a continuous one-parameter family of Hamiltonians $\op{ H}(\varphi)$, where $\varphi\in \mathbb{ R}$, for which this is the case.
		The $\op H(\varphi)$ are translationally-invariant, with nearest-neighbour couplings on a 2D spin lattice.
		As well as implying uncomputablity of phase diagrams, our result also proves that undecidability can hold for a set of positive measure of a Hamiltonian's parameter space, whereas previous results only implied undecidability on a zero measure set.
		This brings the spectral gap undecidability results a step closer to standard condensed matter problems, where one typically studies phase diagrams of many-body models as a function of one or more continuously varying real parameters, such as magnetic field strength or pressure.
	\end{abstract}

	\renewcommand\include\input

\section{Introduction}

Phase transitions and phase diagrams have been a central area of study in condensed matter physics for well over a century.
In particular in the second half of the 20\textsuperscript{th} century, interest in superconductors and topological phases spurred work on quantum phase transitions: a discontinuous change of a macroscopic observable happening at zero temperature due to the change in some non-thermal parameter \cite{Sachdev_2011}.

The phase diagrams for many materials have been well-studied both experimentally and theoretically.
There exist numerous algorithms which are heuristically effective at computing properties of many-body quantum systems, such as the Density Matrix Renormalization Group for 1D gapped systems or density functional theory \cite{White_1992,Jones_2015}.
Classic toy models include the 1D transverse field Ising Model which is known to have a transition from an unordered to ordered phase at a critical magnetic field strength \cite{Sachdev_2011}.
Beyond that, materials with exotic phases, such as topological insulators, or the fractional quantum Hall effect, are becoming increasingly important to understand as they become more applicable to real world applications \cite{Pudalov_Semenchinski_1988, Vobornik_et_al_2011}.

Yet the quantum phase diagram for such systems can be highly complex. Numerical simulations of quantum systems are computationally difficult, and may even be intractable \cite{Staar_Maier_Schulthess, Schuch_Verstrete}.
Experimentally and computationally one of the best studied is the 2D electron gas---a model for free electrons in semiconductors---which is well-known to exhibit a complex phase behaviour: the system undergoes a large number of phase transitions, most notably those associated with the quantum Hall effect.
Indeed, the phase diagrams of such systems are known to be incredibly rich with some producing Hoftstadter butterfly patterns with an infinite number of phases \cite{Osadchy_Avron}.
All of these are important instances of the general problem of computing the phase diagram of a Hamiltonian, which classifies the system's state with respect to a macroscopic observable (such as global magnetization), and with respect to a parameter of the Hamiltonian (such as a transverse field strength).

Quantum phase transitions are associated with the spectral gap of the Hamiltonian closing.
More precisely, a non-analytic change in the ground state energy is a necessary (though not always sufficient) condition for a phase transition to occur, and a closing spectral gap is necessary (though not always sufficent) for a non-analytic change in the ground state energy to occur.
Cubitt, Perez-Garcia, and Wolf \cite{Cubitt_PG_Wolf_Undecidability,Cubitt_PG_Wolf_Nature} showed that given a (finite) description of a translationally invariant, nearest neighbour Hamiltonian on a 2D square lattice, deciding whether it has a spectral gap or not is at least as hard as solving the \textsc{Halting Problem}.
This was subsequently extended to the case of 1D Hamiltonians \cite{Bausch_1D_Undecidable}.

In this work we prove that no general algorithm for determining the phase diagram of a system can exist, even given complete knowledge of the microscopic description the system's interactions.
To show this, we explicitly construct a continuous, one-parameter Hamiltonian $\op H(\varphi)$ on 2D lattice with a fixed, finite-dimensional local Hilbert space $\mathcal H_A \oplus \mathcal H_B$, for which determining whether the low energy subspace below some energy cutoff is supported entirely on the $A$ or $B$ subspace is undecidable (where it is guaranteed that one of the two cases holds on a set of positive measure in the parameter space of the model).
With respect to the parameter $\varphi$, the phase diagram determined with respect to a macroscopic observable $\op O_{A/B}$ that measures support on $\mathcal H_A$ vs.\ $\mathcal H_B$ is thus uncomputable.
This observable can also be restricted to a single lattice site, with the same conclusion.

\newcommand{\hcol}{\op h^\mathrm{col}}
\newcommand{\hrow}{\op h^\mathrm{row}}
\section{Results}

The quantum many-body systems we will consider are translationally invariant, nearest-neighbour, 2D spin lattice models.
The $L\times L$ square lattice with open boundary conditions will be denoted $\Lambda(L)$;
for brevity we leave the lattice size implicit whenever it is clear from context.
Each lattice site is associated with a spin system with local Hilbert space of dimension $d$, $\mathbb{C}^d$.
The spins are coupled with a nearest neighbour, translationally invariant Hamiltonian with local terms $\hcol, \hrow \in \mathcal{B}(\mathbb{C}^d\otimes \field{C}^d)$, such that $\max\{||\hrow||, ||\hcol||\}\leq 2$.
Since we are interested in phase transitions---identified by a discontinuous change of a macroscopic observable $\op O_{A/B}$, which strictly speaking can only occur in the thermodynamic limit of infinitely large lattices---we will take the thermodynamic limit by letting $L\to\infty$.
An alternative definition of a quantum phase transition is a non-analytic change in the ground state energy~\cite{Sachdev_2011}. This will also be satsified by our construction.
The resulting Hamiltonian over the entire lattice is then
\begin{align}\label{eq:Hundec}
\HundecL \coloneqq \sum_{i=1}^L\sum_{j=1}^{L-1}\hrow_{(i,j),(i+1,j)} + \sum_{i=1}^{L-1}\sum_{j=1}^{L}\hcol_{(i,j),(i,j+1)}.
\end{align}

As well as being distinguished by the observable $\op O_{A/B}$, the two phases are also distinguished by the spectral gap of the Hamiltonian $\Hundec$, defined as the difference between the smallest and second smallest eigenvalue of the Hamiltonian:
\begin{align}\label{eq:spectral-gap}
\Delta(\HundecL) \coloneqq \lambda_1(\HundecL)  - \lmin(\HundecL).
\end{align}
As in \cite{Cubitt_PG_Wolf_Undecidability}, we then define a Hamiltonian to be \emph{gapped} if there exist $\gamma$ and $L_0$ such that the spectral gap $\Delta(\HundecL)\geq\gamma$ for all $L>L_0$; and \emph{gapless} if the spectrum above the ground state becomes dense in an interval $[\lmin(\HundecL),\lmin(\HundecL)+c]$ for some $c>0$ in the thermodynamic limit (see definitions \ref{def:gapped} and \ref{def:gapless} for mathematically rigorous statements).
Throughout the paper we will be using the notion of a \emph{continuous family of Hamiltonians}, which---loosely speaking---is a family of Hamiltonians $\{ \op H_i(\varphi) \}_{i\in I}$ such that each $\op H_i(\varphi) = \sum_j \op h_j(\varphi)$, and the matrix elements of $h_j(\varphi)$ depend continuously on $\varphi$ (see \cref{def:continuous}).

Our main result is an explicit construction of a one-parameter continuous family of Hamiltonians, such that for all values $\varphi\in\field{R}$ of the external parameter, the system is guaranteed to be in one of two possible phases, distinguished by an order parameter given by the ground state expectation value of a translationally-invariant macroscopic observable $\op O_{A/B}$. However, determining which phase the system is in is undecidable, hence the phase diagram of the system as a function of $\varphi$ is uncomputable.
More precisely, we prove the following theorem:
\begin{theorem}[Phase Diagram Uncomputability]\label{th:main}
	For any given Turing Machine TM, we can construct explicitly a dimension $d\in\field N$, $d^2\times d^2$ matrices $\op a,\op a',\op b,\op c,\op c'$ and a $d\times d$ matrix $\op m$ with the following properties:
	\begin{enumerate}[(i)]
		\item $\op a,\op c$ and $\op m$ are diagonal with entries in $\field{Z}$
		\item $\op a'$ is Hermitian with entries in $ \field{Z}+ \frac{1}{\sqrt{2}}\field{Z}$,
		\item $\op b$ has integer entries.
		\item $\op c'$ is Hermitian with entries in $\field{Z}$.
		\item
		For any real number $\varphi\in \field{R}$ and any $0\le\beta\le 1$, which can be chosen arbitrarily small, setting
		\begin{align*}
		\hcol& \coloneqq \op c +\beta \op c'
		\quad\quad\text{independent of $\varphi$},\\
		\hrow(\varphi)& \coloneqq \op a + \beta\left( \op a'+\ee^{\ii\pi\varphi} \op b + \ee^{-\ii\pi\varphi} \op b^\dagger \right),
		\end{align*}
		we have $\|\hrow(\varphi)\| \leq 2$, $\|\hcol(\varphi)\| \leq 1$.
	\end{enumerate}
	\noindent
	Define $\HundecL$ as in \cref{eq:Hundec}, and let $\op O_{A/B}\coloneqq L^{-2}\sum_{i\in\Lambda} \op m_i$.
	Then, given $\varphi\in [2^{-\ivar}, 2^{-\ivar}+2^{-\ivar-\ell})$ with $\ivar\in\field N$, the following statements hold:
	\begin{enumerate}
		\item If TM halts on input $\eta$, then for some $\ell\ge 1$, $\Hundec(\varphi)$ is gapless in the sense of \cref{def:gapless}, with a ground state that is critical (i.e.\ with algebraic decay of correlations), and for all eigenstates $\ket{\Psi_B}$ with energy $\bra{\Psi_B}\Hundec(\varphi)\ket{\Psi_B} \le 1$ it holds that $\bra{\Psi_B}\op O_{A/B}\ket{\Psi_B} = 0$.
		\item If TM is non-halting on input $\eta$ and $\ell=1$, then $\Hundec(\varphi)$ is gapped in the sense of \cref{def:gapped}, with a unique, product ground state $\ket{\Psi_A}$ with $\bra{\Psi_A}\op O_{A/B}\ket{\Psi_A} = 1$.
	\end{enumerate}
\end{theorem}
Undecidability of which of the two cases pertains follows immediately from undecidability of the Halting Problem, by choosing TM to be a universal Turing Machine.
For simplicity we will refer to the phases $A$ and $B$ determined by the value for the macroscopic observable $\op O_{A/B}$ as the gapped and gapless phase respectively.

As a consequence of the new Hamiltonian construction in this paper, we also obtain the following result:
\begin{corollary}\label{cor:main'}
	For all $\varphi\in\field [0,1]$, $\Hundec(\varphi)$ is either in a phase with a product ground state and a spectral gap $\geq 1$, or it is in a gapless phase with algebraic decay of correlations, where the two phases are distinguished by the expectation value of a macroscopic observable $\op O_{A/B}$.
	Moreover, there exists a subset $S\subset [0,1]$ with Borel measure $\mu(S)>0$, such that even for computable $\varphi\in S$, determining the  phase that $\Hundec(\varphi)$ is in is uncomputable.
\end{corollary}
\noindent
A less precise but simple interpretation of the above corollary is:
\begin{corollary}[informal]\label{cor:main}
	The phase diagram of $\Hundec(\varphi)$ as a function of its parameter $\varphi$ is uncomputable.
\end{corollary}
\noindent
A set of schematic phase diagrams is shown in \cref{Fig:Phase_Diagram}.

\begin{figure}[t]
	\includegraphics[width=\textwidth]{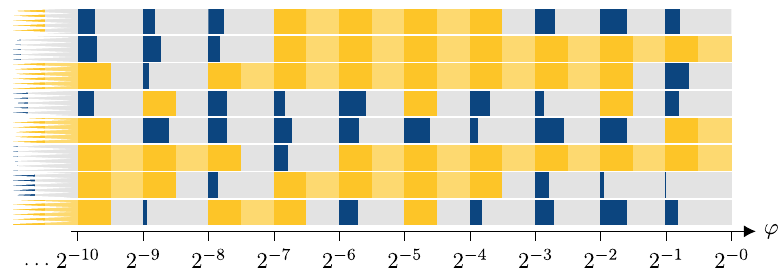}
	\caption{A selection of sample phase diagrams of the continuous family $\{\HundecL(\varphi)\}_{L,\varphi}$ written for a series of possible universal encoded Turing machines varying from top to bottom, plotted against $\varphi$ on the $x$-axis (note the log scaling).
		Blue means gapless (which is where the TM halts asymptotically on input $\varphi$), yellow gapped (TM runs forever).
		At the points $2^{-\ivar}$ for $\ivar\in\field N$ we can have a phase transition between gapped and gapless phases, depending on the behaviour of the encoded TM; there is a positive measure interval above these points where the phase behaviour is consistent.
		The grey sections are parameter ranges which we do not evaluate explicitly; there will be a phase transition at some point within that region if the bounding intervals have different phases.
		The lighter yellow area indicates a changing gapped instance.
		In our construction the gapless behaviour is more intricately-dependent on $\varphi$; but the TM can be chosen such that both halting and non-halting phases cover an order one area of the phase diagram.
	}\label{Fig:Phase_Diagram}
\end{figure}

\subsection{Constructing the Hamiltonian}

Using well-known methods from the field of Hamiltonian complexity it is possible  to construct a quantum many-body system whose lowest-energy eigenstate represents the evolution of any desired computation \cite{Kitaev2002}.
If we introduce a local term in the Hamiltonian that gives additional energy to any state with overlap with the halting state of the computation, we can arrange for states representing computations that halt to pick up additional energy relative to states representing computations that do not halt, and open up a gap in the spectrum.
In this way, Turing's well-known \textsc{Halting problem} can be transcribed into a property of the quantum many-body system, namely whether or not it has a spectral gap.
Thus determining whether the system has a spectral gap is at least as hard the \textsc{Halting problem}.
Since the \textsc{Halting problem} is known to be undecidable, determining whether the Hamiltonian is gapped or gapless is also undecidable.
Conceptually, this is how \cite{Cubitt_PG_Wolf_Nature, Cubitt_PG_Wolf_Undecidability,Bausch_1D_Undecidable} proved undecidability of the spectral gap.


The starting point for our construction is also undecidability of the \textsc{Halting Problem}~\cite{turing_1937}: in brief, this states that determining whether a universal (classical) Turing machine (UTM) halts or not on a given input is, in general, undecidable.
In the quantum computation setting, \cite{Cubitt_PG_Wolf_Undecidability} showed how an input can be extracted from a phase in a quantum gate such as $\op U = \mathrm{diag}(1, \exp(2\pi \ii \varphi))$, using quantum phase estimation (QPE, \cite{Nielsen_and_Chuang}) which outputs a binary expansion of $\varphi$. The latter can then be fed as input to a UTM.
Thus this combination of QPE and UTM runs the universal Turing Machine on any desired input encoded in $\varphi$, and the \textsc{Halting Problem} for this combination is undecidable.

How do we reduce this QTM-based \textsc{Halting Problem} to a result about phases in a many-body system?
This is a culmination of the following techniques from previous works. However, for each one of them, significant obstacles must be overcome to prove uncomputability of phase diagrams.
\begin{enumerate}
	\item The first necessary ingredient is a QTM-to-Hamiltonian mapping which allows the construction of local, translationally-invariant couplings which result in a 1D spin chain Hamiltonian whose ground state energy is exactly zero if the encoded QTM does halts within a certain time interval; or otherwise is positive \cite{Gottesman-Irani}.
	Using such an QTM-to-Hamiltonian mapping, a QTM running the QPE + UTM computation described above is encoded into the spin chain Hamiltonian, with $\varphi$ now appearing as a parameter of the resulting Hamiltonian.
	However, the energy difference between the halting and non-halting cases decreases as the time interval increases, meaning we need further techniques to obtain a non-zero energy difference in the thermodynamic limit.
	\item \label{List:Point 2} A second ingredient is amplifying this penalty.
	In \cite{Cubitt_PG_Wolf_Undecidability} this is done by combining the QTM-to-Hamiltonian mapping with an aperiodic tiling Hamiltonian, thereby ensuring that, for each length of computation, a fixed density of such circuit-to-Hamiltonian instances are distributed across the spin lattice.
	In this way, the ground state energy density is zero iff the QTM-to-Hamiltonian mapping always has zero energy, and thus depends on whether the QPE+UTM computation ever halts.
	
	\item \label{List:Point 3} In \cite{Bausch_1D_Undecidable} point \ref{List:Point 2} is replaced by a so-called Marker Hamiltonian.
	This in combination with a circuit-to-Hamiltonian construction results in a ground state which partitions the spin chain into segments just large enough for the UTM to halt, if it halts.
	Here the segments do not have a fixed length, but instead self-adjust to find their own length.
	In contrast to \cite{Cubitt_PG_Wolf_Undecidability}, this has the effect that either all encoded instances of computations halt, or none do.
	\item The final step is the addition of an Ising-type coupling as in \cite{Cubitt_PG_Wolf_Undecidability,Bausch2015}, which breaks the local Hilbert space up into subspaces $\mathcal H_A \oplus \mathcal H_B$, and which ensures that the low-energy spectrum is contained either entirely in the $A$ or $B$ subspace, depending on the ground state energy density just constructed.
	Since determining the ground state energy density is uncomputable, it is also uncomputable to determine whether the system is in phase $A$ or $B$ with respect to the Hamiltonian parameter $\varphi$.
\end{enumerate}

As mentioned, we require significant alterations to this collection of ingredients.
Concretely, the issue is that if we encode an input $\varphi$ to be extracted using error-free QPE then we require the circuit gates to depend explicitly on the binary length of $\varphi$, denoted $|\varphi|$.
Consequently the resulting matrix elements of the Hamiltonian will also explicitly depend on $|\varphi|$---a discontinuous function of $\varphi$.
To remove this dependence we instead perform the QPE procedure approximately, by using a universal gate set that approximates all the gates depending on $|\varphi|$ \cite{Dawson2005}.
However, \cite{Cubitt_PG_Wolf_Undecidability, Bausch_1D_Undecidable}'s construction crucially relies on the QPE expansion of $\varphi$ to be performed exactly; any errors destroy the construction.

To overcome this obstacle, we first encode this approximate QPE plus the evolution of a UTM in a QTM-to-Hamiltonian mapping, which has positive energy iff the QPE + UTM computation does not halt.
We label the resulting Hamiltonian $\op H_\mathrm{comp}$.
This is outlined in sections \ref{sec:computation-in-ham}, \ref{Sec:Overview:Encoded_Computation} and \ref{Sec:QTM-to-Hamiltonian} where the QTM-to-Hamiltonian mapping and the computation it encodes are explained, respectively.
A significant novel technical contribution of this work is then a proof that the Marker Hamiltonian used in \cite{Bausch_1D_Undecidable} does, in fact, allow for some leeway in the precision to which QPE is performed and can be used to provide a correction for the energy penalty picked up as a result of any errors in the QPE.

To generate the required energy correction, we consider a 2D spin lattice and construct an underlying classical Hamiltonian, which we denote $\op H_\mathrm{cb}$, that partitions the lattice into a uniform grid of squares.
We note that the method from \cite{Cubitt_PG_Wolf_Undecidability} would be inappropriate for this construction as it would lead to an accumulation of energies we cannot correct for without matrix elements depending explicitly on $|\varphi|$.
Within each square, the ground state encodes the evolution of a classical Turing machine (encoded as a tiling problem akin to the ones used in \cite{Berge_Undecidability_Of_Domino_Problem,robinson1971undecidability}) which will calculate the energy correction necessary to offset the error introduced by approximately performing QPE.
The classical Hamiltonian is then coupled to the Marker Hamiltonian. We denote this combination $\HmarkerM$.


\Cref{sec:classical+quantum} describes the ground state of the resulting Hamiltonian $\HmarkerM$ such that the halting or non-halting behaviour together with the Marker Hamiltonian determines whether the energy density of the constructed Hamiltonian is non-negative (in the non-halting case), or negative (in the halting case).
Crucially, it is now robust with respect to the errors present in the expansion of $\varphi$ from the approximate QPE procedure.
Finally, in \cref{Sec:Combined_Result} we show how $\op H_\mathrm{comp}, \op H_\mathrm{cb}, \HmarkerM $ are combined mathematically to lift this undecidability of the ground state energy denisty, to uncomputablity of the phase diagram, using now-standard techniques~\cite{Cubitt_PG_Wolf_Undecidability,Bausch2015,Bausch_1D_Undecidable}.
For a mathematically rigorous derivation we refer the reader to the Supplementary Information.

\subsection{Encoding Computation in Hamiltonians}\label{sec:computation-in-ham}
A QTM can be thought of as a classical TM, but where the TM head and tape configuration can be in a superposition of states.
The updates to the QTM and tape configuration are then described by a transition unitary, $\op U$, describing the transitions of the QTM state and tape, such that the state is updated via the map $\ket{\psi}\mapsto \op U\ket{\psi}$ \cite{Bernstein1997}.

Given a QTM, one can construct a local Hamiltonian which has a ground state encoding the evolution of the computation \cite{Gottesman-Irani,Cubitt_PG_Wolf_Undecidability}, closely related to the Feynman-Kitaev Hamiltonian encoding quantum circuits into Hamiltonians \cite{Feynman_1982,Kitaev2002}.
The ground state encodes $T$ steps of the computation, where $T$ is a predefined and fixed function of the Hamiltonian's size determined by the particular QTM-to-Hamiltonian mapping.
The ground state of such a Hamiltonian is called a \emph{history state} and takes the form
\begin{align}
\ket{\Psi} = \frac{1}{\sqrt T}\sum_{t=1}^{T}\ket{t} \ket{\psi_t},
\end{align}
where the state of the quantum Turing machine at time step $t$ is $\ket{\psi_t}$.
The ground state energy of the Hamiltonian can be made dependent on aspects of the computation by adding a projector that penalises certain computational states, and the resulting energy is known to high precision \cite{Bausch_Crosson_2018,Watson_Hamiltonian_Analysis}.

\subsection{The Encoded Computation} \label{Sec:Overview:Encoded_Computation}

As in \cite{Cubitt_PG_Wolf_Undecidability,Bausch_1D_Undecidable}, the computation we wish to encode via such a QTM-to-Hamiltonian mapping will be a pair of QTMs running in succession: the first will run quantum phase estimation on a quantum gate $\op U_\varphi$ which outputs a number in binary, and the second will be a UTM which takes the output of the QPE as input.
The gate $\op U_\varphi$ is encoded in the transition unitary $\op U$ describing the QTM, which is in turn encoded in the matrix elements of the Hamiltonian.
The energy of the Hamiltonian encoding the computation will then be made dependent on whether the computation halts or not, allowing us to relate its ground state energy to the halting property.

\paragraph{Phase Estimation.}
Given a unitary matrix $\op U_\varphi=\big(\begin{smallmatrix}
1 & 0\\
0 & \ee^{\ii \pi \varphi}
\end{smallmatrix}\big)$, the quantum phase estimation (QPE) algorithm takes as input the eigenvector corresponding to the eigenvalue $\ee^{\ii \pi \varphi}$, and outputs an estimate of $\varphi$ in binary.
If the number of qubits on which the phase estimation is performed is smaller than the number of bits required to express $\varphi$ in full,  the algorithm is only approximate \cite{Nielsen_and_Chuang}.
Furthermore, if a finite gate set is used, some of the required unitary gates in the algorithm must be approximated rather than performed exactly \cite{Dawson2005}.
Hence from phase estimation we get an output state consisting of a superposition over binary strings:

\begin{align}\label{eq:qpe-coefficients-intro}
\ket{\chi(\varphi)}=\!\!\!\!\sum_{x\in\{0,1\}^n} \!\! \beta_x\ket{x},
\end{align}
where the amplitudes $\beta_x$ are concentrated around those values for which $x\approx \varphi$ and rapidly drop off away from $\varphi$.
Details are in \cref{sec:modified-qpe}. %

\paragraph{Universal QTM.}
We then feed the output $\ket{\chi(\varphi)}$ of this phase estimation into the input of a universal Turing Machine, as in \cite{Cubitt_PG_Wolf_Undecidability}, which then runs a computation which may or may not halt.
By the well-known undecidability of the \textsc{Halting Problem}~\cite{turing_1937}, determining whether the QTM halts for a given string is undecidable.

\subsection{From QTM to Hamiltonian}  \label{Sec:QTM-to-Hamiltonian}
Using the QTM-to-Hamiltonian mapping described in \cref{sec:computation-in-ham}, the computation outlined above is mapped to a one-dimensional, translationally-invariant, nearest-neighbour Hamiltonian $\HUTM(\varphi)$ \cite{Gottesman-Irani}, with a penalty for the non-halting case.
It can be shown that the ground state energy of $\HUTM(\varphi)$ scales as
\begin{equation}\label{eq:gs-energy-L}
\lmin(\HUTM(\varphi)) \sim \epsilon(L) / \poly L,
\end{equation}
where
\begin{equation}\label{eq:epsilon}
\epsilon(L) = \sum_{x\in S(L)} |\beta_x|^2.
\end{equation}
The $\beta_x$ are the QPE coefficients in \cref{eq:qpe-coefficients-intro}, and $S(L)$ is the set of inputs for which the universal TM does not halt within time $T(L)$.

Since the $\beta_x$ are concentrated around the binary expansion of $\varphi$, if the latter encodes a halting instance there will be a length $L_0$ for which $\epsilon(L) \approx 0$ for all $L>L_0$; otherwise $\epsilon(L) \approx 1$ for all $L$.
This immediately yields a Hamiltonian for which the ground state energy is halting-dependent (and hence uncomputable).
We refer the reader to \cref{sec:QPE-UTM} for details.

\subsection{Tiling and Classical Computation} \label{Sec:Tiling+Classical}
In \cref{eq:gs-energy-L} we see that the difference between the Hamiltonian's ground state energy in the case where $\epsilon(L)$  from \cref{eq:epsilon} is approximately $1$ or $0$ decreases with the system size $L$.
Thus the energy gap between the two cases goes to zero irrespective of whether $\varphi$ encodes a halting or non-halting instance.
To amplify this gap so that there is a finite energy gap in the thermodynamic limit (as per points \ref{List:Point 2} and \ref{List:Point 3}), we will combine the Feynman-Kitaev Hamiltonian with a classical Hamiltonian based on a Wang tiling that partitions the space suitably to ensure a fixed density of computation instances is spawned across the lattice.
The result we achieve with this is an energy gap opening up as $L$ grows between the cases where $\epsilon$ takes different values.

A set of Wang tiles---square, 2D tiles with coloured sides, with the rule that adjacent sides of neighbouring tiles in a tiling of the plane must have matching colours---can be mapped to a classical Hamiltonian: if tiles $t_i$, $t_j$ cannot be placed next to each other, then we introduce a term $\ketbra{t_it_j}$ into the Hamiltonian.
The overall Hamiltonian is the sum over all such local terms, such that its ground state represents a tiling satisfying the tiling rules (if such a tiling exists).
If no such tiling exists, the ground state energy is $\geq 1$.

Similarly, it is well known that there exist tile sets that encode the evolution of a classical TM \cite{Berge_Undecidability_Of_Domino_Problem,robinson1971undecidability} within a square grid: TM tape configurations are represented by rows, such that
adjacent rows represent successive time steps of the TM (\cref{Fig:TM_Encoded_As_Tiles}).

We combine both Wang tiles and the Turing Machine tiling ideas by constructing a tile set whose valid tilings have the following properties:
\begin{enumerate}
	\item  A tiling pattern that creates a square grid across the lattice $\Lambda$ (much like a checkerboard). The side length of the grid squares can be varied (provided all grid squares are the same size) and still correspond to a valid tiling (depicted in \cref{Fig:Checkerboard_Tiling_2}).
	\item Within each square of the grid we use the tiles to encode a TM which first counts the size of the square it is contained in, and then outputs a marker $\bullet$ on the top border of the square, where the placement of this marker is a function of the size of the square (depicted in in \cref{Fig:TM_Encoded_As_Tiles}).
	
	
\end{enumerate}

Having developed this tiling, we map it to a corresponding tiling Hamiltonian, which we denote $\Hchecker$, such that its ground states retain the properties of the valid tilings listed above.
The reader is refered to \cref{Sec:Lattice_Tiling} for details.

\begin{figure}[t]
	\centering
	\includegraphics{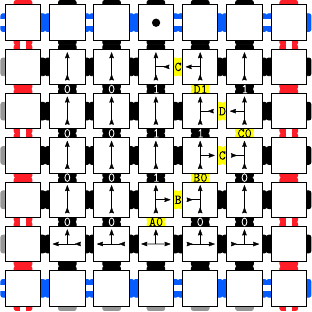}
	\caption{The evolution of a classical TM can be represented by Wang tiles, where colours of adjacent tiles have to match, and arrow heads have to meet arrow tails of the appropriate kind.
		Here the evolution runs from the bottom of the square to the top, where it places a marker $\bullet$ on the boundary as explained in \cref{sec:classical+quantum}.
		In this image, the TM's evolution is contained in an individual square in the checkerboard grid shown in the figure below.}
	\label{Fig:TM_Encoded_As_Tiles}
\end{figure}

\begin{figure}[t]
	\includegraphics[width=\textwidth]{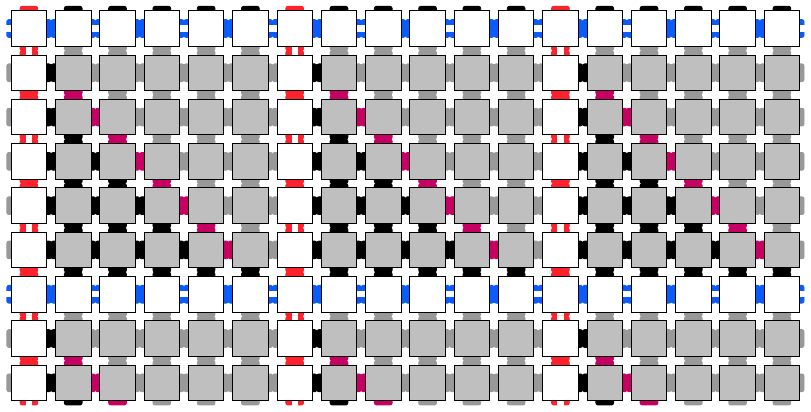}
	\caption{Section of the checkerboard tiling Hamiltonian's ground state.
		The white squares form borders, and in the interior we place tiles simulating the evolution of a classical Turing Machine.}
	\label{Fig:Checkerboard_Tiling_2}
\end{figure}

\subsection{Classical Tiling with Quantum Overlay}\label{sec:classical+quantum}

We now want to combine the classical Hamiltonian encoding the Wang tiles, and the quantum Hamiltonian encoding the \textsc{Halting Problem} computation, to create an overall Hamiltonian which has a large ground state energy difference between the halting and non-halting cases, without the $\sim 1/\poly(L)$ decay in \cref{eq:gs-energy-L}.
To do so, we split the local Hilbert space of each lattice spin into a classical part $\HS_\mathrm{c}$ and a quantum part $\HS_\mathrm{e} \oplus \HS_\mathrm{q}$ giving $\mathcal{H}=\HS_\mathrm{c} \otimes (\HS_\mathrm{e} \oplus \HS_\mathrm{q})$, where $\HS_\mathrm{e}=\{\ket{e}\}$ just contains a filler state $\ket{e}_\mathrm{e}$.
The ground state can then be designed to be a product state of the form $\ket{C}_\mathrm{c}\otimes\ket{\psi_0 }_\mathrm{eq}$, where $\ket{C}_\mathrm{c}$ is a valid classical tiling configuration---as described in \cref{Sec:Tiling+Classical}---and $\ket{\psi_0 }_\mathrm{eq}$ is a quantum state with the following properties:
\begin{enumerate}
	\item  We use the 1D Marker Hamiltonian from \cite{Bausch_1D_Undecidable}, and couple its negative energy contribution to the size of each grid square in the classical tiling and the placement of the $\bullet$ marker.
	The negative energy each square contributes is a determined by where the $\bullet$  marker is placed, and thus by the action of the classical TM.
	We denote this combined Hamiltonian $\HmarkerM$.
	
	\item We effectively place the ground state of a Hamiltonian $\HUTM$ encoding the quantum phase estimation plus universal Turing machine along the top edge of the square, by adding additional penalty terms to the Hamiltonian that penalize the classical and quantum layers to occur in this configuration elsewhere.
	
	\item Everywhere not along the horizontal edge of a grid square in $\HS_\mathrm{c}$ is in the zero-energy $\ket{e}_\mathrm{e}$ filler state in $\HS_\mathrm{e}\oplus \HS_\mathrm{q}$.
	
\end{enumerate}
As mentioned, the patterns in the degenerate ground space of $\Hchecker$ are checkerboard grids of squares with periodicity $w\times w$, where the integer square size $w$ is not fixed.

By choosing the classical TM encoded in the tiling to place a $\bullet $ marker at an appropriate point, we are able to tune the ground state energy of $\HmarkerM$
such that the total energy of a single $w\times w$ square $A$ in the checkerboard pattern is:

\begin{align}\label{eq:lmin}
\lmin(w) := \lmin\left(\HmarkerM|_A + \HUTM|_A\right)
\begin{cases}
\ge 0  &\text{if $\epsilon(w)\geq \epsilon_0(w)$ $\forall w$.} \\
< 0   &  \text{if $\epsilon(w)<\epsilon_0(w)$ $\forall w\geq w_0$}
\end{cases}
\end{align}
where $\epsilon_0(w)$ is some cut-off point, $w_0$ is the halting length (recall from the previous section that the runtime of the computation encoded in the ground state depends on the size of the available tape, i.e.\ the size of the checkerboard square edge that the TM runs on), and where $\lmin(w_0) = -\delta(w_0) < 0$ for the halting length $w_0$ is a small negative constant. 

The Marker Hamiltonian's energy bonus thus compensates for the QPE approximation errors by lowering the energy by just enough such that a halting instance has negative energy.
On the other hand, the energy of the non-halting instance remains large enough that the energy of a single square remains positive \cite{Bausch_1D_Undecidable, Watson_Hamiltonian_Analysis}.

Thus, provided $\epsilon(w)$ is sufficiently small, the ground state of $\Hchecker + \HUTM + \HmarkerM$(+coupling terms) is a checkerboard grid of squares with a constant but negative energy density.
Otherwise the ground state energy density of the lattice is lower-bounded by zero.
Which of the two cases holds depends on determining whether $\epsilon(w)\geq \epsilon_0(w)$ or $<\epsilon_0(w)$, which is undecidable; undecidablity of the ground state energy density follows.
The reader is refered to \cref{Sec:Marker_Hamiltonian} for more details.

We define the Hamiltonian formed by $\Hchecker, \ \HUTM, \ \HmarkerM$ and the coupling terms as $H_u(\varphi)$.
Assume $\varphi$ encodes a halting instance and set $w=\argmin_s\{ \lmin(s)< 0 \}$, and $A$ is a single square of size $w\times w$.
Then the ground state energy $ H_u(\varphi)$ on a grid $\Lambda$ of size $L\times H$ is given by
\begin{equation}\label{eq:gs-energy-density}
\lmin( H_u(\varphi))=\left\lfloor \frac{L}{w} \right\rfloor\left\lfloor \frac{H}{w} \right\rfloor \lmin( H_u(\varphi)|_A).
\end{equation}

\subsection{Uncomputability of the Phase Diagram} \label{Sec:Combined_Result}
To go from undecidability of the ground state energy density, demonstrated at the end of \cref{sec:classical+quantum}, to undecidability of the phase (and spectral gap) we
follow the approach of \cite{Cubitt_PG_Wolf_Undecidability, Cubitt_PG_Wolf_Nature} by combining $H_u(\varphi) $ with a trivial state $\ket{0}$ such that $\ket 0^{\otimes\Lambda}$ has zero energy, and the spectrum of $\op H_u(\varphi)$ is shifted up by $+1$ (see \cref{lem:shift-ham}). 
From this shift and \cref{eq:gs-energy-density} it can be shown:

\begin{align} \label{Eq:GS_Energy_Shifted}
\lmin( H_u(\varphi)) \begin{cases}
\ge 1 & \text{in the non-halting case, and} \\
\longrightarrow-\infty & \text{otherwise.}
\end{cases}
\end{align}
Let $\op h_u(\varphi)$ denote the local terms of $\op H_u(\varphi)$, let $\op h_d$ be the local terms of the critical XY-model, and let $\ket{0}$ be a zero energy state, such that the total Hilbert space is $(\HS_1\ox \HS_2) \oplus \{\ket{0}\}$.
Then the local terms of the total Hamiltonian $\op H^{\Lambda}(\varphi)$ are defined as

\begin{align} \label{Eq:Complete_Local_Terms}
\op h_{i,i+1}(\varphi) \coloneqq& \ketbra{0}^{(i)}\ox (\1 - \ketbra{0})^{(i+1)} + (\1 - \ketbra{0})^{(i)}\ox \ketbra{0}^{(i+1)} \\ 
&+\op h_u^{i,i+1}(\varphi)\ox \1^{(i)}_2 \ox \1^{(i+1)}_2  + \1_1^{(i)}\ox \1_1^{(i+1)} \ox  \op h_d^{i,i+1}.
\end{align}
The result is the following: the overall spectrum of the Hamiltonian is
\[
\spec(\Hundec(\varphi)) = \{0\} \cup (\spec(\op H_u(\varphi)) + \spec(H_d)) \cup G,
\]
where $G\subset [1,\infty)$.
To understand this, we consider the two cases.
If a non-halting instance is encoded $\lmin \geq 0$ in \cref{eq:lmin}, then $ H_u(\varphi) $---from \cref{Eq:GS_Energy_Shifted}---has ground state energy lower-bounded by $1$;
the ground state of the overall Hamiltonian is the trivial zero energy classical product state $\ket{0}^{\ox \Lambda}$, and $\op H^{\Lambda}$ has a constant spectral gap.
If $\lmin<-|\delta|$, then $\op H_u(\varphi) $ has a ground state with energy diverging to $-\infty$.
We further note that the critical XY-model has a dense spectrum with zero ground state energy, hence from the $H_d$ term we obtain a dense spectrum above the ground state \cite{Lieb_Schultz_Mattis_1961}.
As a result the Hamiltonian becomes gapless and has a highly entangled ground state with algebraically decaying correlations.

Since existence of a halting length $w_0$ in \cref{eq:lmin} is undecidable, discriminating between $\lmin\ge 0$ or $<-|\delta|$ is also undecidable.
This implies determining whether the Hamiltonian is in the critical, quantum-correlated phase or the trivial product state gapped phase is undecidable as well.

As $\op H^\Lambda(\varphi)$ is a continuous function of $\varphi$, there exist finite measure regions for which all values of $\varphi$ have the same ground state and for which there is no closing of the spectral gap, which delineates the two phases.
Setting $\Pi_i \coloneqq \ketbra{0}^{(i)}$, the observable $\op O_{A/B}=L^{-2}\sum_{i\in \Lambda} \Pi_i$ has expectation value $1$ when in the state $\ket{0}^{\otimes \Lambda(L)}$, and $0$ in the other case.
This is true even if the observable is restricted to a finite geometrically local subset of the lattice.
We refer the reader to \cref{Sec:Undecidable_Continuous_Family} for more details.

\section{Discussion}

Our result proves undecidability of the phase and spectral gap for a continuous, one-parameter family of Hamiltonians on a two-dimensional lattice.
An immediate consequence is that there is no algorithm which can compute a Hamiltonian's phase diagram in general, even given a complete description of its microscopic interactions.

Qualitatively, this brings the results close to classic condensed matter models, for example the transverse Ising model described by the Hamiltonian $\op H_\mathrm{TIM} = \sum_{\langle i,j\rangle} \sigma_z^{(i)}\sigma_z^{(j)} +\varphi \sum_i \sigma_x^{(i)}$.
Here the real parameter $\varphi$ determines the strength of an external magnetic field.
Its phase diagram comprises an ordered and a disordered phase, with an order parameter given by the macroscopic observable $\op O_\mathrm{TIM} = \frac{1}{N}\sum_{j=1}^N \sigma_z^{(j)}$ corresponding to the global magnetisation per spin.
In the ordered phase, the ground state expectation value $\abs{\langle \op O_\mathrm{TIM}\rangle}=1$.
In the disordered phase, the ground state expectation value is 0.
Both these phases have a non-zero spectral gap in the thermodynamic limit.
At the critical point $\varphi=1$ between the two phases, the system is gapless and exhibits criticality.

As mentioned previously, we can take as the order parameter the macroscopic observable $\op O_{A/B} = \frac{1}{N}\sum_{j}\Pi_A^{(j)}$, where $\Pi_A^{(j)}$ is a projector onto the local Hilbert space $\mathcal H_A$ on lattice site $j$.
This definition yields a familiar picture: for those $\varphi$ such that $\op H(\varphi)$ is in the $A$ phase, we prove that the ground space is non-degenerate with $\langle \op O_{A/B}\rangle = 1$.
In the $B$ phase, all ground states have expectation value $\langle \op O_{A/B} \rangle = 0$ with respect to this observable.
Thus, the phases $A$ and $B$ are distinguished by an order parameter given by a macroscopic observable $\op O_{A/B}$, and the system undergoes a (first order) phase transition between these as $\varphi$ varies.

However, unlike the Ising model, in our case one phase (phase $A$, say) is gapped, but the other (phase $B$) is a gapless phase.
This gapped vs.\ gapless phase transition has further phenomenological consequences. For instance, a transition between $A$ and $B$ implies a transition from exponential decay of correlations to long-range correlations with algebraic decay of correlation functions.
In fact, something stronger holds in our case: in the gapped phase, the ground state is a product state and all connected correlation functions are strictly zero.

If we define a phase diagram of a $k$-parameter Hamiltonian as a normalised parameter space $[0,1]^{k}$ which maps out the different phases as a function of the parameters at each point $p\in [0,1]^{k}$, then previous undecidability results~\cite{Cubitt_PG_Wolf_Undecidability,Bausch_1D_Undecidable} do not imply uncomputablity of phase diagrams, for multiple reasons.
There, as here, the matrix elements of the local interactions of the Hamiltonian $\op H(\varphi)$ depend on an external parameter $\varphi$ which determines the gappedness of the Hamiltonian.
However, importantly, in the previous constructions the matrix elements also depend on the binary length of $\varphi$, denoted $|\varphi|$, which is a discontinuous function of $\varphi$.
A consequence of this, it is not possible to define a meaningful phase diagram for these Hamiltonians over the required parameter range.
This significantly limits the implications one can draw from previous spectral gap undecidability results, in particular for quantum phase diagrams, which are one of the main reasons for caring about spectral gaps in the first place.

Although the construction developed herein proves undecidability between phases defined by $\op O_{A/B}$, the result can be extended to more general phase diagrams by a small modification to the construction.
If we modify the Hamiltonian by introducing two terms $\op h_{\neg X}$ and $\op h_{X}$ which are Hamiltonians that, respectively, have and do not have the ground state property $X$ in the thermodynamic limit (specifically, in equation \ref{Eq:Complete_Local_Terms} defining the Hamiltonian, we replace $\ketbra{0}$ and $\op h_d$ with $\op h_{\neg X}$ and $\op h_X$), then the new overall Hamiltonian will have two phases, one of which has property $X$ and another which does not.
Determining which of the two properties holds is undecidable.

Furthermore, the algorithmic uncomputability of the phase diagram problem implies axiomatic independence of the problem \cite{Poonen}.
That is, for any consistent formal system with a recursive set of
axioms, there exists a Hamiltonian of the form given in \cref{th:main} such that determining the phase diagram from the given axioms is not possible.

There are other consequences:
a common technique in numerical condensed matter physics to estimate the phase of a physical system is to take the Hamiltonian on an $L\times L$ lattice, calculate the phase for this lattice size by some numerical means, and then extrapolate its phase to the thermodynamic limit.
This is justified by the assumption that as long as $L$ is sufficiently large, the system already displays the behaviour of the thermodynamic limit.
In \cite{Cubitt_PG_Wolf_Nature} it was shown that this assumption is not justified in all cases, as the phase may without warning appear completely different at some arbitrarily large and uncomputably system size. 
Leading to the phenomenon of sized-driven phase transitions explored in \cite{Bausch2015}.
Our result further extends this to show that attempting to computing the phase diagram by extrapolating from some finite size system may not reflect the phase diagram in the thermodynamic limit.

We note, however, that our result only establishes uncomputability for certain highly complex and artificially constructed Hamiltonians. 
Furthermore, the Hamiltonians constructed by our techniques are necessarily frustrated. For many commonly occuring Hamiltonians---particularily those with small local Hilbert space dimension---determining the phase may well be rigorously decidable.
For example using techniques from \cite{Knabe_1988,Nachtergaele_1996}, as done e.g.\ in \cite{Bravyi_Gosset_2015}, completely solves the case of frustration-free, nearest-neighbour 1D qubit chains.


As aforementioned, previous gap undecidability results \cite{Cubitt_PG_Wolf_Undecidability,Bausch_1D_Undecidable} required the explicit inclusion of the binary length of $\varphi$, i.e.\ $|\varphi|$, as matrix elements of the form
\[
2^{-2|\varphi|}
\quad\quad\text{or}\quad\quad
\ee^{-\ii \pi2^{-2|\varphi|}}.
\]
It is clear that one cannot vary
$\varphi$ along a continuous path between two points $\varphi_1$ and $\varphi_2$ while keeping the length of its binary expansion $|\varphi|$ fixed at all points along the path.
Moreover, varying $\varphi $ and $|\varphi|$ separately breaks the construction.
As a result, it is impossible to draw a phase diagram with respect to the parameter $\varphi$ for the Hamiltonians of \cite{Cubitt_PG_Wolf_Undecidability,Bausch_1D_Undecidable}.

Phase transitions are typically defined as points at which there is a non-analyticity in the ground state energy (or some associated order parameter) with respect to continuous changes in $\varphi$.
If one were to view $|\varphi|$ as an explicitly discontinuous function of $\varphi$, the ground state may no longer analytically depend on $\varphi$ even at points where the system is gapped.
As such it is unclear if it is even meaningful to define a phase or phase transition for the models presented in \cite{Cubitt_PG_Wolf_Undecidability,Bausch_1D_Undecidable}.
In addition, Hamiltonians depending discontinuously on a continuously varying parameter are not typically encountered in physics.

The family of Hamiltonians we construct in this work is truly continuous, i.e.\ we define our local terms $\hrow(\varphi)$ for arbitrary $\varphi\in\field{R}$,
thus even irrational numbers with infinitely long binary expansions are perfectly fine as instances of our problem set-up.
This brings us qualitatively closer to models of Hamiltonians of real systems, where the parameter varied will typically be some physical property, such as an applied magnetic field, which can be varied continuously.

Since \cref{th:main} shows that for any $\varphi$ there exists a small finite interval around $\varphi$ for which the phase of the Hamiltonian is the same, we have a notion of stability of undecidability under perturbations to $\varphi$---something which was not the case in previous results.
However, it is not clear if there is any stability of the Hamiltonian's properties with respect to perturbations in arbitrary matrix elements.
Stability of undecidability under arbitrary local perturbations to the Hamiltonian remains a challenging but important topic for future research, but one that has yet to be fully resolved even for simple models such as the Ising model.
Finally, it important to emphasise that the Hamiltonian constructed here is highly artificial, in the sense that it has an unnaturally large local Hilbert space dimension and highly complex, specifically tailored interactions.
Whilst size-driven phase transitions have been discovered in much simpler models  \cite{Bausch2015}, and recent results in Hamiltonian complexity theory \cite{Cubitt_Montanaro_2013,Bausch2016,Cubitt_Montanaro_Piddock_2018,Bausch2017} show that related complexity theoretic properties can also occur in far simpler models, it remains an open question whether undecidability occurs in any remotely natural Hamiltonians.

Ongoing research directions both in Hamiltonian complexity and computability focus on reducing the physical dimensionality of the system, reducing the local Hilbert space dimension, and choosing physical interactions comparable to those seen in physical systems.
A further route of investigation would be to determine how difficult it is to compute phases for systems of finite size, rather than in the thermodynamic limit, for some suitable definition of phase in the finite-size setting.


	\subsection*{Data Availability}
	Data sharing not applicable to this article as no datasets were generated or
	analysed during the current study.

	\subsection*{Acknowledgements}
	J.\,B.\ is supported by the Draper's Research Fellowship at Pembroke College.
	J.\,D.\,W.\ is supported by the EPSRC Centre for Doctoral Training in Delivering Quantum Technologies [EP/L015242/1].
	T.\,S.\,C.\ is supported by the Royal Society.
	This work was supported by the EPSRC Prosperity Partnership in Quantum Software for
	Simulation and Modelling (EP/S005021/1).

	\subsection*{Author Contributions}
	J.\,D.\,W., J.\,B., and T.\,S.\,C.\  all contributed significantly to the research.

	\subsection*{Competing Interests}
	The authors declare no competing interests.

\printbibliography

\newpage
\begin{huge}
	\noindent	\textbf{Supplementary Material}
\end{huge}
\appendix

\section{Preliminaries and Definitions}
Gapped and gapless are rigorously defined as follows:
\begin{definition}[Gapped, from \cite{Cubitt_PG_Wolf_Undecidability}]\label{def:gapped}
	We say that $\Hundec[L]$ of Hamiltonians is gapped if there is a constant $\gamma>0$ and a system size $L_0\in\field{ N}$ such that for all $L>L_0$, $\lmin(\Hundec[L])$ is non-degenerate and $\Delta(\Hundec[L])\geq\gamma$. In this case, we say that \emph{the spectral gap is at least $\gamma$}.
\end{definition}
\begin{definition}[Gapless, from \cite{Cubitt_PG_Wolf_Undecidability}]\label{def:gapless}
	We say that $\Hundec[L]$ is gapless if there is a constant $c>0$ such that for all $\epsilon>0$ there is an $L_0\in\field{N}$ so that for all $L>L_0$ any point in $[\lmin(\Hundec[L]),\lmin(\Hundec[L])+c]$ is within distance $\epsilon$ from $\spec \Hundec[L]$.
\end{definition}
As noted in \cite{Cubitt_PG_Wolf_Undecidability}, gapped is not defined as the negation of gapless; there are systems
that fall into neither class, such as systems with closing gap or degenerate ground states.
However, the stronger definitions allow us to avoid any potentially ambiguous cases.

\begin{definition}[Continuous family of Hamiltonians]\label{def:continuous}
	We say that a Hamiltonian $\op H(\varphi) = \sum_j \op h_j(\varphi)$ depending on a parameter $\varphi\in\field{R}$, made up of a sum over local terms $\op h_j(\varphi)$ each acting on a local Hilbert space $\HS$, is \emph{continuous} if each $\op h_j(\varphi):\field{R} \longrightarrow \mathcal B(\HS)$ is a continuous function.
	We say that a family of Hamiltonians $\{ \op H_i(\varphi) \}_{i\in I}$ for some index set $I$ is a \emph{continuous family} if each $\op H_i(\varphi)$ is continuous.
\end{definition}
\section{Modified Quantum Phase Estimation}\label{sec:modified-qpe}

Here we discuss the modified version of quantum phase estimation which will be encoded in the Hamiltonian as discussed in the main body of the paper.
The role of the quantum phase estimation (QPE) is simply to output $\varphi$ in binary on the tape so that it can later be used as the input to a universal quantum TM.
The idea will be to take a given unitary with eigenvalue $e^{\ii\pi \varphi}$ and its associated eigenvector, and then use QPE to provide an estimate of $\varphi$ with the greatest precision possible.
The precision of the QPE procedure will be limited by two factors: imperfect gates and the fact the binary length of $\varphi$ may be longer than the available track length \cite{Nielsen_and_Chuang}.

We first review the QPE used in the related works of \cite{Cubitt_PG_Wolf_Undecidability} and \cite{Bausch_1D_Undecidable} before explaining how our QPE procedure is different and why it needs to be so.
We characterise the amount of error the modified QPE procedure introduces which will play an important part in the construction of the Hamiltonian in later parts of the construction.

\subsection{The State of the Art}

In \cite{Cubitt_PG_Wolf_Undecidability}, the QPE on a unitary $\op U_\varphi$ with eigenvalue $\ee^{\ii\pi \varphi}$ can output the binary form of $\varphi$ exactly:
this is due to the fact that, in their construction, the QTM has access to a perfect gate set that is sufficient to expand exactly $|\varphi|$ digits---in particular, the standard QPE algorithm requires performing small controlled rotation gates $\op R_n$ with angles $2^{\ii\pi 2^{-n}}$ for $n=1,\ldots,|\varphi|$, and since $|\varphi|$ is explicitly encoded in the local terms of the Hamiltonian, this circuit can be performed.

Furthermore, in \cite{Bausch_1D_Undecidable} the QPE procedure can also be performed perfectly by the same method.
In addition one can detect when the binary expansion of $\varphi$ is too long for the tape available to the QTM and penalize said segment lengths accordingly---the Marker Hamiltonian then has as a ground state which is a partition of the 1D spin chain into segments of length just long enough to perform QPE on $\varphi$ and for the dovetailed TM to halt---if it halts.

In our new construction the situation is fundamentally different.
Since the local terms of our Hamiltonian $ \Hundec(\varphi) $ do not explicitly depend on $|\varphi|$ anymore, we cannot provide the QPE with a set of rotation gates sufficient to perform an exact quantum Fourier transform (QFT). 
This means that we cannot guarantee the parameter we are estimating has a binary expansion short enough to be written on the tape available.

We therefore have to change the construction in two key ways.
First, our encoding of $\varphi$ will be in unary instead of binary.
Since this is a undecidability result we are not constrained by poly-time reductions---or indeed any finite computational resources; any runtime overhead is acceptable.
Secondly, we will perform some gates in the QPE only approximately. The gate approximation uses standard gate synthesis algorithms from Solovay-Kitaev \cite{Dawson2005}, where we gear the precision of the algorithm such that it suffices to obtain a large enough certainty on the first $j$ digits of $\varphi$, given our tape has said length.
The error resulting from truncating $\varphi$ to $j$ digit is more involved, as QPE yields a superposition of states close \emph{in value} to $\varphi$, which can for example mean that it rounds an expansion like $0.00001111$ to $0.00010$.
We will circumvent this issue by choosing an encoding which lets us easily discover and penalize a too-short expansion, similar to the one in \cite{Bausch_1D_Undecidable}.

\subsection{Notation}
\newcommand{\upto}[2]{\bar{#1}_{\cdots #2}}
Throughout we will denote the binary expansion of a number $x$ as $\bar{x}$, and the first $j$ digits of such an expansion as $\upto xj$. A questionmark $?$ will denote a digit that can either be a $0$ or a $1$. The $j$\textsuperscript{th} digit of $\bar x$ will then be $\bar x_j$.
For a given number $x$, we define $\clz x$ to be the count of leading zeros until the first 1 within $\bar x$---where we set $\clz 0=\infty$.
Similarly, we define the string $\pfx x$ to be the prefix of the string $\bar{x}$ such that $\pfx x= 0^{\times \clz x}1$, i.e.\ $\bar x = (\pfx x)??\ldots$.

Within this section, we will further denote by $\op U_\varphi$ a local unitary operator with eigenvalue $\ee^{\ii \pi \varphi}$, and will refer to $\varphi$ as the phase to be extracted.

Finally, let $\mathcal{M}$ be a universal reversible classical TM that takes its input in unary, i.e.\ as a string $00\ldots0100\ldots$; everything past the first leading one will be ignored;
we lift $\mathcal M$ to a quantum TM by standard procedures \cite{Bernstein1997}.

In the following analysis we first start with an encoding scheme and analyse how the approximate QPE behaves on it;
we finally show that each encoded parameter $\varphi$ admits a small $\epsilon$-ball around it where the system behaves in an identical fashion, making the behaviour of gapped vs.\ gapless robust and showing that our family of Hamiltonians is undecidable on a non-zero-measure set over the entire parameter range $\varphi\in[0,1]$.
We do not make a claim of knowing how the construction behaves for \emph{any} choice of parameter.
That is, given a particular value of $\varphi$, even if the halting behaviour of $\mathcal{M}$ on input $\clz\varphi$ were known, this would not always be sufficient to determine the behaviour of the Hamiltonian at this point.

\subsection{Exact QPE with Truncated Expansion}
We deal with the expansion error of our phase estimation first. As already mentioned, we need to choose an encoding that lets us detect and penalize expansion failure.
\begin{definition}[Unary Encoding]\label{def:qpe-encoding}
Let $\ivar\in\field N$ be the input we wish to encode. Then
\[
    \varphi = \varphi(\ivar) := \mathtt{0.}\underbrace{\mathtt{000\cdots0}}_{\ivar-1\ \text{digits}}\mathtt{100\cdots} \equiv 2^{-\ivar}.
\]
\end{definition}
As mentioned, it is unclear a priori how much overlap the post-QPE state has with binary strings that encode the same number in unary (i.e.\ the string with the same number of leading 0 digits).
The benefit of using the above encoding is that phase estimation tends to \emph{round} numbers that are too short to be expanded in full.
Since we are encoding small numbers (assuming a little Endian bit order), this rounding will produce a large overlap with the all-zero state $\ket{\bar 0}$.
If we then penalize this outcome---e.g.\ by defining the dovetailed TM to move right forever on a zero input, which means it does not halt---we can ensure that the tape length will be extended until the input can be read in full, at which point there is no further expansion error to deal with.

As a first step we analyse the approximate quantum phase estimation procedure and compare the associated error with the perfect case,
meaning that for now we give the QTM access to the same operations as in \cite{Cubitt_PG_Wolf_Undecidability} and \cite{Bausch_1D_Undecidable}, which includes access to the unitary $\op U_\varphi$ and rotation gates  $\op R_n=\big(\begin{smallmatrix}
1 & 0\\
0 & 2^{\ii\pi2^{-|\varphi|}}
\end{smallmatrix}\big) $ which suffice to perform phase estimation exactly. We then do the QPE algorithm identically to that laid out in \cite{Cubitt_PG_Wolf_Undecidability}; as this is the standard QPE algorithm from \cite{Nielsen_and_Chuang}, we phrase the following lemma in a generic way.

\begin{lemma} \label{Lemma:QPE_Truncation_Error}
		Let $\varphi(\ivar)\in\field R$ be a unary encoding of $\ivar\in\field N$ as per \cref{def:qpe-encoding}.
        On $t$ qubits of precision, QPE is performed on the unitary $\op U_\varphi$ encoding $\varphi(\ivar)$ and its associated eigenstate defined in \cref{def:qpe-encoding};  denote the QPE output by $\ket\chi$.
        Then either:
        \begin{enumerate}
        \item $t\ge|\varphi|$, and $\ket\chi = \ket{\bar \varphi}$,
        \item $t<|\varphi|$, and
        \[
            \ket\chi = \!\!\! \sum_{x\in \{0,1\}^t} \!\!\beta_x \ket x
            \quad
            \text{with}
            \quad
            |\beta_0| \ge \frac 12.
        \]
        \end{enumerate}
\end{lemma}
\begin{proof}
    The first case is clear---we have a perfect gate set and sufficient tape, hence QPE is performed exactly \cite{Nielsen_and_Chuang}.
    
    For the second case where $t<|\varphi|$ the $\beta_x$ are given in \cite[eq.~5.25]{Nielsen_and_Chuang},
    \begin{equation}\label{eq:qpe-weights}
        \beta_x = \frac{1}{2^t} \frac{1-\exp(2\pi \ii (2^t\varphi - (b+x)))}{1-\exp(2\pi\ii(\varphi - (b+x)/2^t))},
    \end{equation}
    where $b$ is the best $t$ bit approximation to $\varphi$ \emph{less} than $\varphi$, i.e.\ $0 \le \varphi - 2^{-t} b \le 2^{-t}$.
    Since we are using a unary encoding, as per \cref{def:qpe-encoding}, it is always the case that $b\equiv 0$, and therefore 
    \begin{align*}
        \beta_x &= \frac{1}{2^t} \frac{1-\exp(2\pi \ii (2^{t-a} - x))}{1-\exp(2\pi\ii(2^{-a} - x/2^t))} \\
    \intertext{Considering the case $x=b$ gives:}
        |\beta_0| &= \frac{1}{2^t} \frac{|1-\exp(2\pi \ii 2^{t-a})|}{|1-\exp(2\pi\ii 2^{-a})|} \\
         &= \frac{1}{2^t}\frac{ \sin(\pi 2^{t-a}) }{ \sin(\pi 2^{-a} )} \\
         &\geq \frac{1}{2^t}\frac{(\pi 2^{t-a})}{\pi 2^{-a}} = \frac{1}{2}.
    \end{align*}
    Going from the second last to last line we have used that $x/2 \le \sin(x) \le x$ for $x\in[0,\pi/2)$.
\end{proof}

\begin{corollary}\label{cor:full-gateset-qpe}
    Take some $\ivar\in\field N$ and $\varphi(\ivar)$ as defined in \cref{def:qpe-encoding}.
    Running the same quantum phase estimation QTM as in \cite{Cubitt_PG_Wolf_Undecidability} to precision $m$ bits yields an output state $\ket\chi$ given in \cref{Lemma:QPE_Truncation_Error}, such that either
    \begin{enumerate}
    \item $m\ge \ivar$ and $\ket\chi=\ket{\bar\varphi(\ivar)}$, or
    \item $m<\ivar$ and $|\braket{\chi}{0}|\ge 1/2$.
    \end{enumerate}
\end{corollary}

What if $\varphi(\ivar)$ is not exactly given by the encoding in \cref{def:qpe-encoding}? It is clear that $\ket\chi$ is still a superposition of bit strings $\ket x$, weighted by $\beta_x$ as in \cref{eq:qpe-weights}.
But our encoding allows us to derive a variant for \cref{cor:full-gateset-qpe} that applies to an interval around the correctly-encoded inputs.
Here we prove that we still have a large overlap with the all zero if the phase $\varphi$ is not expanded fully.
\begin{corollary}\label{cor:full-gateset-qpe2}
	Let $\ivar\in\field N$,  and $\varphi(\ivar)$ as in \cref{def:qpe-encoding}.
	Take a perturbed phase $\varphi' \in [\varphi(\ivar), \varphi(\ivar) + 2^{-\ivar-\ell} )$ for some $\ell\in\field N$, $\ell\ge1$.
	Running the same quantum phase estimation QTM as in \cite{Cubitt_PG_Wolf_Undecidability} to precision $m$ bits yields an output state $\ket\chi$ given in \cref{Lemma:QPE_Truncation_Error}, such that either
	\begin{enumerate}
		\item $m\ge\ivar$ and $|\braket{\chi}{\bar\varphi(\ivar)}| \ge 1 - 2^{-\ell}$, or
		\item $m < \ivar$, and $|\braket{\chi}{0}| \ge 1/4$.
	\end{enumerate}
\end{corollary}
\begin{proof}
	We start with the first case.
	Take $\beta_x$ from \cref{eq:qpe-weights}.
	Assume for now that $m=\ivar$; for increasing $m$ the overlap with $\bar\varphi(\ivar)$ can only increase.
	It is clear that the best $m$ bit approximation to $\varphi'$ less than $\varphi'$ is given by $b=2^m\bar\varphi(\ivar)$ (as the first $\ivar$ digits of both are identical, and $\bar\varphi'_{\ivar+1}=0$ by assumption).
	Then
	\[
		|\beta_0|
		=\frac{1}{2^m} \frac{|1 - \exp(2\pi \ii(2^m\varphi' - b))|}{|1 - \exp(2\pi\ii (\varphi' - b/2^m))|}
		=\frac{1}{2^m} \frac{\sin(\pi 2^m \epsilon)}{\sin(\pi \epsilon)}
		\ge 1 - 2^{-\ell},
	\]
	where $\epsilon=\varphi'-b/2^m$, and the last inequality follows from $\sin(x)/x \ge 1-x$ for $x\in[0,1]$.

	The second claim follows analogously: here again $b=0$, and at most $2^m \varphi' \in [0, 3/4)$;
	the final bound is obtained by applying $x/4 \le \sin(x) \le x$ for $x\in[0,3\pi/4)$, via
	\[
		\frac{1}{2^m}\frac{\sin(\pi 2^m \varphi')}{\sin(\pi \varphi')} \ge \frac14.
		\qedhere
	\]
\end{proof}

\subsection{Solovay-Kitaev Modification to Phase Estimation} \label{Section:Solovay-Kitaev_QTM}
The second step in our QPE analysis is to approximate the small rotation gates that were previously allowed in \cref{cor:full-gateset-qpe}.
We construct a QTM which only uses a standard gate set and $\op U_\varphi$ for some $\varphi=\varphi(\ivar)=2^{-\ivar}$, to run Quantum Phase Estimation (QPE) on $\op U_\varphi$ and output a state which is very close in fidelity to the expansion of $\varphi$ if done without error (i.e. if all gates were exact).

First note that all steps of the QPE procedure as described in \cite{Cubitt_PG_Wolf_Undecidability} can be done exactly up to applying the phase gradient and locating the least significant bit---i.e.\ up until Section 3.6.
However, after this, controlled rotation gates of the form  $\op R_n= 2^{\ii\pi2^{-n}}$, for $1\leq n\leq |\bar{\varphi}|=\ivar$, need to be applied to perform the inverse QFT.
In \cite{Cubitt_PG_Wolf_Undecidability}, this was done by further giving the QTM access to the gate $2^{\ii\pi 2^{-\ivar}}$.
To circumvent this necessity, we approximate small rotation gates using the Solovay-Kitaev algorithm.

\subsubsection{Solovay-Kitaev QTM}
First we introduce the standard statement for the existence of a TM which outputs a high precision approximation to the gate $\op R_n= 2^{\ii\pi2^{-\ivar}}$ using the Solovay-Kitaev algorithm.

\newcommand{\Csk}[1]{c_#1}
\begin{lemma}[SK Machine \cite{Dawson_Nielsen}] \label{Lemma:Gate_SK_Precision}
	There exists a classical TM which, given an integer $k$ and maximum error $\epsilon$, outputs an approximation $\tilde{\op R}_k$ to the gate $\op R_k\in \mathrm{SU}(2)$ such that $\| \tilde{\op R}_k - \op R_k \| < \epsilon$.
    The TM runs in time and space $\bigO(\log^{\Csk1}(1/\epsilon))$ for some $3.97<\Csk1<4$.
\end{lemma}

Part way through the quantum phase estimation procedure, we need to apply the inverse QFT.
However, we do not have access to gates of the form $2^{\ii\pi2^{-\ivar}}$, and our entire QTM will be limited to space $L$.
As a result, whenever the procedure requires a $2^{\ii\pi2^{-\ivar}}$-gate or a power of such a gate, we run the Solovay-Kitaev algorithm to generate an approximation.
As there is $\bigO(\ivar^2)$ many gates to be approximated overall, the procedure will have to be repeated this many times.

However, since we are performing the QPE on a finite length tape, we only have $L$ qubits onto which we can write out the output of the Solovay-Kitaev algorithm; this limits the precision we can achieve using this technique.

Inverting the space bound in \cref{Lemma:Gate_SK_Precision} with respect to the error $\epsilon$, the best approximation obtainable is thus
\begin{equation}\label{eq:sk-best}
\left\| \tilde{\op R}_k - \op R_k \right\| \le \ee^{-\bigO(L^{1/\Csk1})} \le 2^{-\Csk2 L^{1/\Csk1}},
\end{equation}
where we wrote the constant in the exponent as $\Csk2$.
Both Solovay-Kitaev constants $\Csk1$ and $\Csk2$ can be written down explicitly.

\subsubsection{Approximation Error for Output State}
The gates used in the inverse QFT in the previous section were only performed up to a finite precision and hence there will be an error associated with the output state relative to the case with perfect gates.
We will see that the output is then a state that is exponentially close to what would be expected in the case with perfect gates.

Let $\tilde{\op R}_n$ be the approximation to the rotation gate $\op R_n = 2^{\pi \ii 2^{-\ivar} }$ such that $\| \tilde{\op R}_n - \op R_n  \|<\epsilon$, where $\epsilon =  2^{-\Csk2L^{1/\Csk1}}$ is given by the Solovay-Kitaev theorem, \cref{eq:sk-best,Lemma:Gate_SK_Precision}.

    \newcommand{\UQPEapprox}{\tilde{\op U}_\mathrm{QPE}}
    \newcommand{\UQPE}{\op U_\mathrm{QPE}}
    \newcommand{\UQFT}{\op U_\mathrm{QFT}}
    \newcommand{\UQFTapprox}{\tilde{\op U}_\mathrm{QFT}}
    \newcommand{\UPG}{\op U_\mathrm{PG}}
\begin{lemma}\label{Lemma:QTM_Approximation_Error}
	Let $\op U_\mathrm{QPE}$ be the unitary describing the implementation of QPE by a QTM on $m$ qubits with each gate performed exactly.
	Let $\tilde{\op U}_\mathrm{QPE}$ be the unitary describing the same QPE algorithm on $m$ qubits, but where Solovay-Kitaev is used to approximate the rotation gates $\op R_n$ to precision $\epsilon$;
    all other gates are implemented exactly.
    Then the total error of the approximate QPE is
    \begin{align}
	\left\| \UQPEapprox - \UQPE \right\| < \frac{m^2}{2} \epsilon = \frac{m^2}{2}  2^{-\Csk2L^{1/\Csk1}}.
	\end{align}
\end{lemma}
\begin{proof}
    The first part of the phase estimation procedure---the phase gradient operations $\UPG$---can be done exactly in both the approximate and exact cases.
    If QPE is performed to $m$ qudits, we see that there are $m^2/2$ applications of $\op R_n$ gates during the inverse QFT procedure.
    As $\UQPEapprox=\UQFTapprox^\dagger \UPG$, the claim follows from applying the triangle inequality $m^2/2$ times.
\end{proof}

\subsection{Total Quantum Phase Estimation Error}\label{sec:qpe-totalerror}

We have seen previously that there will be errors from both the fact that the parameter $\varphi$ may have a binary expansion longer than the tape length available, and from the Solovay-Kitaev (S-K) algorithm we use to approximate and apply the rotation gates. Here we combine the two errors and upper bound the total deviation introduced.
We continue using $m$ to denote the number of binary digits that $\varphi$ is expanded to, and $L$ is the full tape length.

We emphasize that the two are not necessarily identical, as we can always cordon off a section of the tape to restrict the QPE to only work to within a more limited precision---i.e.\ we can execute the QPE TM on a subsegment of size $m\leq L$ as in \cref{cor:full-gateset-qpe}, and approximate the latter with Solovay-Kitaev that itself can make use of the full tape space available, i.e.\ $L$.
For now we treat $L$ and $m$ as independent quantities, regardless of how they are implemented, and we will choose their specific relation in due course.

\newcommand{\pgood}{p_\mathrm{good}}
\newcommand{\pbad}{p_\mathrm{bad}}
\begin{lemma}\label{lem:approx-QPE-error-good}
Let $\eta\in \field N$ and $\varphi(\ivar)\in\field R$ as in \cref{def:qpe-encoding}, and take $\UQPEapprox$ as the Solovay-Kitaev QPE unitary with output $\ket{\tilde{\chi}}$.
Then either
\begin{enumerate}
\item $m\ge \ivar$ and $|\braket{\tilde\chi}{\bar\varphi(\ivar)}| \ge 1 - \delta(L,m)$, or
\item $m < \ivar$ and $|\braket{\tilde\chi}{0}| \ge 1/2 - \delta(L,m)$.
\end{enumerate}
Here
\[
    \delta(L,m) < \frac{m^2}{2} 2^{-\Csk2 L^{1/\Csk1}}.
\]
\end{lemma}
\begin{proof}
Immediate from \cref{Lemma:QTM_Approximation_Error,eq:sk-best,cor:full-gateset-qpe}.
\end{proof}

As before, we add an approximate variant for the case where $\varphi' \neq \varphi(\ivar)$.
\begin{lemma}\label{lem:approx-QPE-error-other}
		Let $\ivar\in\field N$,  and $\varphi(\ivar)$ as in \cref{def:qpe-encoding}.
	Take a perturbed phase $\varphi' \in [\varphi(\ivar), \varphi(\ivar) + 2^{-\ivar-\ell} )$ for $\ell\in\field N$, $\ell\ge1$, and consider the same setup as in \cref{lem:approx-QPE-error-good}. Then either

	\begin{enumerate}
		\item $m\ge\ivar$ and $|\braket{\tilde\chi}{\bar\varphi(\ivar)}| \ge 1 - 2^{-\ell} - \delta(L,m)$, or
		\item $m < \ivar$, and $|\braket{\tilde\chi}{0}| \ge 1/4 - \delta(L,m)$.
	\end{enumerate}
\end{lemma}
\begin{proof}
	Analogously to \cref{lem:approx-QPE-error-good}, but using \cref{cor:full-gateset-qpe2}.
\end{proof}

The bound in terms of $\delta(L,m)$ is only useful for large $L$, in which case it is easy to see that since $m\leq L$, $\delta\to 0$ for $L\to 0$.
Since we need $\delta$ to be small in due course, we capture a more precise bound in the following remark.
\begin{remark}\label{rem:minimum-L-1}
	For any $\delta_0>0$ there exists an $L_0=L_0(c_1,c_2,\delta_0)$ such that $\delta(L,m)<\delta_0$ for all $L\ge L_0$, where $\delta(L,m)$ is defined in \cref{lem:approx-QPE-error-good}, and $c_1,c_2$ are the Solovay-Kitaev constants from \cref{eq:sk-best}.
\end{remark}
\begin{proof}
	Clear.
\end{proof}


\section{QPE and Universal QTM Hamiltonian}\label{sec:QPE-UTM}
\newcommand\anc{_\mathrm{anc}}
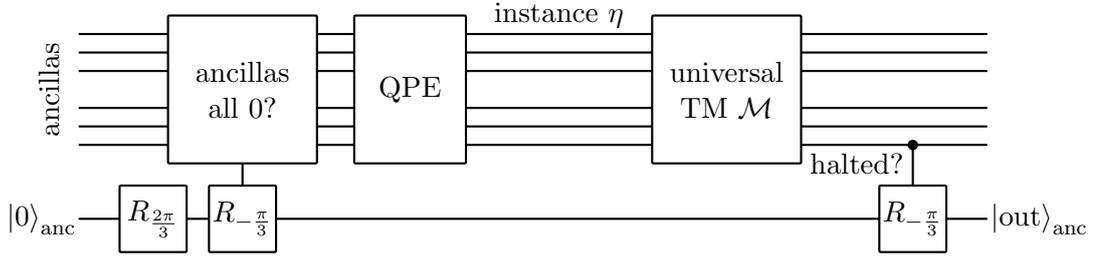
\begin{figure}[t]
	\centering
	\begin{tikzpicture}[
	draw=black,line width=.8pt,
	box/.style = {rounded corners=.5},
	scale=.98
	]
	\newcommand{\cgate}[3]{
		\draw[] (#1,#2) --(#1,#2+1);
		\fill (#1,#2+1) circle [radius=2pt];
		\draw[fill=white,box] (#1-.45,#2-.45) rectangle (#1+.45,#2+.45);
		\node[align=center] at (#1,#2) {#3};
	}
	\draw[] (-.2,0) -- (12,0);
	\foreach\i in {0,1,2,1+0,1+1,1+2}
	{\draw[] (-.2,1+\i/4) -- (12,1+\i/4);}

	\cgate{2}{0}{$\op R_{-\frac\pi3}$}
	\cgate{11}{0}{$\op R_{-\frac\pi3}$}
	\draw[fill=white,box] (1,.75) rectangle (3,2.75);
	\draw[fill=white,box] (3.5,.75) rectangle (5,2.75);
	\draw[fill=white,box] (7.5,.75) rectangle (9.5,2.75);

	\draw[fill=white,box] (.8-.45,-.45) rectangle (.8+.45,.45);
	\node[align=center] at (.8,0) {$\op R_{\frac{2\pi}{3}}$};

	\node[align=center] at (2,1.75) {ancillas\\all $0$?};
	\node[align=center] at (4.25,1.75) {QPE};
	\node[align=center] at (8.5,1.75) {universal\\TM $\mathcal M$};
	\node[align=center] at (10.25,.75) {halted?};
	\node[align=center] at (6.25,2.75) {instance $\ivar$};
	\node[align=right] at (-.7, 0) {$\ket 0\anc$};
	\node[align=right,rotate=90] at (-.6,1.75) {ancillas};
	\node[align=left] at (12.7, 0) {$\ket{\text{out}}\anc$};
	\end{tikzpicture}
	\caption{QPE and universal TM circuit.
		The construction uses one flag ancilla $\ket 0\anc$ to verify that as many ancillas as necessary for the successive computation are correctly-initialized ancillas (e.g.\ $\ket0$), and if not rotating the single guaranteed $\ket 0\anc$ flag  by $\pi/3$.
		On some ancillas, the problem instance $l$ is written out.
		Another rotation by $\pi/3$ is applied depending on whether the dovetailed universal TM $\mathcal M$ halts on $\ivar$ or not within the number of steps allowed by the clock driving its execution, which in turn is limited by the tape length.
	}
	\label{fig:augmented-verifier}
\end{figure}

In this section we examine how to encode the quantum Turing machine performing quantum phase estimation described in \cref{sec:modified-qpe} (and briefly overviewed in section \textcolor{red}{4.2} of the main article) into a Hamiltonian on a spin chain of length $L$,
such that the ground state energy of the Hamiltonian is non-negative if and only if a dovetailed universal Turing machine $\mathcal M$ halts on input $\varphi(\ivar)$ and within tape length $L$.
We further prove that this ground state energy remains non-negative (or negative, respectively) if instead of $\varphi(\ivar)$ we are given a slightly perturbed phase $\varphi' \in [ \varphi(\ivar), \varphi(\ivar) + 2^{-\ivar-\ell} )$, given $\ell\ge1$ is large enough.
This section makes rigorous the results at the end of section \textcolor{red}{4.2} of the main article.

We note that the circuit-to-Hamiltonin mapping used in this work will be a variation of that in \cite{Cubitt_PG_Wolf_Undecidability}, which is itself a modification of the construction in \cite{Gottesman-Irani}. 
\cite{Gottesman-Irani} is a particularly important work to to the field of Hamiltonian complexity as it demonstrates how a computation can be encoded in a 1D, nearest neighbour, translationally invariant Hamiltonian.
This is particularily suprising as it shows the simplest class of Hamiltonians (in terms of dimensionality and interaction type) have hard to compute low-energy properties. 
Furthermore, classically this class of Hamiltonians is computationally tractable.

With the aim of proving the above statements about the ground state energy, we first amend the computation slightly.
In \cite{Cubitt_PG_Wolf_Undecidability}, the authors used \citeauthor{Gottesman-Irani}'s history state construction for a Turing machine with an initially empty tape \cite{Gottesman-Irani}.
To ensure a correctly initialised tape, the authors use an initialization sweep; essentially a single sweep over the entire tape with a special head symbol, under which one can penalize a tape in the wrong state.

Instead of using an initialization sweep, we make do with a single ancilla (denoted with subscript ``anc'' in the following) which is initialized to $\ket 0\anc$, and verify on a circuit level that all the other ancillas are correctly initialized.
In order to achieve this, we first execute a single $\op R_{2\pi/3}$ rotation on $\ket0\anc$ to initialize it to a $\op R_{2\pi/3}\ket0\anc$-rotated state.
Next, we execute a controlled $\op R_{-\pi/3}$ rotation in the opposite direction on $\ket\phi\anc$, where the controls are on all the ancillas we wish to ensure are in the right state.
If and only if \emph{all} of the controlling ancillas are in state $\ket 1$---which we can check e.g.\ with a multi-anticontrolled operation---will we perform a rotation by $\op R_{-\pi/3}$.
After the controlled rotation, we apply $\op X$ flips to all the ancillas we wish to initialize to $\ket 0$.

This ancilla will carry another role: in case the dovetailed universal TM $\mathcal M$ from \cref{sec:modified-qpe} halts, we transition to a finalisation routine that performs another $\op R_{-\pi/3}$ rotation on it.
The net effect of this circuit is that, after the entire computation ends, the ancilla is in state
 $\ket{\text{out}}\anc$ with overlap
\begin{equation}\label{eq:single-ancilla-effect}
	\braket{1}{\text{out}}\anc = \begin{cases}
		0 & \text{if all ancillas are correctly initialized \emph{and} $\mathcal M$ halted, or} \\
		\frac{\sqrt 3}{2} & \text{otherwise.}
	\end{cases}
\end{equation}
This idea in the context of circuit-to-Hamiltonian mappings was introduced in \cite{Bausch2016}; for completeness we give an overall circuit diagram of the entire computation to be mapped to a Hamiltonian in \cref{fig:augmented-verifier}.
We remark that breaking down a multi-controlled quantum gate into a local gate set is a standard procedure described e.g.\ in \cite{Nielsen_and_Chuang}.

We formalise the above procedure in the following lemma:
\begin{lemma}\label{Lemma:Output_Amplitude}
	Consider an initial state
	\[
		\ket{\psi_0} = \ket{0}\anc \left( \alpha\ket{1}^{\otimes L} + \sqrt{1-\alpha^2}\ket{\phi} \right)
		\quad\text{where}\quad
		\ket{\phi} \perp \ket1^{\ox L}.
	\]
	Assume the Turing machine $\UTM$ halts with probability $\epsilon$ when acting on an initial state $\ket{0}\anc\ket{1}^{\otimes L}$.
	Then, the final output state of the computation $\ket{\psi_T}$ satisfies
	\[
	\bra{\psi_T} \left[ \ketbra{1}\anc \otimes \1^{\otimes L} \right] \ket{\psi_T} = \frac 34 \left(1-\alpha^2\epsilon^2\right).
	\]
\end{lemma}
\begin{proof}
	By explicit calculation, we have
	\begin{align*}
	\ket{\psi_0}  & \xmapsto{\op R_{2\pi/3}} \op R_{2\pi/3}\ket{0}\anc \left( \alpha\ket{1}^{\otimes L} + \sqrt{1-\alpha^2}\ket{\phi} \right) \\
	&\xmapsto{\op {cR}_{-\pi/3}} \alpha \op R_{\pi/3}\ket{0}\anc\ket{1}^{\otimes L} + \sqrt{1-\alpha^2} \op R_{2\pi/3}\ket{0}\ket{\phi} \\
	&\xrightarrow{\UTM} \alpha \op R_{\pi/3} \ket{0}\anc \left( \epsilon\ket{\psi_\mathrm{halt}} + \sqrt{1-\epsilon^2}\ket{\psi_\mathrm{non-halt}}  \right)  \\
	& \hspace{.5cm} + \sqrt{1-\alpha^2} \op R_{2\pi/3}\ket{0}\anc \left(\epsilon'\ket{\phi_\mathrm{halt}} + \sqrt{1-\epsilon'^2}\ket{\phi_\mathrm{non-halt}}\right) \\
	&\xrightarrow{\op{cR}_{-\pi/3}} \alpha \epsilon\ket{0}\ket{\psi_\mathrm{halt}} + \alpha \sqrt{1-\epsilon^2}\op R_{\pi/3}\ket{0}\ket{\psi_\mathrm{non-halt}}   \\
	& \hspace{.5cm} + \epsilon' \sqrt{1-\alpha^2} \op R_{\pi/3}\ket{0}\anc \ket{\phi_\mathrm{halt}} + \sqrt{1-\epsilon'^2}\sqrt{1-\alpha^2}  \op R_{2\pi/3}\ket{0}\anc \ket{\phi_\mathrm{non-halt}} \\
	&= \ket{\psi_T}.
	\end{align*}
	Using
	\[
	\ket 0 \xmapsto{\op R_{\pi/3}} \cos\left(\frac\pi3\right)\ket0 + \sin\left(\frac\pi3\right)\ket1
	\quad\text{and}\quad
	\ket 0 \xmapsto{\op R_{2\pi/3}} \cos\left(\frac\pi3\right)\ket0 - \sin\left(\frac\pi3\right)\ket1
	\]
	this means that
	\[
	\bra{\psi_T} \left[ \ketbra{1}\anc \otimes \1^{\otimes L} \right] \ket{\psi_T} = \sin^2\left(\frac{\pi}{3}\right)\left( 1 - \alpha^2 \epsilon^2 \right).
	\qedhere
	\]
\end{proof}
\subsection{Feynman-Kitaev Hamiltonian}
Given our quantum Turing machine from \cref{sec:modified-qpe} augmented with a single necessary ``good'' ancilla $\ket0\anc$ as just described, we apply the \citeauthor{Gottesman-Irani} construction from \cite{Gottesman-Irani} to translate our desired computation in the ground state of a one-dimensional, nearest neighbour, translationally invariant Hamiltonian with open boundary conditions.
We summarize the core ideas to set up the notation used in this section, but refer the reader to \cite{Gottesman-Irani,Cubitt_PG_Wolf_Undecidability,Bausch_Crosson} for details.

\begin{definition}[History state]\label{def:history-state}
	A \emph{history state} $\ket{\Psi} \in \HS_{\mathrm C}\ox\HS_{\mathrm Q}$ is a quantum state of the form
	\begin{equation}
	\ket{\Psi} = \frac{1}{\sqrt{T}} \sum_{t=1}^{T}\ket{t}_{\mathrm C}\ket{\psi_t}_{\mathrm Q},
	\end{equation}
	where $\{\ket{1}, \ldots, \ket T\}$ is an orthonormal basis for $\HS_{\mathrm C}$, and $\ket{\psi_t} = \prod_{i=1}^t \op U_i\ket{\psi_0}$ for some initial state $\ket{\psi_0}\in\HS_{\mathrm Q}$ and set of unitaries $\op U_i\in\mathcal{B}(\HS_{\mathrm Q})$.

	$\HS_{\mathrm C}$ is called the \emph{clock register} and $\HS_{\mathrm Q}$ is called the \emph{computational register}. If $\op U_t$ is the unitary transformation corresponding the $t$\textsuperscript{th} step of a quantum computation---which in our case is not a gate in the circuit model, but a QTM transition---then $\ket{\psi_t}$ is the state of the computation after $t$ steps.
	We say that the history state $\ket{\Psi}$ \emph{encodes} the evolution of the quantum computation.
\end{definition}

As discussed in \cref{sec:qpe-totalerror}, the QPE Turing machine we devised has two meta parameters $L$ and $m$. On a spin chain of length $L$, instead of expanding $L-3$ digits of $\varphi$ as is the case in \cite{Cubitt_PG_Wolf_Undecidability}, we allow the expansion to happen on a smaller sub-segment of length $m$ of the chain.
This can be done dynamically, i.e.\ by adding a Turing machine before the QPE invocation which sections off a part $m=m(L)$ of the tape and places a distinct symbol $\midend$ there.
Since it is obvious how to do this we will not go into detail here, and remark that in the final construction we will choose $m=L-3$:
an explicit construction for such a Turing machine is given in\cite[Lem.\ 15]{Bausch_1D_Undecidable}.
The QPE and dovetailed universal TM---augmented by the single-ancilla construction described at the start of this section---we will jointly call $\mathcal M'=\mathcal M'(L,m)$, i.e.\ such that there is $L$ tape available;
we emphasize that $\mathcal M'(L,m)$ has an identical set of symbols and internal states for all $L$ and $m$.

In all of the following we will analyse the spectrum of the history state Hamiltonian within a ``good'' type of subspace, by which we mean a tape bounded by special endpoint states $\leftend$ and $\rightend$.
This subspace will, analogous to the 2D undecidability construction, be called \emph{bracketed} states; on an overall local Hilbert space $\HS = \HS_a \oplus \HS_b$ such that $\ket*{\leftend},\ket*{\rightend}\in\HS_b$, we set
\newcommand\Sbr{\mathcal S_\mathrm{br}}
	\begin{equation}\label{eq:Sbr}
		\Sbr(m) := \ket*{\leftend} \otimes \HS_a^{\otimes L} \otimes \ket*{\rightend}.
	\end{equation}
Since no transition rule for the history state Hamiltonian ever moves these boundary markers, the overall Hamiltonian we construct will be block-diagonal with respect to signatures determined by the brackets.
A standard argument then shows that within this bracketed subspace, the history state Hamiltonian encoding the QPE Turing machine behaves as designed, and we can analyse the spectrum therein by analysing the encoded computation.
Outside of the bracketed subspace, a variant of the Clairvoyance lemma allows us to always lower-bound the energy, such that it does not interfere with the rest of the construction.

In order to make all of this precise, we first define the full QPE history state Hamiltonian in the following theorem, which is adapted from \cite[Th.~10]{Cubitt_PG_Wolf_Undecidability}.
\newcommand{\HTM}{\op H_\mathrm{QTM}}
\begin{theorem}[QPE history state Hamiltonian]\label{Theorem:QTM_in_local_Hamiltonian}
	Let $L,m\in\field N$, $0 < m \le L-3$.
	Let there exist a Hermitian operator $\op h\in \mathcal B(\field C^d \otimes \field C^d)$, where the local Hilbert space contains special marker states $\ket*{\leftend}$ and $\ket*{\rightend}$ that define the bracketed subspace $\Sbr$ as in \cref{eq:Sbr}, such that
	\begin{enumerate}
		\item $\op h \ge 0$,
		\item $d$ depends (at most polynomially) on the alphabet size and number of internal states of $\mathcal M'$,
		\item \label{QTM_in_Local_Hamiltonian:Matrix_Elements}
		$\op h = \op A + \ee^{\ii\pi\varphi(\ivar)}\op B + \ee^{-\ii\pi\varphi(\ivar)}\op B^\dagger $, where
			\begin{itemize}
				\item $\op B \in \mathcal{B}(\C^d\ox\C^d)$ independent of $\ivar$ and with coefficients in $\field{Z}$, and
				\item $\op A\in \mathcal{B}(\C^d\ox\C^d)$ is Hermitian, independent of $\ivar$, and with coefficients in $\field{Z}+\field{Z}/\sqrt 2$;
			\end{itemize}
	\end{enumerate}
	Furthermore, a spin chain of length $L$ with local dimension $d$, the translationally-invariant nearest-neighbour Hamiltonian $\HTM(L) := \sum_{i=1}^{L-1} \op h^{(i,i+1)}$ has the following properties.
	\begin{enumerate}\addtocounter{enumi}{3}
		\item $\HTM(L)$ is frustration-free, and
		\item the unique ground state of $\HTM(L)|_{\Sbr(m)}$ is a computational history state as in \cref{def:history-state} encoding the evolution of $\mathcal M'(L,m)$.
	\end{enumerate}
	The history state satisfies
	\begin{enumerate}\addtocounter{enumi}{5}
		\item $T=\Omega(\poly(L)2^L)$ time-steps, in either the halting or non-halting case;
		\item If $\mathcal M'$ runs out of tape within a time $T$ less than the number of possible TM steps allowed by the history state clock, the computational history state only encodes the evolution of $\mathcal M'$ up to time $T$.
		\item In either the halting or non-halting case, the remaining time steps of the evolution encoded in the history state leave the computational tape for $\mathcal M'$ unaltered, and instead the QTM runs an arbitary computation on a waste tape as described in \cite{Cubitt_PG_Wolf_Undecidability}.
	\end{enumerate}
\end{theorem}
\begin{proof}
	Almost all of the above follows from \cite[Th.\ 10]{Cubitt_PG_Wolf_Undecidability}.
	\Cref{QTM_in_Local_Hamiltonian:Matrix_Elements} differs only in that we have removed any dependence on $e^{\ii\pi2^{-|\varphi|}}$ due to the new modified transition rules, as we now approximate the necessary rotations using the Solovay-Kitaev theorem (see \cref{sec:modified-qpe}).
\end{proof}
\subsubsection{Clock Construction}\label{Sec:Clock_Construction}
The history state Hamiltonian described above encodes an evolution of a computation for $T(L)$ steps, where $T(L)$ \emph{does not} depend on the computation itself.
This ensures that the history state will be a superposition over $T(L)$ time steps independent on whether $\mathcal M'$ halts on the tape of length $L-2$ and with cordonned-off subsection $m$.
As mentioned previously, in the case of the computation halting, this is done by forcing the QTM head to switch to an additional ``waste tape'' where an arbitrary computation is performed until the clock finishes.


\Cref{Theorem:QTM_in_local_Hamiltonian} uses the clock construction designed in \cite[sec.\ 4.2,\ 4.3,\ 4.4]{Cubitt_PG_Wolf_Undecidability}.
Bounds on the clock runtime are readily obtained: if $T(L)$ denotes the runtime of the clock on a tape of length $L$, we have
\begin{align}\label{eq:clock-runtime}
\Omega\left(L\xi^L \right)  \leq T(L) \leq \BigO\left(L\xi^L \log(L)\right)
\end{align}
for some constant $\xi\in\field N$.

\subsubsection{QTM and Clock Combined}
\Cref{Theorem:QTM_in_local_Hamiltonian} combines the QTM and clock such that the QTM head only makes a transition when the oscillator from the clock part of the history state passes over the QTM head.
Details can be found in \cite[sec.\ 4.6.1]{Cubitt_PG_Wolf_Undecidability}.

\subsection{The Initialisation and Non-Halting Penalty}
We now want to introduce a penalty term which will penalise computations that have not halted and not been initialised correctly.

\paragraph{Initialisation Penalty.}
In order to ensure that the single ancilla we require is correctly initialized, we introduce a projector that penalizes $\ket\psi\anc$ in any state but $\ket0\anc$ at the start of the computation.
This can be done by a term of the form $\ketbra{0}_\mathrm{C} \otimes (\1 - \ketbra 1)\anc$, which is local if and only if we can locally detect the initial clock state $\ket0_\mathrm{C}$ above the single ancilla on the tape.
As per the constructions in \cite{Gottesman-Irani,Cubitt_PG_Wolf_Undecidability}, this state can indeed be locally detected.

\paragraph{Finalisation Penalty.}
The final penalty follows precisely the same pattern: we add a local projector of the form $\ketbra T_\mathrm{C} \otimes (\1 - \ketbra 1)\anc$, and ensure that the final clock state $\ket T_\mathrm{C}$ can be recognized locally above where $\ket{\text{out}}\anc$ sits.
To realise this, we note that the ancilla bit is located at the end of string of qudits encoding the TM tape.
The final clock state can then be locally determined by a nearest-neighbour, translationally invariant term that recognises the final clock state by looking at the pair of qudits at the end of the chain.
Again, this is done in  \cite{Gottesman-Irani,Cubitt_PG_Wolf_Undecidability}.

\paragraph{Penalty Term Construction.}
The amplitude of the output ancilla $\ket\psi\anc$ depends on correct initialization of the ancillas for the QTM, as well as on the halting amplitude, and is given in \cref{eq:single-ancilla-effect}.
To penalize the overlap $\braket{1}{\psi}\anc$---which corresponds to wrong initialization, or halting---we add the following nearest neighbour term to the Hamiltonian:
\newcommand{\out}{\mathrm{(out)}}
\newcommand{\pin}{\mathrm{(in)}}
\begin{align*}
\op h_{i,i+1}^\out= \ketbra{[\blacksquare][\arrLone,\dots, \xi]}_{i,i+1} \otimes (\1_i-\ketbra{1}_{i})\otimes \1_{i+1}.
\end{align*}
As just mentioned, the input penalty term $\op h_{i,i+1}^\pin$ can similarly be written as a nearest-neighbour projector onto a clock state at $t=0$.
Thus, on the entire chain we have the penalty terms
\newcommand{\pen}{\mathrm{(pen)}}
\begin{align}\label{eq:inoutpen}
\op H^\pin &=  \sum_{i=1}^{L-1} \op h_{i,i+1}^\pin  \\
\op H^\out&= \sum_{i=1}^{L-1} \op h_{i,i+1}^\out.
\end{align}

\begin{definition}\label{def:UTM-Ham}
We denote the QPE+QTM history state Hamiltonian including the in- and output penalties from \cref{eq:inoutpen} with $\HUTM(L,\varphi):=\HTM(L,\varphi)+\op H^\pin+\op H^\out$.
\end{definition}

\subsection{Ground State Energy in Halting and Non-Halting Case}
The ground state energy of $\HUTM$ depends on how much penalty is picked up throughout the computation.
Known techniques like Kitaev's geometrical lemma \cite{Kitaev2002,Bausch2016} for a lower bound and a simple triangle inequality for the upper bound can be used to show that
\begin{equation}
\Omega\left(\frac{1}{T^3}\right)  \leq \lmin(\HUTM(\varphi(\ivar))) \leq \BigO\left(\frac{1}{T}\right)
\end{equation}
for a non-halting instance $\ivar\in\field N$.
However, both the upper and lower bounds here are not tight enough for our purposes.

In order to obtain tighter bounds, we realize that our history state construction has a linear clock (i.e.\ one that never branches and simply runs from $t=0$ to $t=T$);
in this case, tight bounds on the overall energy effect of the penalty terms already exist; we refer the reader to \cite{Bausch_Crosson,Caha2017, Watson_Hamiltonian_Analysis} for an extended analysis.
For the sake of completeness and brevity, we quote some of the definitions and lemmas from prior literature in the appendix and reference them in the following.

\begin{lemma}\label{Lemma:GS_Lower_Bound}
	In case $\ivar \in \field N$ correspond to a non-halting instance, the lowest eigenvalue of $\HUTM$ satisfies $\lmin(\HUTM(\varphi(\ivar))) = \Omega(T^{-2})$.
\end{lemma}
\begin{proof}
	In \cref{Lemma:HTM_Standard_Form}, we prove $\HUTM$ is a standard-form Hamiltonian as per \cref{Def:Standard-form_H}, and so as per the Clairvoyance Lemma \cite[Lem.~5.6]{Watson_Hamiltonian_Analysis} we know that $\HUTM$ breaks down into three subspaces.
	The subspaces of types 1 and 2 are trivially shown to have ground state energies $\Omega(T^{-2})$.

	Within the third subspace, which we label $S$, there are no illegal terms and only the in- and output penalties $\op H^\pin + \op H^\out$ from \cref{eq:inoutpen} have to be considered.
	By \cref{Lemma:Clairvoyance} the clock evolution within this subspace is linear---meaning there is never any branching---and hence $\HUTM|_S$ is equivalent to Kitaev's original circuit-to-Hamiltonian construction.
	This means that the Hamiltonian therein is of the form
	\begin{align*}
	\HUTM|_S = \op H_\mathrm{prop} + \ketbra 0_\mathrm{C} \otimes \Pi^\pin + \ketbra T_\mathrm{C} \otimes \Pi^\out
	\end{align*}
	where $\op H_\mathrm{prop} \sim \Delta \ox \1$ for a path graph Laplacian $\Delta$, and $\Pi^{\pin/\out}$ are the in- and output penalties inflicted at time $0$ and $T$; this Hamiltonian is then explicitly of the family of Hamiltonians studied in \cite{Bausch_Crosson}.
	In particular, by \cite[Th.~7]{Bausch_Crosson}, Hamiltonians of this form have ground state energy $\lmin(\HUTM|_S)=\Omega(T^{-2})$.
	Thus all three of the subspaces have a minimum eigenvalue of the form $\Omega(T^{-2})$, and since they are invariant subspaces, we see that the overall minimum eigenvalue must be $\lmin(\HUTM) = \Omega(T^{-2})$.
\end{proof}

\begin{lemma}[Theorem 6.1 from \cite{Watson_Hamiltonian_Analysis}] \label{Lemma:GS_Upper_Bound}
      Let $\op H(\varphi) \in \mathcal{B}(\field{C}^d)^{\otimes L}$ be a standard form Hamiltonian encoding a QTM with runtime $T(L)$, with in- and output penalty terms $\op H^{\pin/\out}$.
      Let there exist a computational path with no illegal states such that the final state of the computation is $\ket{\psi_T}$ and such that the output penalty term satisfies
      \begin{align*}
      \bra{T}\bra{\psi_T}\op H^\out\ket{\psi_T}\ket{T} \leq \epsilon.
      \end{align*}
      Then the ground state energy is bounded by
      \begin{align*}
      0\leq \lmin\big( \op H(\varphi)\big) \leq \epsilon \left(1-\cos\left( \frac{\pi}{2(T-T_\mathrm{init})+1} \right) \right) = \BigO\left( \frac{\epsilon}{T^2} \right),
      \end{align*}
      where $T_\mathrm{init}=\BigO(\log(T))$ is the time frame within which the input penalty term $\op H^\pin$ applies to the history state.
\end{lemma}

With this machinery developed, we can derive the following lemma for the specific Hamiltonian $\HUTM$ at hand.
\begin{theorem}\label{th:HQTM-spectrum}
    Take $\HUTM$ to encode a phase $\varphi' \in [\varphi(\ivar), \varphi(\ivar) + 2^{-\ell})$, with $\varphi(\ivar)$, as per \cref{def:qpe-encoding}, and let $\delta(L,m)$ be as in \cref{lem:approx-QPE-error-good}.
    Then for
    \begin{enumerate}
    	\item $m< \ivar$ we have
    	\[
	    \lmin(\HUTM)=\Omega\left[T^{-2}\right].
    	\] \label{HQTM_Spectrum_1}
    	\item \label{HQTM_Spectrum_2}  $m\geq \ivar$ and $\varphi(\ivar)$ corresponds to a non-halting instance, then
    	\[
    	\lmin(\HUTM)=\Omega\left[T^{-2}\right].
    	\]
    	\item  \label{HQTM_Spectrum_3}$m\geq \ivar$ and $\varphi(\ivar)$ corresponds to a halting instance, then
    	\[
        \lmin(\HUTM)=\BigO\left[\left(2^{-\ell} + \delta(L,m)\right)^2\frac{1}{T^2}\right].
        \]
    \end{enumerate}
\end{theorem}
\begin{proof}
	Combing \cref{lem:approx-QPE-error-other} with \cref{Lemma:Output_Amplitude} we derive upper and lower bounds on the magnitude of the amplitudes that a given instance has on a non-halting state.
	Together with \cref{Lemma:GS_Lower_Bound} this gives us the lower bounds for points \labelcref{HQTM_Spectrum_1,HQTM_Spectrum_2}.
	To get the upper bound in \labelcref{HQTM_Spectrum_3}, by \cref{lem:approx-QPE-error-other,eq:single-ancilla-effect}, the output penalty is bounded as
	\[
	\bra{\psi_T}\Pi^\out\ket{\psi_T}\leq \sin^2\left(\frac\pi3 \right) \left(  2^{-\ell} + \delta(L,m)  \right)^2.
	\]
	Since no other term contributes a positive energy, the ground state of $\HUTM$ can be upper-bounded with \cref{Lemma:GS_Upper_Bound} as
	\[
	\lmin(\HUTM) = \BigO\left( \frac{\left(2^{-\ell} + \delta(L,m)\right)^2}{T^2} \right).
	\qedhere
	\]
\end{proof}


\section{Checkerboard and TM Tiling} \label{Sec:Lattice_Tiling}

In the previous section we saw that the energy penalty from encoding a halting/non-halting computation in a Gottesman-Irani type Hamiltonian decreases as $\sim 1/T^2$ (see \cref{th:HQTM-spectrum}).
To find a way to boost this energy penalty such that it is no-longer dependent on $T$, we combine the Gottesman-Irani Hamiltonian with a classical Hamiltonian.
Later we will partition the local Hilbert space into a classical part and quantum part $\HS_c \ox \HS_q$, and choose interactions between the two parts of the Hilbert space to such that ground states of the Gottesman-Irani Hamiltonian are forced to be present at certain points in the tiling.
We also encode a classical TM in the classical Hamiltonian, which will later be used to correct for errors in the approximate quantum phase estimation encoded in the Gottesman-Irani Hamiltonian.
 
\subsection{Tiling to Hamiltonian Mapping}
Given a fixed set of Wang tiles on a 2D lattice, we can map the corresponding tiling pattern to a classical translationally invariant, nearest neighbour Hamiltonian over spins on the same lattice.
This is used to great effect in \cite{Gottesman-Irani, Cubitt_PG_Wolf_Undecidability} and shown rigorously in lemma 1 and corollary 2 of \cite{Bausch2015}, where the authors also explain how to allow weighted tile sets.
In the latter it is explained how to favour a certain tile by giving a bonus to it, or by giving an especially-strong penalty to a specific combination of tiles.
For the sake of completeness, we will summarize the essence of the result below.

Let $\mathcal{T}$ be a set of Wang tiles.
For an edge $e$ in the interaction graph denote with $K\subset\mathcal T\times\mathcal T$ the subset of tiles that are allowed to be placed next to each other along edge $e$, and a function $w:K\longrightarrow \field R$ assigns a weight to a neighbouring tile pair.
Then the corresponding local term is simply a weighted projector
\begin{align} \label{Eq:TILING_Hamiltonian}
\op h^e=  \!\!\!\!\!\!\!\!  \sum_{(t_1,t_2)=t\in K}  \!\!\!\! \big(\1 - w(t)\ketbra{t_1}\otimes\ketbra{t_2} \big).
\end{align}
The overall Hamiltonian will then be a sum of these terms, i.e.\ $\op H=\sum_e \op h^e$, and its ground state is the highest-score tiling possible on the interaction graph.
In the most simple case where $w\equiv1$ this simply means that the ground state will have zero energy if there exists a tiling without a single mismatch anywhere that tiles the lattice.
Its degeneracy will depend on how many tiling choices without any mismatching edges are possible.

It is obvious that when the original tiling constraints on the interaction graph were translationally-invariant, then so is the constructed Hamiltonian; furthermore, the local dimension of that Hamiltonian will equal the number of tiles that we need to allow per site.

In case we need to have more than one tile set on the same lattice, we can simply introduce lattice layers:
\begin{remark}[Tiling Layers \cite{Gottesman-Irani}]\label{rem:tiling-layers}
For multiple tile sets $\mathcal T_1,\ldots, \mathcal T_\ell$, there exists a meta tileset\ \,$\mathcal T$ with a set of meta-tiling rules, such that the meta-tiling rules are only satisfied iff the tiling rule for each element of the tuple is satisfied.
The corresponding Hamiltonian is defined on the tensor product of the individual Hilbert spaces.
Tile constraints may also be placed between layers.
\end{remark}
\begin{proof}
Given a lattice, we represent the meta-tile set as an $\ell$-tuple associated with each site.
Each element represents a layer in the tiling.
Tiling rules for the $k$\textsuperscript{th} layer are enforced between the $k$\textsuperscript{th} elements of tuples on neighbouring sites.
Tiling rules between layers can be prevented from occurring by disallowing certain tuples from appearing.
\end{proof}

\subsection{Checkerboard Tiling}
In this section, we define a tile set that periodically tiles the infinite plane.
The underlying pattern we wish to create is that of a square lattice, where each grid cell within the pattern has the same side length, much like the boundaries on a checkerboard.
The tiling will not be unique; in fact, there will be a countably infinite number of variants of the tiling which satisfy the tiling rules, corresponding to the pattern's periodicity.
This non-uniquess is intended: the corresponding tiling Hamiltonian will have a degenerate ground state, the interplay of the other Hamiltonians' energy eigenstates that are conditioned on this underlying lattice pattern will then single out a unique ground state.

We constructively define this checkerboard tiling in this section.
In order to explain and prove rigorously how the highest net-bonus tiling is achieved, we break the proof up into two parts; in the first part, we will create a checkerboard pattern of various square sizes, but such that the offset from the lower left corner in the lattice is left unconstrained.
In the second part, we will lift this degeneracy.

\begin{proposition}[Unconstrained Checkerboard Tiling]\label{prop:unconstrained-tiling}
	We define the tileset $\mathcal T_1$ to contain the following edge-colored tiles:
	\begin{align*}
	\mathrm{corner\ \&\ edge:}&\quad
	\Ts(R,B,*R,*B)
	\quad
	\Ts(*R,0,R,1)
	\quad
	\Ts(R,0,*R,1)
	\quad
	\Ts(0,*B,1,B)
	\quad
	\Ts(0,B,1,*B)
	\\
	\mathrm{interior:}&\quad
	\TTs(1,1,D,0,T)
	\quad
	\TTs(1,1,0,D,T)
	\quad
	\TTs(1,1,D,D,T)
	\quad
	\TTs(D,D,0,0,T)
	\quad
	\TTs(0,0,0,0,T)
	\quad
	\TTs(1,1,1,1,T)
	\end{align*}
	The rules for these tiles---by convention---are such that edges have to match up.
	Then all valid tilings for a lattice $\Lambda$ will either:
	\begin{enumerate}
		\item have no corner tile present, or
		\item have corner tiles present as shown in \cref{fig:checkerboard-tiling}, i.e.\ such that they are part of a checkerboard pattern of squares, where the squares' side length---and the offset of the left- and bottommost corner tile---is unconstrained.
	\end{enumerate}
\end{proposition}
\begin{proof}
	\Cref{fig:checkerboard-tiling} forms a valid tiling by inspection.
	What is left to prove is that given we demand at least one corner tile to be present this is the only tiling pattern possible.

	To this end, we first note that the tiles directly adjacent to the corner tile are necessarily of the following configuration:
	\[
		\begin{tiles}
			\TT(1,1,0,D,T)
			\T(*R,0,R,1)
			\~
			\$
			\T(0,*B,1,B)
			\T(R,B,*R,*B)
			\T(0,*B,1,B)
			\$
			\~
			\T(*R,0,R,1)
			\TT(1,1,D,0,T)
		\end{tiles}
	\]
	We then note that the only way for multiple of these corner tiles to join up is via blue horizontal links (called configuration $A_2$), red vertical ones (configuration $A_3$), or diagonal purple ones (configuration $A_1$);
	we show sections of these links $A_1,A_2$ and $A_3$ in \cref{fig:primitive-tilings-4}.

	\begin{figure}[!htb]
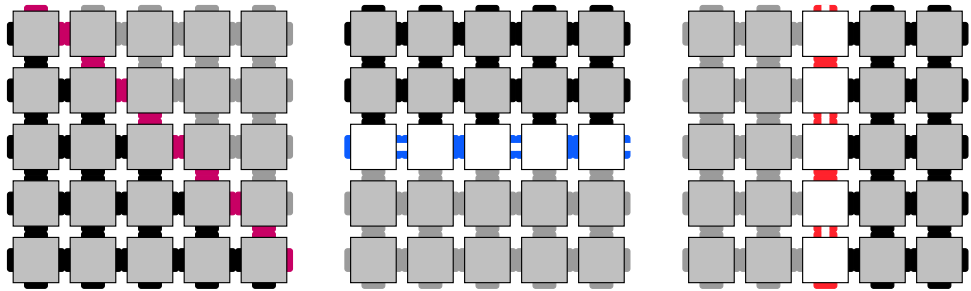

		\newcommand{\tilelineA}[1]{%
			\TT(#1,#1,#1,#1,T)
			\TT(#1,#1,#1,#1,T)
			\TT(#1,#1,#1,#1,T)
			\TT(#1,#1,#1,#1,T)
			\TT(#1,#1,#1,#1,T)
		}
		\newcommand{\tilelineB}[1]{%
			\TT(1,1,1,1,T)
			\TT(1,1,1,1,T)
			\T(#1)
			\TT(0,0,0,0,T)
			\TT(0,0,0,0,T)
		}
		\begin{tiles}
			\TT(D,D,0,0,T)
			\TT(1,1,D,D,T)
			\TT(1,1,1,1,T)
			\TT(1,1,1,1,T)
			\TT(1,1,1,1,T)
			\$
			\TT(0,0,0,0,T)
			\TT(D,D,0,0,T)
			\TT(1,1,D,D,T)
			\TT(1,1,1,1,T)
			\TT(1,1,1,1,T)
			\$
			\TT(0,0,0,0,T)
			\TT(0,0,0,0,T)
			\TT(D,D,0,0,T)
			\TT(1,1,D,D,T)
			\TT(1,1,1,1,T)
			\$
			\TT(0,0,0,0,T)
			\TT(0,0,0,0,T)
			\TT(0,0,0,0,T)
			\TT(D,D,0,0,T)
			\TT(1,1,D,D,T)
			\$
			\TT(0,0,0,0,T)
			\TT(0,0,0,0,T)
			\TT(0,0,0,0,T)
			\TT(0,0,0,0,T)
			\TT(D,D,0,0,T)
		\end{tiles}
		\hfill
		\begin{tiles}
			\tilelineA0
			\$
			\tilelineA0
			\$
			\T(0,*B,1,B)
			\T(0,B,1,*B)
			\T(0,*B,1,B)
			\T(0,B,1,*B)
			\T(0,*B,1,B)
			\$
			\tilelineA1
			\$
			\tilelineA1
		\end{tiles}
		\hfill
		\begin{tiles}
			\tilelineB{*R,0,R,1}
			\$
			\tilelineB{R,0,*R,1}
			\$
			\tilelineB{*R,0,R,1}
			\$
			\tilelineB{R,0,*R,1}
			\$
			\tilelineB{*R,0,R,1}
		\end{tiles}
		\caption{Sub-tiling patterns $A_1,A_2$ and $A_3$ from left to right, possible with at most four tiles in \cref{prop:unconstrained-tiling}.}
		\label{fig:primitive-tilings-4}
	\end{figure}

	This reduces the problem to finding valid geometric patterns of horizontal blue, vertical red and diagonal purple lines, which are only ever allowed to intersect jointly together;
	the resulting pattern is a grid of squares laid out by the red and blue edge tiles, where the fact that each enclosed area is a square is enforced by the purple diagonals.
	If the square size is bigger than the lattice $\Lambda$, this means that only a single corner tile is present; otherwise there is multiple ones, as shown in \cref{fig:checkerboard-tiling}.
	Naturally, offset and square sizes remain unconstrained; the claim follows.
\end{proof}
We emphasize that the tileset $\mathcal T$ in \cref{prop:unconstrained-tiling} does have valid tilings that are e.g.\ all-grey- or all-black-edged areas, or those where only a purple diagonal with grey on one side, and black on the other side is present, as shown in \cref{fig:primitive-tilings-4}.
For this reason, and in order to lift the offset degeneracy still present, we add extra constraints to the tileset.

\newcommand\TTur{\TT(1,1,1,1,T) }
\newcommand\TTll{\TT(0,0,0,0,T) }
\newcommand{\CBPRa}{
\T(R,B,*R,*B)
\T(0,*B,1,B)
\T(0,B,1,*B)
\T(0,*B,1,B)
\T(0,B,1,*B)
\T(0,*B,1,B)
}\newcommand{\CBPRb}{
\T(*R,0,R,1)
\TT(1,1,D,0,T)
\TTur\TTur\TTur\TTur
}\newcommand{\CBPRc}{
\T(R,0,*R,1)
\TT(D,D,0,0,T)
\TT(1,1,D,D,T)
\TTur\TTur\TTur
}\newcommand{\CBPRd}{
\T(*R,0,R,1)
\TTll
\TT(D,D,0,0,T)
\TT(1,1,D,D,T)
\TTur\TTur
}\newcommand{\CBPRe}{
\T(R,0,*R,1)
\TTll\TTll
\TT(D,D,0,0,T)
\TT(1,1,D,D,T)
\TTur
}\newcommand{\CBPRf}{
\T(*R,0,R,1)
\TTll\TTll\TTll
\TT(D,D,0,0,T)
\TT(1,1,0,D,T)
}\newcommand{\CBPRg}{
\T(R,B,*R,*B)
\T(0,*B,1,B)
\T(0,B,1,*B)
\T(0,*B,1,B)
\T(0,B,1,*B)
\T(0,*B,1,B)
} 
\begin{figure}[!b]
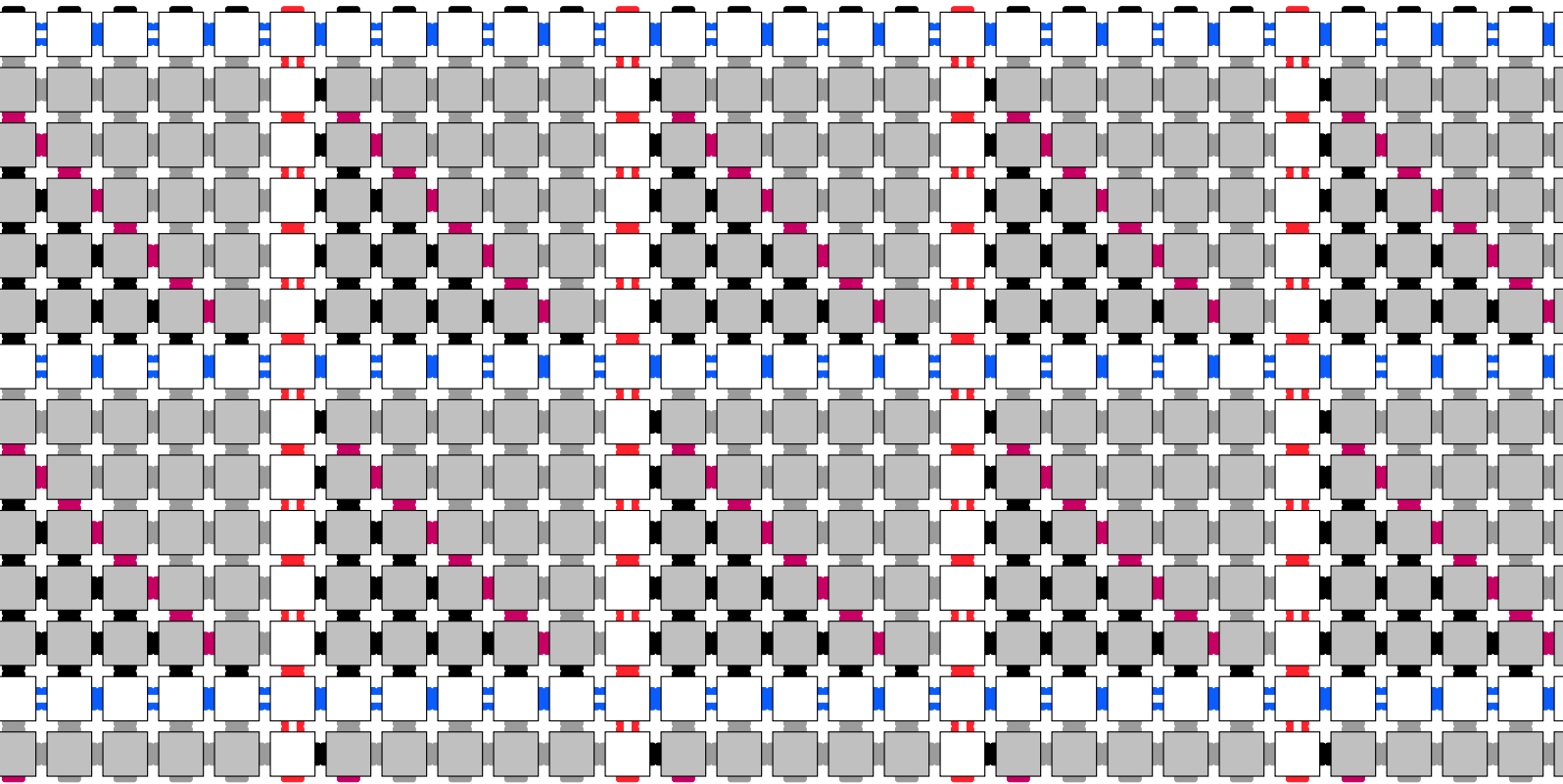

\hspace{-5.2cm}
\begin{tiles}
\CBPRa\CBPRa\CBPRa\CBPRa\CBPRa
\$
\CBPRb\CBPRb\CBPRb\CBPRb\CBPRb
\$
\CBPRc\CBPRc\CBPRc\CBPRc\CBPRc
\$
\CBPRd\CBPRd\CBPRd\CBPRd\CBPRd
\$
\CBPRe\CBPRe\CBPRe\CBPRe\CBPRe
\$
\CBPRf\CBPRf\CBPRf\CBPRf\CBPRf
\$
\CBPRg\CBPRg\CBPRg\CBPRg\CBPRg
\$
\CBPRb\CBPRb\CBPRb\CBPRb\CBPRb
\$
\CBPRc\CBPRc\CBPRc\CBPRc\CBPRc
\$
\CBPRd\CBPRd\CBPRd\CBPRd\CBPRd
\$
\CBPRe\CBPRe\CBPRe\CBPRe\CBPRe
\$
\CBPRf\CBPRf\CBPRf\CBPRf\CBPRf
\$
\CBPRg\CBPRg\CBPRg\CBPRg\CBPRg
\$
\CBPRb\CBPRb\CBPRb\CBPRb\CBPRb
\end{tiles}
\caption{Section of the checkerboard tiling Hamiltonian's ground state.}
\label{fig:checkerboard-tiling}
\end{figure}

In order to single out all those patterns that commence with a full square in the lower left corner of the lattice region, we employ \citeauthor{Gottesman-Irani}'s boundary trick which exploits the fact that on any hyperlattice there is always a very specific mismatch between the number of vertices and the number of edges.
In our case it reads as follows.

\begin{proposition}[Constrained Checkerboard Tiling]\label{prop:G-I}
	Take the tileset $\mathcal T$ from \cref{prop:unconstrained-tiling} with the same edge-matching tiling rules, and define a new tileset $\mathcal T'$ with the following additional bonuses and penalties:
	\begin{enumerate}
		\item any interior tile gets a bonus of $-1$ if it appears to the top, and a bonus of $-1$ if it appears to the right of another tile, and
		\item any interior tile gets an unconditional penalty of $2$.
	\end{enumerate}
	Then the highest score tilings possible with $\mathcal T'$ on a square lattice $\Lambda$ are the checkerboard patterns shown in \cref{fig:checkerboard-tiling}, but such that a corner tile lies in the lower left of the lattice.
    All other tilings have net score $\ge 1$.
\end{proposition}
\begin{proof}
	The only effect of the extra bonus and penalty terms are that the grey interior tiles can no longer appear on the left or bottom lattice boundaries; edge tiles have to be placed there.
	This, in turn, means that the only zero penalty configuration for the lower left corner is to place a corner tile there, meaning that the only net zero penalty configurations have at least one corner tile present.
	The rest of the claim then follows from \cref{prop:unconstrained-tiling}.
\end{proof}

With the tileset $\mathcal T'$ defined such that the highest net-score tilings are checkerboard patterns with unconstrained square sizes and offset $(0,0)$ from the lower left corner in the spin lattice, we can formalize the tiling Hamiltonian in the following lemma.
\begin{lemma}[Checkerboard Tiling Hamiltonian]\label{lem:Hchecker}
There exists a diagonal Hermitian operator $\op h \in \mathcal B(\field C^d \ox \field C^d)$ for $d=11$ with matrix entries in $\field Z$ as in \cref{Eq:TILING_Hamiltonian} such that the corresponding tiling Hamiltonian $\Hchecker=\sum_{i\sim j} \op h^{(i,j)}$ on a square lattice $\Lambda$ has a degenerate zero energy ground space $S_\mathrm{cb}$ spanned by checkerboard tilings as in \cref{fig:checkerboard-tiling},
of all possible square sizes, where the pattern starts with a corner tile at the origin (i.e.\ in the lower left corner of the lattice), as laid out in \cref{prop:G-I}.
Any other eigenstate not contained in this family of zero energy states has eigenvalue $\ge 1$.
\end{lemma}
\begin{proof}
	Translating the tileset $\mathcal T'$ from \cref{prop:G-I} into local terms as in \cref{Eq:TILING_Hamiltonian} via \cite{Bausch2015} yields local Hamiltonian terms $\op h \in \mathcal B(\field C^{d} \ox \field C^{d})$, where $d$ is the number of tiles in the tileset---here 11; the local terms have entries in $\field Z$ because all the weights (bonuses and penalties) in the tileset are integers.
	$\Hchecker$ will have a ground space spanned by tilings with net score $0$, which we proved in \cref{prop:G-I} to look as claimed.

	Furthermore, since all other tilings must have integer net penalty, all other tiling eigenstates will have energy $\ge 1$.
	The claim follows.
\end{proof}

\subsection{Classical Turing Machine Tiling}
It is well know that a classical TM which runs for time $N$ and uses a tape of length $N$ can be encoded in an $N\times N$ grid of tiles  \cite{Berge_Undecidability_Of_Domino_Problem}.
A brief overview of how this is done is given in the following.
We first recall that a TM is specified by a tuple $(\Sigma, Q, \delta)$ where $Q$ is the TM state, $\Sigma$ is the TM alphabet, and $\delta$ is a transition function
\begin{equation}
\delta: Q \times \Sigma \rightarrow Q \times \Sigma \times \{L,R\}.
\end{equation}
as well as an initial state $q_0$, an accepting state $q_a$, and a blank symbol $\# \in \Sigma$.
Here $L,R$ in the transition function output tell the TM head whether to move left or right respectively.

We now take the $N\times N$ grid which we can place tiles on. We will identify the rows of the grid with the tape of the TM, where successive rows will be successive time steps. Each tile now represents a cell of the TM's tape at a given time step. We now introduce a set of tiles which encode the evolution of the TM. We will need tiles which represents every possible configuration that a cell can take (what is written in the cell, whether the TM head is there, etc.).

To encode the evolution of such a TM into a set of tiles, we introduce three types of tiles: variety 1 which is specified only by an element of $\Sigma$, variety 2 specified by $\Sigma \times Q \times \{r,l \}$, and variety 3 specified by $\Sigma \times Q \times \{R,L\}$.
At position $P$ offset from the left within a row, these tiles have the following function:
\begin{enumerate}
	\item[Variety 1] With marking $(c)$, $c\in\Sigma$, the corresponding cell on the TM's tape contains $c$, and the TM head is \emph{not} at position $P$ at the corresponding time step.
	\item[Variety 2] With marking $(c,q, d)$, $c\in \Sigma$, $q\in Q$, $d\in\{r,l\}$, the corresponding cell on the TM's tape contains $c$, the TM head is at position $P$ at this time step, but has not yet overwritten the tape symbol. The TM is in state $q$ and the TM head has just moved from the right/left of $P$.
	\item[Variety 3] With entry $(c,q, D)$, $c\in \Sigma$, $q\in Q$, $D\in\{R,L\}$, the corresponding cell on the TM's tape contains $c$, the TM head has just moved right/left from position $P$ where it has just overwritten the previous symbol. The TM is in state $q$ at this time step.
\end{enumerate}

As a last remark, we note that one can always dovetail multiple TM tilings, as shown by \citeauthor{Gottesman-Irani}.
\begin{lemma}[Tiling-Layer Dovetailing \cite{Gottesman-Irani}] \label{Lemma:Tiling-Layer_Dovetailing}
	Let $\mathcal M_1$ and $\mathcal M_2$ be classical Turing machines with the same alphabet $\Sigma$ such that their evolution is encoded in a tiling pattern on different tiling layers (see \cref{rem:tiling-layers}) of a rectangular grid with a border as in \cref{fig:checkerboard-tiling}. Then---by potentially altering the tile sets---it is possible to constrain the tiling layers at the border such that $\mathcal M_2$ takes the output of $\mathcal M_1$ as its input and continues the computation.
\end{lemma}
\begin{proof}
If $\mathcal M_1$ and $\mathcal M_2$ are TMs, then there exists a TM $M$ which carries out $\mathcal M_1$ followed by $\mathcal M_2$ \cite{turing_1937}.
    Define a tileset on each layer that corresponds to said Turing machine, such that 
    $\mathcal M_1$ runs from bottom-to-top and $\mathcal M_2$ runs top-to-bottom on each respective layer.
	We now need to show that there is a way of enforcing equality of the tapes of the two tiling layers next to the boundary; then the claim follows.

    Similar to \cref{rem:tiling-layers}, let the meta tile at position $k$ be specified by a 2-tuple $T_k=(t_i,t_j)_k$.
	Let the set of tiles making up the border be $B$.
	Then we enforce the 2-local tiling rule that the only valid tiles that can appear next to the upper border tiles have the form $((t_i,t_i), b)$, where $b\in B$ (i.e.\ the tiles must have the same markings in both layers). Thus the output of $\mathcal M_1$ is the input of $\mathcal M_2$. $\mathcal M_2$ then continues the computation on the top layer of the grid.
\end{proof}

In this fashion, any Turing machine (e.g.\ a universal one) can be encoded in a grid of tiles, which in turn can be used to define a local Hamiltonian with a ground state that corresponds to the TM's valid evolution; given a TM tiling, this can be achieved by using the tiling-to-Hamiltonian mapping already explained.
Giving due credit, we capture this mapping for TM tilings in the following lemma.
\begin{lemma}[\citeauthor{Berge_Undecidability_Of_Domino_Problem}'s Turing machine Tiling Hamiltonian \cite{Berge_Undecidability_Of_Domino_Problem}]\label{lem:TM-tiling}
	For any classical Turing machine $(\Sigma,Q,\delta)$ there exists a diagonal Hermitian operator $\op h \in \mathcal B(\field C^d \ox \field C^d)$ for $d=\poly(|\Sigma|,|Q|)$ with matrix entries in $\field Z$ as in \cref{Eq:TILING_Hamiltonian} such that the corresponding tiling Hamiltonian $\sum_{i\sim j} \op h^{(i,j)}$ on a square lattice $\Lambda$ has a degenerate ground space $S_\mathrm{TM,tiling}$ containing
	\begin{enumerate}
		\item any tape configuration without TM head tiling the plane forward indefinitely,
		\item a tiling pattern corresponding to valid Turing machine evolutions where the initial head is aligned on one side of the lattice and where the TM does not halt on the initial tape and space provided, and
		\item any valid Turing machine evolution starting mid-way that does not halt within the space provided.
	\end{enumerate}
\end{lemma}
\begin{proof}
	See \cite{Berge_Undecidability_Of_Domino_Problem};
	the fact that the tape without head tiles the infinite plane is obvious since  the tape can be initialized arbitrarily and will consistently cover the lattice, i.e.\ by being copied forwar.
	If the TM's head is present in a tile, and since there is no transition \emph{into} the initial state $q_i\in Q$ of the TM, if the initial state is present it has to reside on one side of the lattice.
	Similarly, if the TM halts within the space provided there is no forward transition, meaning that tiling cannot have zero energy.
	Finally, if neither initial nor final state are present the tiling can show a consisten Turing machine evolution starting mid-way, with the tape being copied forward, or potentially altered if the TM head passes by.
\end{proof}
We emphasize that the TM in \cref{lem:TM-tiling} does not have to be reversible.
We will later lift the large degeneracy of the so-defined ground space by forcing an initial tape and head configuration; with such an initial setup, the tiling becomes unique since the forward evolution of a TM head and tape is always unambiguous.

\subsection{Combining Checkerboard and Turing Machine Tiling}
As seen in \cref{rem:tiling-layers,Lemma:Tiling-Layer_Dovetailing}, we can combine two tilesets into one, by defining the new tileset as the Cartesian product of the two.
In this fashion we couple the TM tile set to apper above grey-shaded interior of the squares in the underlying checkerboard pattern from \cref{prop:G-I};
the area above the edge tiles we fill with a dummy border tile.
We use this dummy border to enforce initialization of the TM's tape and head: for a tape cell \emph{above but not to the right} of the border, the tape cell is blank.
For a cell \emph{above and to the right} of a border, we put the TM into its initial configuration $q_0$.

In case we need our Turing machine to run for more steps than are available on a single $L\times L$ grid, we can do so as well by introducing multiple layers as per \cref{rem:tiling-layers}.
\begin{lemma}\label{rem:constant-overhead-TM}
	Let $n_1,n_2\in\field N$ be constant, and take a TM tileset $\mathcal S$ such that the TM tiles appear over the grey-shaded interior of the checkerboard pattern in \cref{fig:checkerboard-tiling}.
	We can define a new tileset $\mathcal S'$ such that the TM head will start in the lower left corner on an empty tape;
	on a grey square of side length $L$ it will have a tape of length $n_1L$ and runtime $n_2L$ available.
\end{lemma}
\begin{proof}
	Initializing the head and tape on one edge of the grey square is achieved by penalizing any other tiles from appearing there, which we can do using inter-layer constraints as in \cref{rem:tiling-layers}.
	Once the TM tiling reaches one end of the grey square, we can similarly copy its state to another layer with a separate TM tileset that makes it evolve in the opposite direction.
	This shows that one can increase the available number of time steps by another constant $n_2$.
	An even simpler argument shows that on finitely many grid cells $L$ one can always increase the number of tape cells by a constant factor $n_1$, by redefining $n_1$ sets of separate symbols.
	The claims follow.
\end{proof}

\subsection{Cordonning off an Edge Subsection}\label{sec:edge-subsection}
In this section, we show that one can define a classical TM tiling that puts a single marker on an extra layer within each checkerboard square in \cref{fig:checkerboard-tiling}, namely on the lower edge, and at position $x=\lceil L^{1/c}\rceil$ for any $c\in\field N$, $c>0$.
Since we have already shown how to define a classical TM tiling to appear only within the grey shaded interior of each checkerboard sqaure (\cref{rem:tiling-layers}), how to allow constant tape and runtime overhead (\cref{rem:constant-overhead-TM}), and how to dovetail TM tilings (\cref{Lemma:Tiling-Layer_Dovetailing}), the claim is immediate from the following two lemmas.
\begin{lemma}\label{lem:compute-within-tiles}
	Let $f:\field{N} \rightarrow\field{N}$ such that $1 \leq f(N)\leq N$ be computable within time $\BigO(2^{|x|})$ where $|x|$ is the binary length of its input.
	Consider the checkerboard tiling constructed in \cref{prop:G-I} such that each square has side length $N+2$ (which is measured between corner tiles).
    Then there exists a set of tiles which has this same checkerboard pattern, but for every corner tile, except those along the bottom edge, there is a special symbol $\bullet$ at distance $f(N)$ from the left border along the top edge.
\end{lemma}
\begin{proof}
	The proof is a variant of a construction from \cite{Gottesman-Irani}.
	First, we add an extra tiling layer above the grey interior of the checkerboard tilings which translates the square's side length $N$ into binary;
	this can be done with a counter tiling, see \cite{Gottesman-Irani,Patitz2014} and \cite[sec.\ F2.3]{Bausch2017}.

	Using \cref{Lemma:Tiling-Layer_Dovetailing}, we then dovetail this output with a TM that computes the function $f(N)$
    by taking input from the previous layer.
    Since this is promised to be computable in time $\BigO(2^{|N|})=\BigO(N)$, this can be done via \cref{rem:constant-overhead-TM}.
    The output of this computation is then $f(N)$ in binary.

	Finally we run a binary-to-unary converting TM on the binary output of $f(N)$ by reversing the binary counter tiling in \cite{Gottesman-Irani,Patitz2014}; this requires $N$ steps.
	This leaves a marker at distance $f(N)$ along the square interior.
	We can then introduce a tiling rule which forces a $\bullet$ marker onto the edge above it.
    The configuration on the upper white edge of each complete square of the tiling is then
    \begin{equation}\label{eq:good-edge}
        \begin{tiles}*
        \T(R,B,*R,*B)
        \T(0,*B,1,B)
        \T(0,B,1,*B)
        \T(0,*B,1,B)
        \end{tiles}
        \ \cdots\
        \begin{tiles}*
        \T(0,*B,1,B)
        \Ta(0,*B,1,B,$\bullet$)
        \T(0,B,1,*B)
        \end{tiles}
        \ \cdots\
        \begin{tiles}*
        \T(0,B,1,*B)
        \T(0,*B,1,B)
        \T(R,B,*R,*B)
        \end{tiles}
    \end{equation}
    where the black dot $\bullet$ marks distance $f(N)$ away from the left border.
\end{proof}

With \cref{lem:compute-within-tiles} in place, all that is left is to show existence of a TM that calculates the $8$\textsuperscript{th} root of a number given in binary, and obeys the required constraint on the number of steps---i.e.\ at most linear in the square's side length $L$.
\begin{lemma}\label{lem:8th-root-subsection}
Let $c\in\field N$, $c>0$.
There exists a classical TM which, on binary input $L$, computes $\lceil L^{1/8} \rceil$ in binary, and requires at most $\BigO(\log_2^8(L))$ steps.
\end{lemma}
\begin{proof}
	It is known that taking the square root of a number has the same time complexity as multiplication (see \cite{Alt1979}).
	For a number of $\ell = \log_2 L$ digits, long multiplication has time- and space complexity $\sim \log_2^2 L$.
	Taking the $8$\textsuperscript{th} root can thus be done in $\sim\log_2^8 L$ steps by calculating
	\[
	\sqrt[8]\cdot = \sqrt{\sqrt{\sqrt{\cdot}}}
	\qedhere
	\]
\end{proof}

Now we have all the ingredients together to define the following augmented checkerboard Hamiltonian, which is in essence the checkerboard tiling Hamiltonian defined in \cref{lem:Hchecker}, but with a classical TM acting within its grey squares to place an additional marker onto the horizontal edges.
\begin{lemma}[Augmented Checkerboard Tiling Hamiltonian]\label{lem:Hchecker'}
	Let $\Hchecker$ be the Hamiltonian defined in \cref{lem:Hchecker}.
	Then we can increase the local Hilbert space dimension to accommodate for the extra tileset necessary in \cref{lem:8th-root-subsection}, and define a
	new Hamiltonian $\Hchecker'$ as per \cref{Eq:TILING_Hamiltonian} where
	\begin{enumerate}
		\item the zero energy ground state is spanned by the same checkerboard patterns as in \cref{lem:Hchecker}, but such that the horizontal edges \emph{above} a grey square carry a special marker $\bullet$ at offset $L^{1/8}$ from the left cornerstone,
		\item any other eigenstate has eigenvalue $\ge 1$.
	\end{enumerate}
\end{lemma}
\begin{proof}
	The first claim follows by \cref{lem:compute-within-tiles,lem:8th-root-subsection}, and \cref{lem:TM-tiling}.
	The second claim follows since the grey TM interiors feature unique tilings, enforced by penalties only.
\end{proof}

For later reference, we further prove the following two tiling robustness facts.
\begin{remark}[Checkerboard Tiling Robustness]\label{rem:Tiling-Robustness}
	We single out the pair of tiles
	\[
          \begin{tiles}
            \T(0,*B,1,B)
            \T(R,B,*R,*B)
          \end{tiles}
          \]
 in the tileset $\mathcal T'$ used in \cref{prop:G-I}.
 Then either
 \begin{enumerate}
 	\item the pair of tiles is part of an edge of some length $L$ as shown in the proof of \cref{lem:compute-within-tiles}---i.e. fig.\ \ref{eq:good-edge}---with a grey square of size $L\times L$ \emph{below} it, and a valid TM tiling enforcing the position of the extra edge marker $\bullet$ at position $\lceil L^{1/8} \rceil$, or
 	\item there exists a unique penalty $\ge1$ at another location in the lattice that can be associated to the tile pair.
 \end{enumerate}
\end{remark}
\begin{proof}
	Since the corner tile in the pair cannot be one on the left lattice boundary, we follow the tiling to its left; it has to be a blue edge $A_2$ pattern as in \cref{fig:primitive-tilings-4}, and necessarily end in another corner tile---if not, take the mismatching tile and resulting penalty of size $1$ as the unique associated one.

	Given the blue horizontal edge is intact, this defines a distance between the two corner tiles, $L$.
	The subsquare $L\times L$ \emph{below} this defined edge then has to be a valid checkerboard square with augmenting TM as in \cref{lem:Hchecker'}, which in turn enforces the position of the $\bullet$ marker between the two upper corner tiles at the specified offset.
	If the square is not intact---which includes it being cut off---take the closest penalty in Manhattan distance from the tile pair as the associated penalty of size $\ge 1$ (or one of the closest one in case of ambiguities).
\end{proof}

\begin{remark}[Augmented Checkerboard Tiling Robustness]\label{Remark:Single_Bullet_Marker}
	In any given ground state of the checkerboard tiling, there can be at most one $\bullet$ between two cornerstone markers; this marker is only ever present on blue horizontal edges that have a full grey interior square below them, meaning the $\bullet$ is offset at $\lceil L^{1/8} \rceil$ from its left, as in \cref{lem:8th-root-subsection}.
	Any other configuration introduces a penalty $\ge 1$.
\end{remark}
\begin{proof}
	A bullet can only appear above the appropriate marker in the classical TM.
	We design the TM such that it produces exactly one such marker and such a marker gets a penalty if it is not above the point at which the TM places it.
	Thus, if there exists more than one $\bullet$ per edge joining two cornerstones, at most one of them can be above the marker left by the classical TM, and hence the other will receive an energy penalty.

	Furthermore, since $\bullet$ can only occur at the output of a valid TM tiling, it can only occur on edges that lie above a full TM tiling.
	Since the lattice boundaries are white edges by \cref{prop:G-I}, the claim follows.
\end{proof}


\section{A 2D Marker Checkerboard} \label{Sec:Marker_Hamiltonian}
In this section we will introduce a Hamiltonian on a one-dimensional spin chain which has a fine-tuned negative energy.
More specifically, our goal in this section is to take the tiling pattern given in \cref{fig:checkerboard-tiling} used to define $\Hchecker'$ in \cref{lem:Hchecker'}, and on a separate layer add the
Marker Hamiltonian $\Hmarker$ from \cite[Thm.\ 11]{Bausch_1D_Undecidable}.
The Marker Hamiltonian, plus the classical TM encoded in the classical Hamiltonian will then be used to correct for the energy penalties caused by the quantum phase estimation being done approximately (and thus picking up some error).

\subsection{Combining the Marker Hamiltonian with the Classical Tiling}
In slight extension from the construction therein, we only allow the boundary markers $\ket
{\bd}$ to coincide with the cornerstones of the checkerboard tiling
\[
\begin{tiles}
\T(R,B,*R,*B)
\end{tiles}
\]
and condition the transition terms $\op h_1$ and $\op h_2$ from \cite[Lem.\ 2]{Bausch_1D_Undecidable} to only occur in between two cornerstones and if and only if the marker $\bullet$ is present there,\footnote{This can easily be enforced with a regular expression.} i.e.\ on the blue horizontal edge
\begin{equation}\label{inline-fig:white-edge}
\underbrace{
\begin{tiles}*
\T(R,B,*R,*B)
\T(0,*B,1,B)
\end{tiles}
\ \cdots\
\begin{tiles}*
\T(0,*B,1,B)
\Ta(0,*B,1,B,$\bullet$)
\T(0,B,1,*B)
\end{tiles}
\ \cdots\
\begin{tiles}*
\T(0,B,1,*B)
\T(0,*B,1,B)
\T(0,B,1,*B)
\T(0,*B,1,B)
\T(R,B,*R,*B)
\end{tiles}
}_{\uparrow\hspace{.72cm}L\ \text{tiles}\phantom{\hspace{.74cm}}}
\end{equation}
\newcommand{\Eedge}{E_\mathrm{edge}}
All other configurations are energetically penalised.
The negative energy contribution of one such edge---and thus by \cref{Remark:Single_Bullet_Marker} also of one square \emph{below} said edge in the checkerboard pattern---is
\begin{align}
\Eedge(L) := \lmin(\Hmarker|_{S(L)}),
\end{align}
where $S(L)$ denotes a single $\bd$-bounded segment of the original marker construction of length $L$.
The arrow $\uparrow$ denotes the position of the special symbol that indicates position $L^{1/8}$, as explained in \cref{lem:8th-root-subsection,lem:Hchecker'}.

As the ground state energy of $\Hmarker$ depends on the choice of the falloff $f$
we carefully pick this function to be able to discriminate between the halting and non-halting cases in \cref{th:HQTM-spectrum}.
In particular, we will choose $f$ such that if a universal TM halts on input $\varphi(\ivar)$, then $\min_L(\Eedge(L) + \lmin(\HUTM))<0$, if it does not halt then $\min_L (\Eedge(L) + \lmin(\HUTM))\geq 0$, where we assumed the Turing machine's tape length is $L$ as well.
For inputs $\varphi'\in[\varphi(\ivar), \varphi(\ivar) + 2^{-\ivar-\ell})$ for some $\ell\ge1$ a similar condition will be true depending on the amplitude that the output state has on halting and non-halting.

One obstacle is that the bounds on the energy contribution in \cite[Lem.\ 7]{Bausch_1D_Undecidable} is too loose for our purposes, i.e.\ it was asymptotically bounded as lying in the interval $\lmin(\op H^{(f)})\in(-2^{-f(L)}, -4^{-f(L)})$.
In the following section, we prove that the scaling of the upper bound is in fact tight.

\subsection{A Tight Marker Hamiltonian Bound}
In this section we improve on the bounds set out in \cite[Lem.\ 7]{Bausch_1D_Undecidable} for the ground state energy of the Marker Hamiltonian.
To do this, we consider the following $w \times w$ matrix:
\begin{align}
\Delta_w' = \Delta^{(w)}- \ketbra{w}.
\end{align}
We now adapt \cite[Lem.\ 7]{Bausch_1D_Undecidable} to prove a better lower bound on the lowest eigenvalue.
\begin{lemma} \label{Lemma:Tighter_Marker_Bound}
	The minimum eigenvalue of $\Delta_w'$ satisfies
	\begin{align}
	\lmin(\Delta_w') \geq -\frac{1}{2} - \frac{3}{4^w}.
	\end{align}
\end{lemma}
\begin{proof}
	Our proof is essentially the same as in \citeauthor{Bausch_1D_Undecidable} except we use a better ansatz for the lower bound on the ground state energy.
	We begin by noting that, as in the proof of \cite[Lem.\ 7]{Bausch_1D_Undecidable}, the characteristic polynomial of $\Delta_w'$ is
	\begin{equation}\label{eq:pn}
	p_w(\lambda)=-\frac{2^{-w-1}}{\sqrt{\lambda -4}}\left(
	3\sqrt\lambda (x_w(\lambda) - y_w(\lambda))
	+\sqrt{\lambda-4} (x_w(\lambda) + y_w(\lambda))
	\right)
	\end{equation}
	where
	\begin{align*}
	x_w(\lambda) &= \left(\lambda -\sqrt{\lambda -4} \sqrt\lambda-2\right)^w \\
	y_w(\lambda) &= \left(\lambda +\sqrt{\lambda -4} \sqrt\lambda-2\right)^w.
	\end{align*}
	Since it is not clear if $p_w(\lambda)=0$ has any closed form solutions in expressible in $\lambda$ directly, we instead try to bound where the solutions can be.

	First we calculate
	$
	p_w(-1/2)=(-1)^{1 + w} 2^{-w},
	$
	and thus know that $\sgn \ p_w(-1/2)=1$ for $w$ odd, and $-1$ for $w$ even.
	If we can show that $p_w(-1/2-f(w))$ has the opposite sign for some function $f(w)\geq0$, then by the intermediate value theorem we know there has to exist a root in the interval $[-1/2-f(w),-1/2]$.
	Since we are trying to prove a tighter bound than \cite[Lem.\ 7]{Bausch_1D_Undecidable}, we will assume $0\leq f(w)\leq 2^{-w}$.

	Let $p_w(-1/2-f(w)) =: A_w/B_w$, where we use the notation of \cite[Lem.\ 7]{Bausch_1D_Undecidable}:
	\begin{align*}
	B_w     & = 2^{w+1} \sqrt{f(w)+\frac{9}{2}},                                                       \\
	A_w     & = -a_{1,w}(x'_w - y'_w) - a_{2,w}(x'_w + y'_w),                                              \\
	a_{1,w} & = 3 \sqrt{f(w)+\frac{1}{2}},                                                            \\
	a_{2,w} & = \sqrt{f(w)+\frac{9}{2}},                                                              \\
	x'_w    & = \left(\sqrt{f(w)+\frac{9}{2}} \sqrt{f(w)+\frac{1}{2}}-f(w)-\frac{5}{2}\right)^{\!\!w}, \\
	y'_w    & = \left(-\sqrt{f(w)+\frac{9}{2}} \sqrt{f(w)+\frac{1}{2}}-f(w)-\frac{5}{2}\right)^{\!\!w}.
	\end{align*}
	Then $B_w$, $a_{1,w}$ and $a_{2,w}$ are real positive for all $w$. We distinguish two cases.
	\paragraph{$\mathbf w$ even.} If $w$ is even, we need to show $p_w(-1/2-1/2^w)\ge0$, which is equivalent to
	\begin{align*}
	& 0 \le \frac{A_w}{B_w}  \\
	\Longleftrightarrow\quad & 0 \le A_w=- a_{1,w}(x'_w - y'_w) - a_{2,w}(x'_w + y'_w)                          \\
	\Longleftrightarrow\quad & 0 \ge a (x'_w-y'_w) + (x'_w+y'_w)\quad\quad\text{where}\ a:=\frac{a_{1,w}}{a_{2,w}}\in[1,2] \\
	\Longleftrightarrow\quad & \frac{a-1}{a+1}y'_w\ge x'_w.
	\end{align*}
	For $w$ even, $y'_w\ge x'_w$, and furthermore we find that $x_w^{\prime1/w}/y_w^{\prime1/w}$ is monotonically decreasing (assuming that $f(w)\geq 0$ and is itself monotonically decreasing), so it suffices to find a $f(w)$ which satisfies
	\begin{align}
	\frac{a-1}{a+1}\ge\left(\frac52-\frac32 \right)^w\left(\frac52+\frac32\right)^{-w} = \frac{1}{4^w}.
	\end{align}
	Expanding out $a$ as
	\begin{align}
	a = 3\sqrt{\frac{f(w)+1/2}{f(w)+9/2}},
	\end{align}
	and substituting this into the above, we find
	\begin{align}
	f(w )\geq \frac{9}{4(4^w) - 10 + 5(4^{-w})}.
	\end{align}
	Hence we can choose $f(w) = 3/4^w$, which works for all $w\ge2$.

	\paragraph{$\mathbf w$ odd.} Now $y'_w\le x'_w$, and it suffices to show
	$$
	\frac{a-1}{a+1} y_w' \le x_w'
		$$
	which is true provided
	$$\frac{a-1}{a+1}\le 1.$$
	This also holds true for all $w\ge0$ for $f(w)=3/4^w$.
	This finishes the proof.
\end{proof}

\begin{theorem}\label{th:better-marker-bounds}
	The minimum eigenvalue of $\Delta_w'$ satisfies
	\begin{align}
	-\frac{1}{2} - \frac{3}{4^w} \leq \lmin(\Delta_w') \leq -\frac{1}{2} - \frac{1}{4^w}.
	\end{align}
\end{theorem}
\begin{proof}
	\Cref{Lemma:Tighter_Marker_Bound} gives the lower bound, and \cite[Lem.\ 8]{Bausch_1D_Undecidable} gives the upper bound.
\end{proof}

\newcommand{\Epen}[1]{E_\text{pen,#1}}
\subsection{Balancing QPE Error and True Halting Penalty}\label{sec:balance}
With this tighter bound derived in \cref{th:better-marker-bounds}, we can calculate the necessary magnitude and scaling of $\Eedge(L)$ as explained at the start of \cref{Sec:Marker_Hamiltonian} as follows.
As a first step, we notice that the clock runtime $T=T(L)$ of the QTM is bounded by \cref{eq:clock-runtime}, which holds both in the halting and non-halting case, since the clock idles after the computation is done.
That is, the clock runtime \emph{does not} depend on the input to the computation.

Let $\Epen{halt}(L)$ and $\Epen{too short}(L)$ be the ground state energies of $\HUTM(L)$ in the case where the encoded computation does not halt with high probability, and when the binary expansion of the encoded phase is too long, respectively, i.e.\ when $|\varphi'|>m$.
Then from \cref{th:HQTM-spectrum} we get:
\begin{align}\label{eq:Epen-nonhalt-toosmall}
\Epen{non-halt}(L) \ge \Epen{too short}(L) = \Omega\left[\frac{1}{T^2}\right]  \overset*\ge  \frac{K_1}{ L^2 \xi^{2L}\log^2L},
\end{align}
where we made use of  \cref{rem:minimum-L-1} at step $(*)$.
Similarly, let $\Epen{halt}(L)$ be the minimum eigenvalue when the QTM halts on input $\varphi' \in [\varphi(\ivar), \varphi(\ivar) + 2^{-\ivar-\ell})$, as given in \cref{def:qpe-encoding}.
Then again from \cref{th:HQTM-spectrum} and for sufficiently large $\ell$ we get:
\begin{align}
\Epen{halt}(L) =& \BigO\left[\left( 2^{-\ell} + \delta(L,m)\right) \frac{1}{T^2} \right]  \nonumber\\
\overset{**}=& \BigO\left[ \left( 2^{-\ell} + L^2 2^{-L^{1/4}} \right)\frac{1}{T^2} \right] \le  \frac{K_2 2^{-L^{1/4}}}{\xi^{2L}}.   \label{eq:Epen-halt}
\end{align}
where in step $(**)$ we have used the fact that $m\le L$, $c_1<4$ and $c_2\ge1$.
Both $K_1$ and $K_2$ in \cref{eq:Epen-nonhalt-toosmall,eq:Epen-halt} are positive constants, chosen sufficiently small and large to satisfy the two bounds.
How large does $\ell$ have to be---or in other words, how small does the interval around $\varphi(\ivar)$ have to be that $\varphi'$ is chosen from---for \cref{eq:Epen-halt} to hold?
\begin{equation}\label{eq:ell-bound}
2^{-\ell} \le L^2 2^{-L^{1/4}}
\quad\Leftrightarrow\quad
\ell \ge \log_2\left(L^{-2} 2^{L^{1/4}} \right).
\end{equation}
In order to discriminate between the two asymptotic history state penalties in \cref{eq:Epen-halt,eq:Epen-nonhalt-toosmall},
$\Eedge(L)$ thus has to lie asymptotically between these two bounds, i.e.\ we need
\[
\Eedge(L) = \smallO\left(\frac{1}{L^2\xi^{2L}\log^2 L} \right)
\quad\text{and}\quad
\Eedge(L) = \omega\left(\frac{1}{\xi^{2L}2^{L^{1/4}}} \right).
\]

Now we know by \cref{th:better-marker-bounds} that $\Eedge(L) \sim 4^{-f(L)}$ for some $f:\field N\longrightarrow \field N$ marker falloff, which itself has to be computable by a history state construction on the segment of length $L$.
We therefore require
\begin{alignat}{2}
\smallO\left(\frac{1}{L^2\xi^{2L}\log^2 L} \right)  &=  \frac{1}{4^{f(L)}}  &=&\  \omega\left(\frac{1}{\xi^{2L}2^{L^{1/4}}} \right), \quad\text{or}  \nonumber\\
\omega\left( L + \log L + \log\log L \right) &=  f(L)  &=&\ \smallO\left(  L + L^{1/4}  \right).
\label{eq:f-bound}
\end{alignat}

This lets us formulate the following conclusion.
\begin{corollary}\label{cor:balance}
	There exists a constant $C$ such that $f(L) = C(L + L^{1/8})$ asymptotically satisfies \cref{eq:f-bound}.
\end{corollary}

\subsection{Marker Hamiltonian with L + L\textsuperscript{1/8} Falloff}
The crucial question is: can we create a Marker Hamiltonian with a falloff exponent like $f(L) = C(L+L^{1/8})$, which would satisfy \cref{cor:balance}?
As discussed in \cite{Bausch_1D_Undecidable}, this is certainly possible for any polynomial of $L$, or even an exponential---in essence it is a question of creating another history state clock for which the runtime of the segment of length $L$ equals $f(L)$.
Herein lies the problem: while a runtime $L$ is easy---just have a superposition of a particle sweeping from one side to the other---how do we perform $L^{1/8}$ additional steps?

While there might be a clever way of doing this purely within the scope of a history state construction, we take the easy way out.%
\footnote{
We note that if this task is possible within the history state framework, then it may be possible to prove the main result of this paper for 1D.
Indeed, the 2D tiling construction is only used to allow the 1D Marker Hamiltonian to have the correct drop off.}
In \cref{sec:edge-subsection}, we discussed how we can place a special symbol on the lower edge, which by \cref{lem:8th-root-subsection} can be at distance $L^{1/8}$ from the left corner.
With this in mind and with the tighter marker Hamiltonian spectral bound from \cref{Lemma:Tighter_Marker_Bound} to define the following variant of a marker Hamiltonian:
\begin{lemma} \label{Lemma:1D_Marker_Hamiltonian}
	Let $C\in\field N$ be constant.
	Take the standard marker Hamiltonian $\Hmarker_0$ from \cite{Bausch_1D_Undecidable} defined on a local Hilbert space $\HS_0=\field C^{d'}$, where $d'$ depends on the decay function $f$ to be implemented.
	Then there exists a variant $\Hmarker$ with local Hilbert space $\HS=\HS_0 \ox \field C^2$, where $\ket\bdbullet$ is one of the basis states of the second subspace, such that $\Hmarker$ has the following additional properties:
	\begin{enumerate}
		\item $\Hmarker=\sum_i \op h_i$, with $\op h_i \in \mathcal B(\field C^d \ox \field C^d)$, and $d=\BigO(C)$.
		\item $[\op h, \ketbra{\bigstar}]=0$. \label{Lemma:1D_Marker_Hamiltonian:Commuting}
		\item If $S(r)$ is the subspace of a single $\bd$-bounded segment of length $L$, containing a single $\bdbullet$ offset at position $r$, then
		\begin{align}
		- \frac{3}{4^{f(L)}}\leq \lmin\left(\Hmarker|_{S(r)}\right) \leq - \frac{1}{4^{f(L)}},
		\end{align}
		where $f(L) = C(L + r)$.
	\end{enumerate}
\end{lemma}
\begin{proof}
	We design the marker Hamiltonian variant to perform the following procedure before stopping:
	\begin{enumerate}
		\item Sweep the length of the edge $L$,
		\item Sweep back to the $\ket{\bdbullet}$ symbol sitting at offset $r$.
		\item If the number of rounds is not yet $C$, switch to another head state and repeat, where even iterations run in reverse.
	\end{enumerate}
	Finally, employ \citeauthor{Gottesman-Irani}'s boundary trick, used as in \cite[Rem.\ 3]{Bausch_1D_Undecidable}, which exploits the mismatch in number of one- and two-local interaction terms to remove the constant $-1/2$ offset present in \cref{th:better-marker-bounds} by only adding translationally-invariant nearest neighbour terms to the Hamiltonian.
	The energy scaling then follows directly from \cref{th:better-marker-bounds}, and the dimension and
	$[\op h_M, \ketbra{\bdbullet}]=0$ follow by construction.
\end{proof}

This marker Hamiltonian we will now combine with the Hilbert space of the checkerboard Hamiltonian $\Hchecker'$ from \cref{lem:Hchecker'}, to obtain a 1D marker Hamiltonian where the location of the boundary symbols $\bd$ and offset marker $\bdbullet$ align with the checkerboard tiles as
\begin{equation}\label{eq:tile-alignment}
		\bd\ \longleftrightarrow\ \Ts(R,B,*R,*B)
		\quad\text{and}\quad
		\bdbullet\ \longleftrightarrow\ \begin{tiles}*
		\Ta(0,*B,1,B,$\bullet$)
		\end{tiles}
\end{equation}
and such that the marker Hamiltonian terms do not occur above any other but the blue edge tiles.
\begin{corollary}[1D Marker Hamiltonian]\label{cor:marker-hamiltonian-with-8th-root-falloff}
	Let $\Hchecker'$ be the checkerboard Hamiltonian from \cref{lem:Hchecker'}, with local Hilbert space $\HS_\mathrm{cb}$.
	Take $\Hmarker$ from \cref{Lemma:1D_Marker_Hamiltonian}, with local Hilbert space $\HS$, and let $C\in\field N$, $C\ge1$.
	Then there exists a marker Hamiltonian $\Hmarker_1$ with one- and two-local interactions $\op h_1\in\mathcal B(\HS')$, $\op h_2\in\mathcal B(\HS' \ox \HS')$ where $\HS' := (\HS \oplus\,\field C) \ox \HS_\mathrm{cb} $, and such that $\Hmarker_1$ has the following properties.
	\begin{enumerate}
		\item If $S(r)$ denotes the subspace of a good tiling edge segment \cref{eq:good-edge} of length $L$, where the marker $\bullet$ is offset at position $r$ from the left, then
		\[
			-\frac{3}{4^{f(L)}}\leq \lmin\left( \Hmarker_1|_{S(r)} \right) \leq -\frac{1}{4^{f(L)}},
		\]
		with $f(L) = C(L + r)$.
		\item Restricted to any other tiling subspace $S'$ which does not contain the pair of tiles
		\[
		\begin{tiles}
		\T(0,*B,1,B)
		\T(R,B,*R,*B)
		\end{tiles}
		\]
		we have $\lmin(\Hmarker_1|_{S'})\ge 0$.
	\end{enumerate}
\end{corollary}
\begin{proof}
	Let $\op h'_1$ and $\op h'_2$ denote the one- and two-local terms of $\Hmarker$, trivially extended to the larger Hilbert space $\HS\oplus\,\field C$.
	Let $\ket 0$ denote the extra basis state in $\HS \oplus\,\field C$.
	Denote with $\Pi$ a projector onto the tiling subspace spanned by the corner and blue edge tiles given in \cref{prop:unconstrained-tiling}.
	We explicitly construct the local interactions $\op h_1$ and $\op h_2$ of $\Hmarker_1$ by setting
	\begin{align*}
	\op h_1 :=&\  \op h'_1\ox \Pi  +  \ketbra 0 \ox \Pi + (\1 - \ketbra 0) \ox \Pi^{\perp} + \\
	&\hspace{2mm} (\1 - \ketbra\bdbullet) \ox \ketbra{\begin{tiles}*
		\Ta(0,*B,1,B,$\bullet$)
		\end{tiles}
	}
	+ (\1 - \ketbra\bd) \ox \ketbra{\Ts(R,B,*R,*B)}
	\intertext{and}
	\op h_2 :=&\ \op h'_2 \ox \Pi^{\ox 2}.
	\end{align*}
	The marker bonus is only ever picked up by the (final state) marker head running into the right boundary in a configuration $\ket{\cdots\bl\bl\hd\raisebox{.05ex}{\bd}}$, which by the one-local Hamiltonian constraints newly imposed can only occur above the tile pair blue edge--corner given; any other configuration will have a net penalty $\ge0$.
	By construction, the ground space of $\Hmarker_1$ features the required alignment from \cref{eq:tile-alignment}.
	The claim then follows from \cref{Lemma:1D_Marker_Hamiltonian}.
\end{proof}

This is the last ingredient we require to formulate a two-dimensional variant of the Marker Hamiltonian, with the required falloff from \cref{cor:balance}.
\begin{theorem}[2D Marker Hamiltonian]\label{th:marker-ham}
	We denote with $\Lambda$ the given lattice.
	Let $\op h_1$ and $\op h_2$ be the local terms defining the 1D marker Hamiltonian from \cref{cor:marker-hamiltonian-with-8th-root-falloff} with constant $C\in\field N$, $C\ge1$.
	Further let $\Hchecker'$ be the augmented checkerboard lattice with symbol $\bullet$ offset by $L^{1/8}$ on each of the horizontal edges, as defined in \cref{lem:Hchecker'}.
	On the joint Hilbert space we set
	\[
	\HmarkerN := \1 \ox \Hchecker' + \sum_{i\in\Lambda}\op h_1^{(i)} + \sum_{i\in\Lambda} \op h_2^{(i)}
	\]
	where the second sum runs over any grid index where the $2\times1$-sized interaction can be placed.
	Then the following hold:
	\begin{enumerate}
		\item $\HmarkerN$ block-decomposes as $\HmarkerN = \bigoplus_{s=1}^L\HmarkerN_s \oplus \op B$;
		the family $\HmarkerN_s$ corresponds to all those tiling patterns compatible with the augmented checkerboard pattern in \cref{lem:Hchecker'} with square size $s$.
		$\op B$ collects all other tiling configurations.

		\item The ground state of $\HmarkerN_s$, labelled $\ket{\psi_s}$ is product across squares $\ket{\psi_s}=\bigotimes_i \ket{\phi_i}$, where $i$ runs over all squares in the tiling. \label{Point:2DMarker:Product_GS}

		\item $\op B \ge 0$. \label{Point:2DMarker:Bad_Tilings}

		\item Denote with $A$ a single square of the ground state $\ket{\boxplus_s}$ (i.e. a square making up the grid), denoted $\ket{\boxplus_s}_A$. \label{Lemma:2D_Marker:Square_Energy}
		Then its energy contribution to the ground state of $\HmarkerN_s$ is
		\[
		- \frac{3}{4^{C(s+s^{1/8})}}  \le  \bra{\boxplus_s}_{A} \HmarkerN(s)|_{A} \ket{\boxplus_s}_A   \le   - \frac{1}{4^{C(s+s^{1/8})}}.
		\]
		where $C$ is the constant from \cref{cor:marker-hamiltonian-with-8th-root-falloff}.
		\item Denote with $\Pi= \ketbra{\boxplus_s}_A$ the projector onto the orthogonal complement of the ground state of $\HmarkerN_s|_A$. Then
		\[
		\Pi \HmarkerN_s|_A \Pi \ge 0.
		\]
	\end{enumerate}
\end{theorem}
\begin{proof}
	We prove the claims step by step.
	\begin{enumerate}
		\item[Claim 1\ \& 2] The classical tiling Hamiltonian $\Hchecker'$ is diagonal in the computational basis.
		Furthermore, by construction, $\op h_1$ and $\op h_2$ defined in \cref{cor:marker-hamiltonian-with-8th-root-falloff} commute with the tiling terms.
		\item[Claim 3]  The bonus of $-1/2$ introduced in the marker Hamiltonian can only ever act across a pair of tiles
		\[
		\begin{tiles}
		\T(0,*B,1,B)
		\T(R,B,*R,*B)
		\end{tiles}
		\]
		Since we have proven the checkerboard tiling to be robust with respect to the occurence of this tile pair in \cref{rem:Tiling-Robustness},
		we know that the combination carries at least a penalty $\ge 1$ if it occurs in any non-checkerboard configuration; this means that any tiling in $\op B$ can never have a sub-configuration such that the marker bonus offsets penalties inflicted by the tiling constraints; $\op B\ge 0$ follows.
		\item[Claim 4] The ``good'' subspace in the fourth claim we know by \cref{Remark:Single_Bullet_Marker} to necessarily look as the blue edge segment \cref{eq:good-edge}.
		This, in turn, means that $r = \lceil L^{1/8} \rceil$ in \cref{Lemma:1D_Marker_Hamiltonian}, and the claim follows from the first energy bound proven therein.
		\item[Claim 5] Follows in a similar fashion as the fourth claim, from \cref{Remark:Single_Bullet_Marker} and from the second claim in \cref{Lemma:1D_Marker_Hamiltonian}.\qedhere
	\end{enumerate}
\end{proof}


\section{Spectral Gap Undecidability of a Continuous~Family~of~Hamiltonians} \label{Sec:Undecidable_Continuous_Family}
In this section, we combine the 2D Marker Hamiltonian with the QPE History State construction.
Despite the two-dimensional marker Hamiltonian, the setup is very reminiscent of the 1D construction; the crucial difference being the more finely-geared bonus and penalties we need to analyse.
We show that the energy contributions from each of the checkerboard squares, defined in the classical Hamiltonian, is either positive or negative depending on whether an encoded computation halts or not.
This provides a constant ground state energy density which is either positive or negative which can be leveraged to prove undecidability of the spectral gap and phase.

\subsection{Uncomputability of the Ground State Energy Density}
\newcommand{\PiEdge}{\Pi_\mathrm{edge}}
\newcommand{\PiCorner}{\Pi_\mathrm{corner}}
\newcommand{\HSm}{\HS_\mathrm{m}}
\newcommand{\HSq}{\HS_\mathrm{q}}
\newcommand{\sups}[1][tot]{^\mathrm{#1}}
\begin{lemma}\label{lem:undec-tech1}
	Let $\op h_1,\op h_2\sups[row],\op h_2\sups[col]$ be the one- and two-local terms of $\HmarkerN$ with local Hilbert space $\HSm$, and similarly denote with $\op q_1,\op q_2$ be the one- and two-local terms of $\HUTM$ from \cref{def:UTM-Ham} with local Hilbert space $\HSq$, respectively.
	Let $\PiEdge$ be a projector onto the edge tiles in \cref{prop:unconstrained-tiling}.
	Define the combined Hilbert space $\HS := \HSm \otimes (\HSq \oplus \field C)$, where $\ket 0$ denotes the basis state for the extension of $\HSq$.

	We define the following one- and two-local interactions:
	\begin{align*}
		\op h_1\sups &:= \op h_1 \ox \1 + \PiEdge \ox \op q_1 + \PiEdge \ox \ketbra0 +  \PiEdge^\perp \ox (1 - \ketbra 0) \\
		\op h_2\sups[tot,row] &:= \op h_2\sups[row]  \ox \1  +  \PiEdge^{\ox 2} \ox \op q_2 \\
		\op h_2\sups[tot,col] &:= \op h_2\sups[col] \ox \1 \\
		\op p_2\sups[tot,row] &:= \left[ \ketbra{\Ts(R,B,*R,*B)} \ox \1 \right]  \ox \left[ \1 \ox \ketbra*{\leftend} \right] + \\
		&\hspace{5.05mm} \left[ \1 \ox \ketbra{\Ts(R,B,*R,*B)}  \right]  \ox \left[ \ketbra*{\rightend} \ox \1 \right]
	\end{align*}
	On a lattice $\Lambda$ define the overall Hamiltonian
	\[
		\op H:=\sum_{i\in\Lambda} \op h_{1,(i)}\sups + \sum_{i\in\Lambda} \left(\op h_{2,(i)}\sups[tot,row] + \op p_{2,(i)}\sups[tot,row]  \right) + \sum_{i\in\Lambda} \op h_{2,(i)}\sups[tot,col],
	\]
	where each sum index runs over the lattice $\Lambda$ where the corresponding Hamiltonian term can be placed.
	Then $\op H$ has the following properties:
	\begin{enumerate}
		\item $\op H = \bigoplus_s \op H_s \oplus \op B'$ block-decomposes as $\HmarkerN$ in \cref{th:marker-ham}, where $\op B'=\op B\ox\1$.
		\item $\op B'\ge 0$.
		\item All eigenstates of\, $\op H_s$ are product states across squares in the tiling with square size $s$, product across rows within each square, and product across the local Hilbert space $\HSm\otimes (\HSq \oplus \field C)$.
		\item Within a single square $A$ of side length $s$ within a block $\op H_s$, all eigenstates are of the form $\ket{\boxplus_s}|_A \ox  \ket{r_0} \ox \ket r$, where
		\begin{enumerate}
			\item $\ket{\boxplus_s}$ is the ground state of the 2D marker Hamiltonian block $\HmarkerN_s$,
			\item $\ket{r_0}$ is an eigenstate of $\HUTM\oplus\mathbf 0$, i.e.\ the history state Hamiltonian with local padded Hilbert space $\HSq \oplus \field C$, and
			\item $\ket{r} \in (\HSq\oplus\field C)^{\ox (s\times(s-1))}$ defines the state elsewhere.
		\end{enumerate}
		\item The ground state of $\op H_s|_A$ is unique and given by$\ket r=\ket 0^{\ox(s\times(s-1))}$ and $\ket{r_0} = \ket\Psi$, where
		\[
			\ket\Psi = \sum_{t=0}^{T-1} \ket t \ket{\psi_t}
		\]
		is the history state of $\HUTM$ as per \cref{Theorem:QTM_in_local_Hamiltonian}, and such that $\ket{\psi_0}$ is correctly initialized.
	\end{enumerate}
\end{lemma}
\begin{proof}
	We already have all the machinery in place to swiftly prove this lemma.
	First note that, by construction, all of $\{\op h_1\sups, \op h_2\sups[tot,row], \op h_2\sups[tot,col], \op h_2\sups[col], \op p_2\sups[tot,row]\}$ pairwise commute with the respective tiling Hamiltonian terms $\{ \op h_1,\op h_2\sups[row], \op h_2\sups[col]\}$.
	Furthermore, the local terms from $\HUTM$---$\op q_1$ and $\op q_2$---are positive semi-definite; together with \cref{th:marker-ham} this proves the first three claims.
	As shown in \cref{Theorem:QTM_in_local_Hamiltonian} and since the Hamiltonian constraints in $\op p_2\sups[tot,row]$ enforce the ground state of the top row within the square $A$ to be bracketed, the first and third claim imply the fourth and fifth.
\end{proof}

\begin{lemma}\label{lem:undec-tech2}
	Take the same setup as in \cref{lem:undec-tech1}, and let $\HUTM=\HUTM(\varphi')$ for $\varphi' \in [\varphi(\ivar), \varphi(\ivar) + 2^{-\ivar-\ell})$, where $\varphi(\ivar)$ is the unary encoding of $\ivar\in\field N$ from \cref{def:qpe-encoding}.
	As usual $\ell\ge1$. Then for a block $\op H_s$ we have
	\begin{enumerate}
	\item If $s<\ivar$, $\op H_s\ge 0$.
	\item If $s\ge \ivar$ and $\mathcal M$ does not halt on input $\ivar$ within space $s$, then $\op H_s\ge 0$.
	\item If $s\ge \ivar$ and $\mathcal M$ halting on input $\ivar$, and $\ell \ge \log_2(s^{-2}2^{s^{1/4}})$ as per \cref{eq:ell-bound}, then $\lmin(\op H_s) < 0$.
	\end{enumerate}
\end{lemma}
\begin{proof}
We start with the first claim.
By \cref{lem:undec-tech1}, it suffices to analyse a single square $A$ of side length $s$; the proof then essentially follows that of \cite[Thm.\ 20]{Bausch_1D_Undecidable}.
We first assume $s<\ivar$.
Using the same notation as in \cref{th:marker-ham}, and denoting with $\Pi_\mathrm{edge}$ the projector onto the white horizontal edge within $A$, we have
\begin{align*}
    \lmin( \op H_s|_A )  &= \lmin\left[ \HmarkerN(s)|_A \otimes \1  +  \Pi_\mathrm{edge} \otimes \HUTM(\varphi') \right] \\
    &= \Eedge(s) + \Epen{tooshort}(s) \geq 0,
\end{align*}
where we used \cref{cor:balance,th:marker-ham} and the fact that the two Hamiltonian terms in the sum commute.

The other claims follow equivalently: in each case by \cref{cor:balance}, the sum of the edge bonus and TM penalties satisfy \cref{eq:f-bound}.
For the second claim, by the same process we thus get
\[
\lmin( \op H_s|_A )   = \Eedge(s) + \Epen{non-halt}(s) \geq 0.
\]
Then for the third claim,
\[
\lmin( \op H_s|_A )   = \Eedge(s) + \Epen{halt}(s) < 0.
\qedhere
\]
\end{proof}

\begin{corollary}  \label{Corollary:GSE_of_Lattice}
	Take the same setup as in \cref{lem:undec-tech2}, and let $\varphi(\ivar)$ encode a halting instance.
    Set $w=\argmin_s\{ \lmin({\op H}_s)< 0 \}$, and $W$ a single square of size $w\times w$.
	Then the ground state energy of $\op H(\varphi')$ on a grid $\Lambda$ of size $L\times H$ is bounded as
	\begin{equation}\label{eq:gs-energy}
	\lmin(\op H(\varphi'))=\left\lfloor \frac{L}{w} \right\rfloor\left\lfloor \frac{H}{w} \right\rfloor \lmin(\op H(\varphi')|_W).
	\end{equation}
\end{corollary}
\begin{proof}
	From \cref{lem:undec-tech1}, we know the ground state of $\op H(\varphi')$ is a grid with offset $(0,0)$ from the lattice's origin in the lower left.
    Each square of the grid contributes energy $\lmin(\op H(\varphi')|_W)<0$;
    the prefactor in \cref{eq:gs-energy} is simply the number of complete squares within the lattice.

    For all truncated squares on the right hand side, $\HUTM$ from \cref{def:UTM-Ham} with either the left or right ends truncated has zero ground state energy, since it is either free of the in- or output penalty terms.
    Furthermore, we see that if we truncate the right end of the 1D Marker Hamiltonian $\Hmarker_1$ in \cref{Lemma:1D_Marker_Hamiltonian}, it has a zero energy ground state since it never encounters the tile pair
    \[
    \begin{tiles}
    \T(0,*B,1,B)
    \T(R,B,*R,*B)
    \end{tiles}
    \]
    from \cref{th:marker-ham} necessary for a bonus.
    Truncating squares at the top does not yield any positive or negative energy contribution.
    The total lattice energy is therefore simply the number of complete squares on the lattice, multiplied by the energy contribution of each square.
\end{proof}

\begin{theorem}[Undecidability of Ground State Energy Density]
	Discriminating between a negative or nonnegative ground state energy density of $\op H(\varphi')$ is undecidable.
\end{theorem}
\begin{proof}
	Immediate from \cref{lem:undec-tech2,Corollary:GSE_of_Lattice}; the energy of a single square is either a small negative constant, or nonnegative.
	Determining which is at least as hard as solving the halting problem.
\end{proof}

With this result we can almost lift the undecidability of ground state energy density to the spectral gap problem.
In order to make the result slightly stronger, for this we first shift the energy of $\op H(\varphi')$ by a constant.
\begin{lemma}[\makebox{\cite[Lem.\ 23]{Bausch_1D_Undecidable}}]\label{lem:shift-ham}
	By adding at most two-local identity terms, we can shift the energy of $\op H$ from \cref{lem:undec-tech1} such that
	\[
	\lmin(\op H) \begin{cases}
	\ge 1 & \text{in the non-halting case, and} \\
	\longrightarrow-\infty & \text{otherwise.}
	\end{cases}
	\]
\end{lemma}

\subsection{Undecidability of the Spectral Gap}
With the proven uncomputability of the ground state energy density, we can lift the result using the usual ingredients---a Hamiltonian with a trivial ground state, as well as a dense spectrum Hamiltonian that will be pulled down alongside the spectrum of the QPE Hamiltonian, if the encoded universal Turing machine halts on the input encoded in the phase parameter---to prove that the existence of a spectral gap for our constructed one-parameter family of Hamiltonians is undecidable as well.

\begin{theorem}[Undecidability of the Spectral Gap]\label{th:main-2}
    For a continuous-parameter family of Hamiltonians, discriminating between gapped with trivial ground state $\ket{0}^{\ox \Lambda}$, and gapless as defined in \cref{def:gapped,def:gapless}, is undecidable.
\end{theorem}
\begin{proof}
So far we have constructed a Hamiltonian $\op H(\varphi')$ with undecidable ground state energy asymptotics given in \cref{lem:shift-ham}; we denote its Hilbert space with $\HS_1$.
We add the usual Hamiltonian ingredients as in \cite{Cubitt_PG_Wolf_Undecidability} or \cite[Thm.\ 25]{Bausch_1D_Undecidable}:
\begin{description}
\item[$\Hdens$] Asymptotically dense spectrum in $[0,\infty)$ on Hilbert space $\HS_2$.
\item[$\Htriv$] Diagonal in the computational basis, with a single $0$ energy product ground state $\ket{0}^{\ox \Lambda}$, and a spectral gap of $1$ (i.e.\ all other eigenstates have nonnegative energy $\ge 0$); its Hilbert space we denote with $\HS_3$.
\item[$\Hguard$] A 2-local Ising type interaction on $\HS:=\HS_1\otimes\HS_2\oplus\HS_3$ defined as
\[
    \Hguard:=\sum_{i\sim j}\left(\1_{1,2}^{(i)}\otimes\1_3^{(j)} + \1_3^{(i)}\otimes\1_{1,2}^{(j)}\right),
\]
\noindent
where the summation runs over all neighbouring spin sites of the underlying lattice $\Lambda$ (horizontal and vertical).
\end{description}
We then define
\[
\Hundec[L](\varphi') := \op H(\varphi')\otimes\1_2 \oplus \op 0_3 + \1_1\otimes\Hdens\oplus \op 0_3 + 0_{1,2} \oplus \Htriv + \Hguard.
\]
The guard Hamiltonian ensures that any state with overlap both with $\HS_1\otimes\HS_2$ and $\HS_3$ will incur a penalty $\ge 1$.
It is then straightforward to check that the spectrum of $\op H_\mathrm{tot}$ is given by
\[
 \spec(\Hundec) = \{0\} \cup (\spec(\op H(\varphi')) + \spec(\Hdens)) \cup G
 \]
 for some $G \subset [1,\infty)$, where the single zero energy eigenstate stems from $\Htriv$.

In case that $\lmin(\op H(\varphi'))\geq 1$, $\spec(\op H(\varphi')) + \spec(\Hdens) \subset [1,\infty)$ and hence the ground state of $\Hundec$ is the ground state of $\Htriv$ with a spectral gap of size one.

For $\lmin(\op H(\varphi')) \longrightarrow -\infty$,  $\Hdens$ is asymptotically gapless and dense; this means that $\Hundec$ becomes asymptotically gapless as well.
\end{proof}

Since the spectral properties of $\op H(\varphi')$ are---by \cref{lem:undec-tech2}---robust to a choice of $\varphi'$ within an interval around an encoded instance $\varphi(\ivar)$ as per \cref{def:qpe-encoding}---i.e.\ for large enough $\ell$ we can vary $\varphi'\in[\varphi(\ivar), \varphi(\ivar) + 2^{-\ivar-\ell})$---\cref{th:main-2} immediately proves theorem \textcolor{red}{2.1}, corollary \textcolor{red}{2.2} and corollary \textcolor{red}{2.3} in the main article.

As a small addition, we generalise the above to show undecidability between any two phases with an arbitrary ground state property.
\begin{corollary}
	Let there exist Hamiltonians $\op H_{X}$ and $\op H_{\neg X}$ defined by local terms that their ground states respectively have the property $X$ and do not have the property $X$, for all $L\geq L_0$ for some $L_0$, and both have zero ground state energy.
	Without loss of generality, let the ground state of $\op H(\varphi')$ have the property $X$.
	For a continuous-parameter family of Hamiltonians, discriminating between a phase with property $X$, and a phase without ground state property $X$, is undecidable.
\end{corollary}
\begin{proof}
	We then define
	\[
	\Hundec[L](\varphi') := \op H(\varphi')\otimes\1_2 \oplus \op 0_3 + \1_1\otimes\op H_{X}\oplus \op 0_3 + 0_{1,2} \oplus \op H_{\neg X} + \Hguard.
	\]
	If $\lambda_0(\op H(\varphi'))\rightarrow -\infty$, the ground state of $\Hundec[L](\varphi')$ has property $X$.
	If $\lambda_0(\op H(\varphi'))\geq 1$, the ground state does not have property $X$.
	Since the ground state energy of $\op H(\varphi')$ is undecidable, then determining whether the ground state has property $X$ is also undecidable.
\end{proof}

\subsection{Subtleties Concerning Computable and Uncomputable Numbers} \label{Sec:Uncomputable_Numbers}

In the statement of the main result and throughout the analysis, we have been careful to avoid a subtle point related to allowing $\varphi$ to vary over all of $\mathbb{R}$: a number chosen uniformly from $\mathbb{R}$ is almost surely uncomputable.
For such uncomputable values of $\varphi$, we cannot even write out the matrix elements of the local terms of the Hamiltonian.
Thus determining the phase of the Hamiltonian is trivially uncomputable for such values.

However, it is worth noting that this is not an issue unique to the Hamiltonian constructed in this work; \emph{any} Hamiltonian which is a function of parameters in $\mathbb{R}$ will suffer from this.
For example, take the 1D Ising model with Hamiltonian $\op H_\mathrm{TIM}(\varphi) = \sum_{\langle i,j\rangle} \sigma_z^{(i)}\sigma_z^{(j)} +\varphi \sum_i \sigma_x^{(i)}$ for $\varphi\in \mathbb{R}$.
The matrix elements for the local terms become uncomputable when $\varphi$ is uncomputable.
Despite this, the phase diagram for the model is computable \cite{Sachdev_2011}, and for all computable numbers $\varphi$ the phase of $\op H_\mathrm{TIM}(\varphi)$ can be determined.

Our results are non-trivial because they hold even for \emph{computable} values of $\varphi\in\mathbb{R}$.


\section{Conclusions}
One of the main aims of this work was as a first foray into the rigorous study of the computability and computational complexity of phase transitions.
Quantum phase transitions are one of the best studied physical phenomena, but still incompletely understood.
We anticipate that this work can be extended in several directions:
\newline\newline\noindent
\textbf{Uncomputablity in 1D.} Here we have only studied phase diagrams in 2D.
As described in \ref{Sec:Marker_Hamiltonian}, our construction has relied on the fact we can encode a classical Turing Machine into 2D tilings.
This is not possible in 1D.
However, since 1D systems tend to be fundamentally easier to solve than 2D systems, it may still be the case that the phase diagram of a 1D system is computable.
However, given the undecidability of the spectral gap in 1D \cite{Bausch_1D_Undecidable}, it would not be unexpected that computing phase diagrams of 1D systems can also be shown to be uncomputable.
\newline\newline\noindent
\textbf{More Realistic Systems.} 
	This work has shown that determining phase diagrams is in general uncomputable. But does this remain the case for any physically realistic systems---for example those with smaller Hilbert space dimension or interactions limited to a certain form?
\newline\newline\noindent
\textbf{Finite Systems.} In this work we have only studied phase diagrams in the thermodynamic limit.
Yet for any finite-sized system, determining any property is necessarily decidable (as we can simply diagonalise the Hamiltonian).
A natural question is thus what we can say about the complexity of determining phases and phase parameters for finite system sizes, for a suitable notion of phase transition in this context.

\section*{Acknowledgements}
J.\,B.\ would like to thank the Draper's Research Fellowship at Pembroke College.
J.\,D.\,W.\ is supported by the EPSRC
Centre for Doctoral Training in Delivering Quantum Technologies.
T.\,S.\,C.\ is supported by the Royal Society.
This work was supported by the EPSRC Prosperity Partnership in Quantum Software for
Simulation and Modelling (EP/S005021/1).


\section{Standard Form Hamiltonians}

We begin with the following definition for a 1D chain of spins:
\begin{definition}[Standard Basis States, from Section 4.1 of \cite{Cubitt_PG_Wolf_Undecidability}]
	Let the single site Hilbert space be $\mathcal{H}=\otimes_i \mathcal{H}_i$ and fix some orthonormal basis for the single site Hilbert space. Then a \emph{Standard Basis State} for $\mathcal{H}^{\otimes L}$ are product states over the single site basis.
\end{definition}
\noindent We now define standard-form Hamiltonians -- extending the definition from \cite{Cubitt_PG_Wolf_Undecidability}:
\begin{definition}[Standard-form Hamiltonian, from \cite{Watson_Hamiltonian_Analysis}, extended from \cite{Cubitt_PG_Wolf_Undecidability}]
	\label{Def:Standard-form_H}
	We say that a Hamiltonian $H = H_{trans} + H_{pen} + H_{in} + H_{out}$ acting on a Hilbert space $\mathcal{H} = (\mathds{C}^C\otimes\mathds{C}^Q)^{\otimes L} = (\mathds{C}^C)^{\otimes L}\otimes(\mathds{C}^Q)^{\otimes L} =: \mathcal{H}_C\otimes\mathcal{H}_Q$ is of standard form if $H_{trans,pen, in, out} = \sum_{i=1}^{L-1} h_{trans,pen, in, out}^{(i,i+1)}$, and $h_{trans,pen, in, out}$ satisfy the following conditions:
	\begin{enumerate}
		\item $h_{trans} \in \mathcal{B}\left((\mathds{C}^C\otimes\mathds{C}^Q)^{\otimes 2}\right)$ is a sum of transition rule terms, where all the transition rules act diagonally on $\mathds{C}^C\otimes\mathds{C}^C$ in the following sense. Given standard basis states $a,b,c,d\in\mathds{C}^C$, exactly one of the following holds:
		\begin{itemize}
			\item there is no transition from $ab$ to $cd$ at all; or
			\item $a,b,c,d\in\mathds{C}^C$ and there exists a unitary $U_{abcd}$ acting on $\mathds{C}^Q\otimes\mathds{C}^Q$ together with an orthonormal basis $\{\ket{\psi_{abcd}^i}\}_i$ for $\mathds{C}^Q\otimes\mathds{C}^Q$, both depending only on $a,b,c,d$, such that the transition rules from $ab$ to $cd$ appearing in $h_{trans}$ are exactly $\ket{ab}\ket{\psi^i_{abcd}}\rightarrow \ket{cd}U_{abcd}\ket{\psi^i_{abcd}}$ for all $i$. There is then a corresponding term in the Hamiltonian of the form
			$(\ket{cd}\otimes U_{abcd} - \ket{ab})(\bra{cd}\otimes U_{abcd}^\dagger - \bra{ab})$.
		\end{itemize}
		\label{standard-form_H:transition_terms}
		\item $h_{pen} \in \mathcal{B}\left((\mathds{C}^C\otimes\mathds{C}^Q)^{\otimes 2}\right)$ is a sum of penalty terms which act non-trivially only on $(\mathds{C}^C)^{\otimes 2}$ and are diagonal in the standard basis, such that $h_{pen} = \sum_{(ab) \ Illegal  } \ket{ab}_C\bra{ab}\otimes \mathds{1}_{Q}$, where $(ab)$ are members of a disallowed/illegal subspace.
		\label{standard-form_H:penalty_terms}
		\item $h_{in}=\sum_{ab} \ket{ab}\bra{ab}_C \otimes \Pi_{ab}$, where $\ket{ab}\bra{ab}_C \in (\mathds{C}^C)^{\otimes2}$ is a projector onto $(\mathds{C}^C)^{\otimes 2} $ basis states, and $\Pi_{ab}^{(in)} \in (\mathds{C}^{Q})^{\otimes 2}$ are orthogonal projectors onto $(\mathds{C}^{Q})^{\otimes 2}$ basis states.
		\item $h_{out}= \ket{xy}\bra{xy}_C \otimes \Pi_{xy}$, where $\ket{xy}\bra{xy}_C \in (\mathds{C}^C)^{\otimes2}$ is a projector onto $(\mathds{C}^C)^{\otimes 2} $ basis states, and $\Pi_{xy}^{(in)} \in (\mathds{C}^{Q})^{\otimes 2}$ are orthogonal projectors onto $(\mathds{C}^{Q})^{\otimes 2}$ basis states.
	\end{enumerate}
\end{definition}

We note that although $h_{out}$ and $h_{in}$ have essentially the same form, they will play a conceptually different role.

\begin{lemma}\label{Lemma:HTM_Standard_Form}
$\HTM$ is a standard form Hamiltonian.
\end{lemma}
\begin{proof}
Comparing with \cref{Def:Standard-form_H}, we see that all terms fall into one of the four classifications, and hence it is standard form.
\end{proof}

We now introduce the following definition.
\begin{definition}[Legal and Illegal Pairs and States, from \cite{Cubitt_PG_Wolf_Undecidability}] \label{Def:Legal_and_Illegal}
	The pair $ab$ is an \emph{illegal pair} if the penalty term $\ket{ab}\bra{ab}_C\otimes \mathds{1}_Q$ is in the support of the $H_{pen}$ component of the Hamiltonian. If a pair is not illegal, it is legal. We call a standard basis state \emph{legal} if it
	does not contain any illegal pairs, and illegal otherwise.
\end{definition}

Then the following is a straightforward extension of Lemma 42 of \cite{Cubitt_PG_Wolf_Undecidability} with $H_{in}$ and $H_{out}$ terms included.
\begin{lemma}[Invariant subspaces, extended from Lemma 42 of \cite{Cubitt_PG_Wolf_Undecidability}]
	\label{Lemma:Invariant_subspaces}

	Let $H_{trans}$, $H_{pen}$, $H_{in}$ and $H_{out}$ define a standard-form Hamiltonian as defined in \cref{Def:Standard-form_H}. Let $\mathcal{S}=\{S_i\}$ be a partition of the standard basis states of $\mathcal{H}_C$ into minimal subsets $S_i$ that are closed under the transition rules (where a transition rule $\ket{ab}_{CD}\ket{\psi} \rightarrow \ket{cd}_{CD}U_{abcd}\ket{\psi}$ acts on $\mathcal{H}_C$ by restriction to $(\mathds{C}^C)^{\otimes 2}$, i.e. it acts as $ab \rightarrow cd$).
	Then $\mathcal{H} = (\bigoplus_S \mathcal{K}_{S_i})\otimes\mathcal{H}_Q$ decomposes into invariant subspaces $\mathcal{K}_{S_i}\otimes\mathcal{H}_Q$ of $H = H_{pen} + H_{trans} + H_{in} + H_{out}$ where $\mathcal{K}_{S_i}$ is spanned by $S_i$.
\end{lemma}

\begin{lemma}[Clairvoyance Lemma, extended from Lemma 43 of \cite{Cubitt_PG_Wolf_Undecidability}]\label{Lemma:Clairvoyance}
	Let $H = H_{trans} + H_{pen} +H_{in} + H_{out}$ be a standard-form Hamiltonian, as defined in \cref{Def:Standard-form_H}, and let $\mathcal{K}_S$ be defined as in Lemma \ref{Lemma:Invariant_subspaces}. Let $\lambda_0(\mathcal{K}_S)$ denote the minimum eigenvalue of the restriction $H\vert_{\mathcal{K}_S\otimes \mathcal{H}_Q}$ of $H = H_{trans} + H_{pen}+ H_{in} + H_{out}$ to the invariant subspace $\mathcal{K}_S \otimes\mathcal{H}_Q$.

	Assume that there exists a subset $\mathcal{W}$ of standard basis states for $\mathcal{H}_C$ with the following properties:
	\begin{enumerate}
		\item All legal standard basis states for $\mathcal{H}_C$ are contained in $\mathcal{W}$. \label{clairvoyance:legal}
		\item $\mathcal{W}$ is closed with respect to the transition rules.  \label{clairvoyance:closure}
		\item  At most one transition rule applies in each direction to any state in $\mathcal{W}$. Furthermore, there exists an ordering on the states in each $S$ such that the forwards transition (if it exists) is from $\ket{t} \rightarrow \ket{t+1}$ and the backwards transition (if it exists) is $\ket{t}\rightarrow \ket{t-1}$.
		\label{clairvoyance:single_transition}
		\item For any subset $S\subseteq\mathcal{W}$ that contains only legal states, there exists at least one state to which no backwards transition applies and one state to which no forwards transition applies.
		Furthermore, the unitaries associated with the transition $\ket{t} \rightarrow \ket{t+1}$ are $U_t=\mathds{1}_Q$, for $0\leq t\leq T_{init}-1$ and $T_{init}<T$, and that the final state $\ket{T}$ is detectable by a 2-local projector acting only on nearest neighbour qudits.
		\label{clairvoyance:initial_state}
	\end{enumerate}

	\noindent Then each subspace $\mathcal{K}_S$ falls into one of the following categories:
	\begin{enumerate} 
		\item $S$ contains only illegal states, and $H\vert_{\mathcal{K}_S\otimes \mathcal{H}_Q} \geq \mathds{1}$. \label{Clv_Lemma:_Illegal}

		\item $S$ contains both legal and illegal states, and
		\begin{equation}
		W^\dagger H\vert_{\mathcal{K}_S\otimes \mathcal{H}_Q}W \geq \bigoplus_i \big(\Delta^{(\abs{S})} + \sum_{\ket{k}\in K_i}\ket{k}\bra{k} \big)
		\end{equation}
		where  $\sum_{\ket{k}\in K_i}\ket{k}\bra{k} := H_{pen}\vert_{\mathcal{K}_S\otimes \mathcal{H}_Q}$ and $K_i$ is some non-empty set of basis states and $W$ is some unitary.  \label{Clv_Lemma:Evolve_to_Illegal}

		\item 	 $S$ contains only legal states, then there exists a unitary $R=W(\mathds{1}_C\otimes(X\oplus Y)_Q)$ that puts $H\vert_{\mathcal{K}_S\otimes \mathcal{H}_Q}$ in the form \label{Clv_Lemma:Legal}
		\begin{align}
		R^\dagger H\vert_{\mathcal{K}_S\otimes \mathcal{H}_Q} R = \begin{pmatrix}
		H_{aa} & H_{ab} \\
		H_{ab}^\dagger & H_{bb}
		\end{pmatrix},
		\end{align}
		where, defining $G:=\supp\bigg( \sum_{t=0}^{T_{init}-1} \Pi_t^{(in)} \bigg)$ and $s:= \dim G$,
		\begin{itemize}
			\item $X:G \rightarrow G$.
			\item $Y:G^c \rightarrow G^c$.
			\item $H_{aa}$ is an $s\times s$ matrix.
			\item  $H_{aa}, H_{bb} \geq 0$ and are rank $r_a, r_b$ respectively.
			\item  $H_{aa}$ has the form
			\begin{align}
			H_{aa} = \bigoplus_i \big( \Delta^{(|S|)} + \alpha_i\ket{|S|-1}\bra{|S|-1} \big) + \sum_{t=0}^{T_{init-1}} \ket{t}\bra{t}\otimes X^\dagger \Pi_t|_G X.
			\end{align}
			\item  $H_{bb}$ is a tridiagonal, stoquastic matrix of the form
			\begin{equation}
			H_{bb} = \bigoplus_i(\Delta^{(\abs{S})} + \beta_i \ket{|S|-1}\bra{|S|-1} ).
			\end{equation}
			\item  $H_{ab} = H_{ba}$ is a real, negative diagonal matrix with rank $\min\{r_a, r_b\}$.
			\begin{equation}
			H_{ab} = H_{ba} = \bigoplus_i \gamma_i\ket{|S|-1}\bra{|S|-1}.
			\end{equation}
		\end{itemize}
		where either we get pairings between the blocks such that
		\begin{small}
			\begin{align}
			\begin{pmatrix}
			\alpha_i & \gamma_i \\
			\gamma_i & \beta_i
			\end{pmatrix} =
			\begin{pmatrix}
			1-\mu_i & -\sqrt{\mu_i(1-\mu_i)} \\
			-\sqrt{\mu_i(1-\mu_i)} & \mu_i
			\end{pmatrix}\quad  or 	\quad
			\begin{pmatrix}
			1 & 0 \\
			0 & 1
			\end{pmatrix},
			\end{align}
		\end{small}
		for $0\leq\mu_i\leq 1$, or we get unpaired values of $\alpha_i=0,1$ or $\beta_i=0,1$ for which we have no associated value of $\gamma_i$.
	\end{enumerate}
\end{lemma}

\end{document}